\catcode`\@=11
\newif\if@fewtab\@fewtabtrue
{\count255=\time\divide\count255 by 60
\xdef\hourmin{\number\count255}
\multiply\count255 by-60\advance\count255 by\time
\xdef\hourmin{\hourmin:\ifnum\count255<10 0\fi\the\count255}}
\def\ps@draft{\let\@mkboth\@gobbletwo
     \def\@oddfoot{\hbox to 7 cm{\tiny \versionno
        \hfil}\hskip -7cm\hfil\rm\thepage \hfil {\tiny\draftdate}}
     \def\@oddhead{}
     \def\@evenhead{}\let\@evenfoot\@oddfoot}
\def\draftdate{\number\month/\number\day/\number\year\ \ \ \hourmin }

\def\citen#1{\if@filesw \immediate\write \@auxout {\string\citation{#1}}\fi%
\@tempcntb\m@ne \let\@h@ld\relax \def\@citea{}%
\@for \@citeb:=#1\do {\@ifundefined {b@\@citeb}%
     {\@h@ld\@citea\@tempcntb\m@ne{\bf ?}%
     \@warning {Citation `\@citeb ' on page \thepage \space undefined}}%
     {\@tempcnta\@tempcntb \advance\@tempcnta\@ne
     \setbox\z@\hbox\bgroup\ifcat0\csname b@\@citeb \endcsname \relax
     \egroup \@tempcntb\number\csname b@\@citeb \endcsname \relax
     \else \egroup \@tempcntb\m@ne \fi \ifnum\@tempcnta=\@tempcntb
     \ifx\@h@ld\relax \edef \@h@ld{\@citea\csname b@\@citeb\endcsname}%
     \else \edef\@h@ld{\hbox{--}\penalty\@highpenalty
     \csname b@\@citeb\endcsname}\fi
     \else \@h@ld\@citea\csname b@\@citeb \endcsname \let\@h@ld\relax \fi}%
\def\@citea{,\penalty\@highpenalty\hskip.13em plus.13em minus.13em}}\@h@ld}
\def\@citex[#1]#2{\@cite{\citen{#2}}{#1}}%
\def\@cite#1#2{\leavevmode\unskip\ifnum\lastpenalty=\z@\penalty\@highpenalty\fi%
   \ [{\multiply\@highpenalty 3 #1%
   \if@tempswa,\penalty\@highpenalty\ #2\fi}]}   %
\makeatother 
\catcode`\@=12

\def\apo           {\mbox{\sc s}}
\def\apoi          {\mbox{\sc s}^{-1}}
\def\aposm         {\mbox{\sse\sc s}}
\def\be            {\begin{equation}}
\def\bearl         {\begin{array}{l}}
\def\bearll        {\begin{array}{ll}}
\newcommand\bee[5] {\begin{eqnarray} #5 \nonumber\\[-#1.#2em]~\\[#3.#4em]~\nonumber\end{eqnarray}}

\newcommand\BG[2]  {B_{#1,#2}}

\def\bicoaa        {{\triangleright\triangleleft}}

\def\boti          {{\,\boxtimes\,}}

\def\C             {{\ensuremath{\mathcal C}}}

\def\Cb            {\ensuremath{\hspace*{.3pt}{\mathcal C}^{\mathrm{rev}}}}
\def\CbC           {{\ensuremath{\hspace*{.3pt}\Cb\hspace*{.6pt}{\boxtimes}%
                    \hspace*{2.2pt}\mathcal C}\xspace}}

\def\cft           {conformal field theory}
\def\chii          {\raisebox{.15em}{$\chi$}}
\def\cir           {\,{\circ}\,}

\newcommand\coen[1]{\int^{#1}\hspace*{-.23em}}
\newcommand\coendF[1]{\int^{#1}\hspace*{-.33em}#1^\vee\hspace*{.07em}{\boxtimes}\hspace*{.16em}#1}
\newcommand\coendFomega[1]{\int^{#1}\hspace*{-.33em}#1^\vee\hspace*{.07em}{\boxtimes}\hspace*{.19em}\omega(#1)}
\newcommand\coendFx[1]{\int^{#1}\hspace*{-.22em}#1^\vee\hspace*{.07em}{\boxtimes}\hspace*{.16em}#1}

\def\coa           {_\triangleright}

\def\Cob           {\ensuremath{\mathcal Cob}}
\def\complex       {{\ensuremath{\mathbb C}}}

\def\CopC          {{\ensuremath{\mathcal C\op\hspace*{.8pt}{\times}\hspace*{1.5pt}\mathcal C}}}

\newcommand\Cor[2] {\ensuremath{\mathrm{Cor}^{}_{#1;#2}}}
\newcommand\Coro[2]{\ensuremath{\mathrm{Cor}_{#1;#2}}}
\newcommand\Corr[3]{\ensuremath{\mathrm{Cor}^{}_{#1;#2,#3}}}
\newcommand\Corro[3]{\ensuremath{\mathrm{Cor}_{#1;#2,#3}}}
\def\D             {{\ensuremath{\mathcal D}}}

\def\dim           {\mathrm{dim}}

\def\dimk          {\mathrm{dim}_\ko}
\def\dsty          {\displaystyle }
\def\ee            {\end{equation}}

\def\eear          {\end{array}}

\def\EndC          {{\ensuremath{\mathrm{End}_\C}}}

\def\eps           {\varepsilon}

\def\eq            {\,{=}\,}

\newcommand\erf[1] {(\ref{#1})}
\def\FF            {\mathrm F^{\mathcal Z}_\C}
\def\Fo            {{\ensuremath{F_{\scriptscriptstyle\!C}}}}
\def\Fomega        {{F_\omega}}
\def\FomegaR       {{F_\omega^\rational}}
\def\FoR           {{\ensuremath{F_{\scriptscriptstyle\!C}^\rational}}}
\def\findim        {fi\-ni\-te-di\-men\-si\-o\-nal}

\def\HBimod        {{\ensuremath{H\mbox{-}\mathrm{Bimod}}}}
\def\HBimodsm      {{H\text{-Bimod}}}

\def\HMod          {{\ensuremath{H}\text{-Mod}}}
\def\HModsm        {{H\text{-Mod}}}
\def\Hom           {{\ensuremath{\mathrm{Hom}}}}
\def\HomA          {{\ensuremath{\mathrm{Hom}_A}}}
\def\HomAA         {{\ensuremath{\mathrm{Hom}_{A|A}}}}
\def\HomC          {{\ensuremath{\mathrm{Hom}_\C}}}
\def\HomCC         {{\ensuremath{\mathrm{Hom}_\CbC}}}
\def\HomH          {{\ensuremath{\mathrm{Hom}_H}}}

\def\Homk          {{\ensuremath{\mathrm{Hom}_\ko}}}
\def\Hs            {{\ensuremath{H^*}}}
\newcommand\hsp[1] {\mbox{\hspace{#1 em}}}
\def\Hss           {{H^*_{}}}

\def\I             {\ensuremath{\mathcal I}}
\def\id            {\mbox{\sl id}}
\def\Id            {\mbox{\sl Id}}
\def\idHs          {\ensuremath{\id_{{H^{\phantom:}}^{\!\!*}}}}

\def\idsm          {\mbox{\footnotesize\sl id}}
\def\idUs          {\ensuremath{\id_{U^{*_{}}_{\phantom:}}}}

\def\iF            {\imath^\circ}

\def\iK            {\imath^K}
\def\iL            {\imath^L}
\def\iN            {\,{\in}\,}
\def\J             {{\ensuremath{\mathcal J}}}
\def\ko            {{\ensuremath{\Bbbk}}}
\def\KR            {{K^\rational_{}}}
\newcommand\labl[1]{\label{#1}\ee}
\def\LR            {{L^\rational_{}}}
\newcommand\Map[2] {\ensuremath{\text{Map}_{#1;#2}}}
\def\Mapgn         {\ensuremath{\text{Map}_{g:n}}}
\def\Mapgpq        {\ensuremath{\text{Map}_{g:p,q}}}
\def\Mapgppq       {\ensuremath{\text{Map}_{g:p+q}}}
\def\Mod           {\mbox{-Mod}}

\def\nx            {\raisebox{.08em}{\rule{.44em}{.44em}}\hsp{.4}}
\newcommand\nxl[1] {\\[#1mm]}
\newcommand\Nxl[1] {\\[-1.3em]\\[#1mm]}
\def\ohr           {\reflectbox{$\rho$}}
\def\ohrad         {\ohr_\diamond}

\def\OL            {\ensuremath{\mathcal O^L}}
\def\one           {{\bf1}}

\def\op            {^{\mathrm{op}}}
\def\oti           {\,{\otimes}\,}
\def\Oti           {{\otimes}}
\def\otic          {\,{\otimes_\complex}\,}
\def\otik          {\,{\otimes_\ko}\,}

\def\pigpq         {\ensuremath{\pi_{g:p,q}^{K;F}}}
\def\pigppq        {\ensuremath{\pi_{g:p+q}^{K;F}}}

\def\QB            {\mathcal Q^K}
\def\QL            {\ensuremath{\mathcal Q^L}}
\def\qquand        {\qquad{\rm and}\qquad}

\def\rational      {{\scriptscriptstyle\mathrm{rat.}}}
\def\rep           {representation}
\def\repV          {{\ensuremath{\mathcal{R}ep(\V)}}}
\def\req           {\varkappa}
\def\RR            {\mathrm R^{\mathcal Z}_\C}

\newcommand\setulen[2]{\setlength\unitlength{.#1#2pt}}
\def\SK            {S^K}
\def\slz           {\ensuremath{\mathrm{SL}(2,\zet)}}
\def\sse           {\scriptsize }
\def\ssg           {\scriptstyle }
\def\sss           {\scriptscriptstyle }
\def\tauHH         {\tau^{}_{\!H,H}}
\def\tauHHv        {\tau^{}_{\!H,H^*_{}}}
\def\tauHvH        {\tau^{}_{\!H^*_{}\!,H}}
\def\Times         {\,{\times}\,}
\def\TK            {T^K}
\def\To            {\,{\to}\,}

\def\V             {\ensuremath{\mathscr V}}

\def\Vectk         {\ensuremath{\mathcal V\mbox{\sl ect}_\ko}}
\def\Vee           {{}^{\vee\!}}
\def\Xs            {X^{*_{}}_{\phantom:}}
\newcommand\Z[1]   {\mathcal Z(#1)}
\def\zet           {{\ensuremath{\mathbb Z}}}

\newcommand\Includepichtft[1] {{\begin{picture}(0,0)(0,0)
                    \scalebox{.38}{\includegraphics{imgs/pic_htft_#1.eps}}\end{picture}}}
\newcommand\INcludepichtft[2] {{\begin{picture}(0,0)(0,0)
                   \scalebox{.#2}{\includegraphics{imgs/pic_htft_#1.eps}}\end{picture}}}
\newcommand\includepichtftsm[1] {{\begin{picture}(0,0)(0,0)
                    \scalebox{.266}{\includegraphics{imgs/pic_htft_#1.eps}}\end{picture}}}
\newcommand\Includepichopfsm[1] {{\begin{picture}(0,0)(0,0)
                   \scalebox{.266}{\includegraphics{imgs/pic_hopf_#1.eps}}\end{picture}}}
\newcommand\eqpic[4]{\begin{eqnarray}
                    \begin{picture}(#2,#3){}\end{picture}\nonumber\\
                    \raisebox{-#3pt}{ \begin{picture}(#2,#3) #4 \end{picture} }
                    \label{#1} \\~\nonumber \end{eqnarray} }
\newcommand\Eqpic[4]{\begin{eqnarray}
                    \begin{picture}(#2,#3){}\end{picture}\nonumber\\
                    \raisebox{-#3pt}{ \begin{picture}(#2,#3) #4 \end{picture} }
                    \nonumber \\[-3pt]~\label{#1} \end{eqnarray} }

\documentclass[12pt]{article}
\usepackage{latexsym, amsmath, amsthm, amsfonts, enumerate, amssymb, xspace, xypic }
\usepackage{fancybox}
\usepackage[all]{xy}
\usepackage[mathscr]{eucal}
\usepackage{graphicx} \usepackage{rotating}
\usepackage{epstopdf,hyperref}

\newtheorem{thm}{Theorem}

\newtheorem{prop}[thm]{Proposition}

\theoremstyle{definition}
\newtheorem{rem}[thm]{Remark}

\begin{document}

\def\cir{\,{\circ}\,} 
\numberwithin{equation}{section}
\numberwithin{thm}{section}

\begin{flushright}
    {\sf ZMP-HH/13-2 }\\
    {\sf Hamburger$\;$Beitr\"age$\;$zur$\;$Mathematik$\;$Nr.$\;$468}\\[2mm]
    February 2013
\end{flushright}
\vskip 3.1em
\begin{center}
\begin{tabular}c \Large\bf FROM NON-SEMISIMPLE HOPF ALGEBRAS TO 
         \\[3mm] \Large\bf CORRELATION FUNCTIONS FOR LOGARITHMIC CFT
\end{tabular}
\end{center}\vskip 2.1em
\begin{center}
   ~J\"urgen Fuchs\,$^{\,a}$,~
   ~Christoph Schweigert\,$^{\,b}$,~
   ~Carl Stigner\,$^{\,c}$
\end{center}

\vskip 9mm

\begin{center}\it$^a$
   Teoretisk fysik, \ Karlstads Universitet\\
   Universitetsgatan 21, \ S\,--\,651\,88\, Karlstad
\end{center}
\begin{center}\it$^b$
   Organisationseinheit Mathematik, \ Universit\"at Hamburg\\
   Bereich Algebra und Zahlentheorie\\
   Bundesstra\ss e 55, \ D\,--\,20\,146\, Hamburg
\end{center}
\begin{center}\it$^c$
   Camatec Industriteknik AB\\
   Box 5134, \ S\,--\,650\,05\, Karlstad
\end{center}
\vskip 4.3em

\noindent{\sc Abstract}
\\[3pt]
We use factorizable finite tensor categories, and specifically the 
representation categories of factorizable ribbon Hopf algebras $H$, as a
laboratory for exploring bulk correlation functions in local logarithmic 
conformal field theories.
For any ribbon Hopf algebra automorphism $\omega$ of $H$ we present a 
candidate for the space of bulk fields and endow it with
a natural structure of a commutative symmetric Frobenius algebra.
We derive an expression for the corresponding bulk partition functions 
as bilinear combinations of irreducible characters; as a
crucial ingredient this involves the Cartan matrix of the category.
We also show how for any candidate bulk state space of the type we consider,
correlation functions of bulk fields for closed oriented world sheets 
of any genus can be constructed that are
invariant under the natural action of the relevant mapping class group.

  \newpage

\section{Introduction}

Understanding a quantum field theory includes in particular having a
full grasp of its correlators on various space-time manifolds, including 
the relation between correlation functions on different space-times. 
This ambitious goal has been reached for different types of
theories to a variable extent. Next to free field  theories and to topological
ones, primarily in two and three dimensions, two-dimensional rational 
conformal field theories are, arguably, best under control.

This has its origin not only in the (chiral) symmetry structures that are 
present in conformal field theory, but also in the fact that for rational CFT
these symmetry structures have particularly strong representation theoretic 
properties: they give rise to modular tensor categories, and are thus
in particular finitely semisimple. In many applications semisimplicity 
is, however, not a physical requirement. Indeed there are physically 
relevant models, like those describing percolation problems, which are not 
semisimple, but still enjoy certain finiteness properties.

It is thus natural to weaken the requirement that the representation category
of the chiral symmetries should be a modular tensor category. A natural
generalization is to consider \emph{factorizable finite ribbon 
categories} (see Remark \ref{rem:factorizable}(i) for a definition of this
class of categories). By Kazhdan-Lusztig type dualities such categories are 
closely related to categories of \findim\ modules over \findim\ complex Hopf 
algebras \cite{fgst}. For this reason, we study in this paper structures in 
representation categories of \findim\ factorizable ribbon Hopf algebras.

\medskip

Let us summarize the main results of this contribution. We concentrate on bulk 
fields. In the semisimple case, the structure of bulk partition functions has 
been clarified long ago; in particular, partition functions of automorphism
type have been identified as a significant subclass \cite{mose4}.
Here we deal with the analogue of such partition functions without imposing 
se\-mi\-simplicity. Specifically, we assume that the representation category of 
the chiral symmetries has been realized as the category of \findim\ left modules 
over a \findim\ factorizable ribbon Hopf algebra. For any ribbon Hopf algebra
automorphism $\omega$ we then construct a candidate for the space of bulk fields 
for the corresponding automorphism invariant and show that it has a natural
structure of a commutative symmetric Frobenius algebra (Theorem \ref{thmfuSs4}).
The proof that the space of bulk fields is a commutative algebra also works for 
arbitrary factorizable finite ribbon categories (Proposition \ref{prop:Calgebra}).

We are able to express the 
bulk partition functions that are associated to these spaces as bilinear 
combinations of irreducible characters -- such a decomposition can still exist
because characters behave additively under short exact sequences. We find
(Theorem \ref{thm:Fomega}) that the crucial ingredient (apart from 
$\omega$) is the Cartan matrix of the underlying category, This result is 
most gratifying, as the Cartan matrix has a natural categorical meaning and 
is stable under Morita equivalence and under Kazhdan-Lusztig correspondences 
of abelian categories. The Cartan matrix enters in particular in 
the analogue of what in the semisimple case is the charge-conjugation 
partition function. As the latter is, for theories with compatible boundary 
conditions, often called the Cardy case, we refer to its generalization 
in the non-semisimple case as the \emph{Cardy-Cartan modular invariant}.

Finally we describe how for any bulk state space of the type we consider,
correlation functions of bulk fields for closed oriented world sheets of any 
genus can be found that are invariant under the natural action of the mapping 
class group on the relevant space of chiral blocks. The construction of these 
correlators is algebraically natural, once one has realized (see Proposition 
\ref{prop:YD}) that the monodromy derived from the braiding furnishes 
a natural action of a canonical Hopf algebra object in
the category of chiral data on any representation of the chiral algebra.

\medskip

The rest of this paper is organized as follows. In Section \ref{sec:F}
we describe the bulk state space of a \cft\ from several different points 
of view. In particular, we show that a \emph{coend}, formalizing the idea
of pairing each representation with their conjugate to obtain a
valid partition function, provides a candidate for a bulk state space. 
Section \ref{sec:hndl} is devoted to the \emph{handle algebra}, a Hopf 
algebra internal to a modular tensor category that crucially enters in
Lyubashenko's construction of representations of mapping class groups.
The topic of Section \ref{sec:pf} is the torus partition function; in 
particular we highlight the fact that in the partition function for the
bulk algebras constructed in Section \ref{sec:F} the Cartan matrix of the
underlying category enters naturally. In Section \ref{sec:corfus} we 
finally extend our results to surfaces of arbitrary genus and provides
invariants of their mapping class groups that are natural candidates
for correlation functions.
 
In an appendix we provide mathematical background material. In
Appendix \ref{Acoend} we review the definition and properties of dinatural
transformations and coends. In Appendix \ref{Acenter}, we recall the notion 
of the full center of an algebra in a fusion category, which is an algebra 
in the Drinfeld double of that categry. Appendix \ref{algchar} supplies
some necessary background from the representation theory of
finite-dimensional associative algebras, including in particular the
definition of the Cartan matrix of an abelian category (which is a
crucial ingredient of the simplest torus partition function in the
non-semisimple case). In Appendices \ref{facHopf} and \ref{app:mpg} we 
collect information about the factorizable Hopf algebras and the 
representations of mapping class groups that are needed in the body of the 
paper.


\section{The bulk state space}\label{sec:F}

\subsection{Holomorphic factorization}

The central ingredient of chiral conformal field theory is a chiral symmetry 
algebra. Different mathematical notions formalizing this physical idea are
available. Any such concept of a chiral algebra \V\ must provide a
suitable notion of representation category \repV, which should have, 
at least, the structure of a \complex-linear abelian category.

For concreteness, we think about \V\ as a vertex algebra with a conformal 
structure. A vertex algebra \V\ and its representation category \repV\ allow 
one to build a system of sheaves on moduli spaces of curves with marked points, 
called conformal blocks, or chiral blocks. These sheaves are endowed with a 
(projectively flat) connection. Their monodromies thus lead to (projective)
representations of the fundamental groups of the moduli 
spaces, i.e.\ of the mapping class groups of surfaces.
This endows the category \repV\ with much additional structure. In particular, 
from the chiral blocks on the three-punctured sphere one extracts a monoidal 
structure on \repV, which formalizes the physical idea of operator product of 
(chiral) fields. Furthermore, from the monodromies one obtains natural
transformations which encode a braiding as well as a twist. In this way one keeps 
enough information to be able to recover the representations of the mapping 
class groups from the category \repV. Moreover, if \repV\ is also endowed with 
left and right dualities, the braiding and twist relate the two dualities to 
each other, and in particular they can fit together to the structure of a 
ribbon category. 
In this paper we assume that \repV\ indeed is a ribbon category; there exist 
classes of vertex algebras which do have such a representation category and 
which are relevant to families of logarithmic conformal field theories.

The category \repV\ -- or any ribbon category \C\ that is ribbon equivalent 
to it -- is called the category of chiral data, or of Moore-Seiberg 
\cite{mose3,baKir} data. For sufficiently nice chiral algebras \V\ the number 
of irreducible representations is finite and \repV\ carries the structure of 
a factorizable finite ribbon category. For the purposes of this paper we 
restrict our attention to the case that \repV\ has this structure.

\medskip

While chiral CFT is of much mathematical interest and also plays a role in 
modeling certain physical systems, like in the description of universality 
classes of quantum Hall fluids, the vast majority of physical applications of 
CFT involves full, local CFT. It is generally expected that a full CFT can be 
obtained from an underlying chiral theory by suitably ``combining'' 
holomorphic and anti-holomorphic chiral degrees of freedom or, in free 
field terminology, left- and right movers. 
Evidence for such a \emph{holomorphic factorization} comes from CFTs that 
possess a Lagrangian description (see e.g.\ \cite{witt72}). Conversely, 
the \emph{postulate} of holomorphic factorization can be phrased in an
elegant geometric way as the requirement that correlation functions
on a surface $\Sigma$ are specific sections in the chiral blocks 
associated to a double cover $\widehat\Sigma$ of $\Sigma$. These particular
sections are demanded to be invariant under the action of the mapping class 
group of $\Sigma$, and to be compatible with sewing of surfaces. 

These conditions only involve properties of the representations of mapping 
class groups that are remembered by the additional structure of the 
category \repV\ of chiral data. Accordingly, to find and characterize 
solutions to these constraints it suffices (and is even appropriate) to 
work at the level of \repV\ as an abstract factorizable ribbon category.
We refer to \cite{BAki} for details on how the vector bundles of chiral 
blocks, which form a complex analytic modular functor (in the terminology 
of \cite{BAki}) can be recovered from representation theoretic data
that correspond to a topological modular functor,
and to \cite[Sects.\,5\,\&\,6.1]{fuRs10} for a more detailed discussion of 
this relationship for the chiral blocks that appear in the simplest
correlation functions, involving few points on a sphere.

In \emph{rational} conformal field theories, for which the category \repV\ 
has the structure of a (semisimple) modular tensor category, the problem
of finding and classifying solutions to all these constraints 
has a very satisfactory solution (see e.g.\ \cite{scfr2} for a review).
In \emph{logarithmic} CFTs, on the other hand, no evidence for holomorphic 
factorization is available from a Lagrangian formulation. Instead, in this 
paper we take holomorphic factorization as a starting point. Until recently, 
only few model-independent results for logarithmic CFTs were available.
Here we will present a whole class of solutions for correlators of bulk 
fields, on orientable surfaces of any genus with any number of insertions.
It is remarkable that, once relevant expressions, like e.g.\ sums over 
isomorphism classes of simple objects, that are suggestive in the semisimple 
case have been substituted with the right categorical constructions, we find 
a whole class of solutions which work very much in the same spirit as for 
semisimple theories.

\medskip

In this contribution we first focus on the \emph{bulk state space} $F$ of the 
theory. Invoking the state-field correspondence, $F$ is also called the space of
bulk fields. A bulk field carries both holomorphic and anti-holomorphic degrees 
of freedom. When the former are expressed in terms of $\C\,{\simeq}\,\repV$, 
then for the latter one has to use the 
\emph{reverse category} \Cb, which is obtained from \C\ by inverting the 
braiding and the twist isomorphisms. As a consequence, a bulk field, and in 
particular the bulk state space $F$, is an object in the \emph{enveloping 
category} $\Cb\boxtimes \C$, the Deligne tensor product of \Cb\ and \C.

It should be appreciated that owing to the opposite braiding and twist in its 
two factors, the enveloping category is in many respects simpler than the 
category \C\ of chiral data. This is a prerequisite for the possibility of 
having correlation functions that are local and invariant under the mapping 
class group of the world sheet. If \C\ is semisimple and factorizable, then 
a mathematical manifestation of this simplicity of $\Cb \boxtimes\C$ is the 
fact that its class
in the Witt group \cite{dmno} of non-degenerate fusion categories vanishes.

The most direct way of joining objects of \Cb\ and \C\ to form bulk fields that 
comes to mind is to combine the `same' objects in each factor, which in view 
of the distinction between \Cb\ and \C\ means that any object $U$ of \C\ is to 
be combined with its (right, say) dual $U^\vee$ in \Cb. When restricting to 
\emph{simple} objects only, this idea results in the familiar expression
  \be
  F = \FoR := \bigoplus_{i\in\I}\, S_i^\vee \boxtimes S_i^{}
  \labl{Fo}
for the bulk state space of a rational CFT; here the index set $\I$ is finite 
and $(S_i)_{i\in\I}$ is a collection of representatives for the 
isomorphism classes of simple objects of \C.

In rational CFT the object \erf{Fo} is 
known as the charge conjugation bulk state space, and its character
  \be
  \chii_\FoR = \sum_{i\in\I}\, \chii_i^*\, \chii_i^{}
  \labl{chiFo}
as the \emph{charge conjugation modular invariant}.
Moreover, it can be shown \cite{fffs3} that for any modular tensor category
\C\ the function \erf{chiFo} is not only invariant under the action 
of the modular group, as befits the torus partition function of 
a CFT, but that it actually appears as part of a consistent full CFT, and 
hence \FoR\ as given by \erf{Fo} is indeed a valid bulk state space.
In contrast, for non-rational CFT \erf{Fo} is no longer a valid bulk state space
and the function \erf{chiFo} is no longer modular invariant, even if the
index set $\I$ is still finite.


\subsection{The bulk state space as a coend}

It is, however, not obvious how the formula \erf{Fo}, which
involves only simple objects, relates to the 
original idea of combining \emph{every} object of \C\ with its conjugate.
Fortunately there is a purely categorical construction by which that idea can 
be made precise, namely via the notion of a \emph{coend}. Basically, the coend 
provides the proper concept of summing over \emph{all objects} of a category, 
namely doing so in such a manner that at the same time
\emph{all relations} that exist between objects are accounted for, meaning that
all morphisms between objects are suitably divided out. 

The notion of a coend can be considered for any functor $G$ from $\C\op \Times \C$ 
to any other category $\mathcal D$. That it embraces also the morphisms of
$\mathcal D$ manifests itself in that the coend of $G$ is not just an object
$D$ of $\mathcal D$, but it also comes with a \emph{dinatural family} of 
morphisms from $G(U,U)$ to $D$ (but still one commonly refers also to the
object $D$ itself as the coend of $G$). For details about coends and dinatural 
families we refer to Appendix \ref{Acoend}. In the case at hand, $\mathcal D$
is the enveloping category \CbC, and the coend of our interest is
  \be
  \Fo := \coendF U \,\in\, \CbC \,.
  \labl{coendCC}

Whether this coend indeed exists as an object of \CbC\ depends on the
category \C, but if it exists, then it is unique.
If \C\ is cocomplete, then the coend exists.  In particular, the coend does 
indeed exist for all finite tensor categories (to be defined at the beginning 
of Subsection \ref{LCFT+frc}).  This includes all modular tensor categories, 
and thus all categories of chiral data that appear in rational CFT. Moreover, 
a modular tensor category \C\ is semisimple, so that accounting for all 
morphisms precisely amounts to disregarding any non-trivial direct sums, and 
thus to restricting the summation to simple objects. Hence for modular \C\ one gets
  \be
  \coendF U = \FoR \labl{coendFo}
with \FoR\ as given by \erf{Fo}. In short, once we realize the physical idea of 
summing over all states in the proper way that is suggested by elementary 
categorical considerations, it does explain the ansatz \erf{Fo} for the bulk 
state space of a rational CFT.

This result directly extends to all bulk state spaces that are of 
\emph{automorphism type}. Namely, for any ribbon automorphism of \C, i.e.\
any autoequivalence $\omega$ of \C\ that is compatible with its ribbon
structure, we can consider the coend
  \be
  \Fomega := \coendFomega U \,.
  \labl{Fomega}
In the rational case this gives
  \be
  \FomegaR = \bigoplus_{i\in\I}\, S_i^\vee \boxtimes \omega(S_i^{}) \,,
  \labl{ssFomega}
with associated torus partition function
  \be
  \chii_\FomegaR = \sum_{i\in\I}\, \chii_i^*\, \chii_{\bar\omega(i)}^{}
  \labl{chiFomega}
where $\bar\omega$ is the permutation of the index set $\I$ for which
$S_{\bar\omega(i)}$ is isomorphic to $ \omega(S_i)$.


\subsection{The bulk state space as a center}

In rational CFT, the object \erf{Fo} of \CbC\ can also be obtained by another 
purely categorical construction from the category \C, and one may hope that 
this again extends beyond the rational case. To explain this construction, we 
need the notions of the \emph{monoidal center} (or Drinfeld center) $\Z\C$ of
a monoidal category \C\ and of the \emph{full center} of an algebra. $\Z\C$ is
a braided monoidal category; its objects are pairs $(U,z)$ consisting of objects
and of so-called half-braidings of \C. The full center $Z(A)$ of an algebra
$A \iN \C$ is a uniquely determined commutative algebra in $\Z\C$ whose
half-braiding is in a suitable manner compatible with its multiplication.
For details about the monoidal center of a category and the full center of an
algebra see Appendix \ref{Acenter}.
If \C\ is modular, then \cite[Thm.\,7.10]{muge9} the monoidal center is 
monoidally equivalent to the enveloping category,
  \be
  \CbC \,\simeq\, \Z\C \,,
  \labl{CbC=ZC}
so that in particular $Z(A)$ is an object in \CbC.

Now any monoidal category \C\ contains a distinguished algebra object, namely 
the tensor unit $\one$, which is even a symmetric Frobenius algebra 
(with all structural morphisms being identity morphisms). We thus know that 
$Z(\one)$ is a commutative algebra in $\Z\C$. If \C\ is modular, then this 
algebra is an object in \CbC, and gives us the bulk state space \erf{Fo},
  \be
  Z(\one) = \FoR \,.
  \labl{Z(1)}

Moreover, let us assume for the moment that a rational CFT can be consistently 
formulated on any world sheet, including world sheets with boundary, with 
non-degenerate two-point functions of bulk fields on the sphere and of 
boundary fields on the disk. It is known \cite{fjfrs2} that \emph{any} bulk 
state space of such a CFT is necessarily of the form
  \be
  F = Z(A)
  \labl{FZA}
for some simple symmetric special Frobenius algebra $A$ in \C. Further,
the object $Z(A)$
decomposes into a direct sum of simple objects of \CbC\ according to \cite{fuRs4}
  \be
  F = Z(A) = \bigoplus_{i,j\in\I}\, \big( S_i^\vee \boxtimes S_j^{} {\big)}
  ^{\oplus Z_{ij}(A)} 
  \labl{F=ZA}
with the multiplicities $Z_{ij}(A)$ given by the dimensions
  \be
  Z_{ij}(A)
  := \dimk \big(\HomAA(S_i^\vee{\otimes^{+\!}}A\,{\otimes^-}S_j^{},A) \big)
  \ee
of morphisms of $A$-bimodules. Here the symbols $\otimes^\pm$ indicate the two 
natural ways of constructing induced $A$-bimodules with the help of the braiding
of \C. (In case $A$ is an Azumaya algebra, this yields \cite[Sect.\,10]{fuRs11}
the {automorphism type bulk state spaces \erf{ssFomega} for which 
$Z_{ij}(A) \eq \delta_{\bar j,\bar\omega(i)}$, see formula \erf{chiFomega}).

The result \erf{FZA} implies further that $F$ is not just an object of \CbC,
but in addition carries natural algebraic structure: 

\begin{prop}\label{prop:rffs}{\rm \cite[Prop.\,3]{rffs}}\\
For $A$ a symmetric special Frobenius algebra in a modular tensor category
\C, the full center $Z(A)$ is a commutative symmetric Frobenius algebra in \CbC.
\end{prop}
 
In field theoretic terms, the multiplication on $F$ describes the
operator product of bulk fields, while the non-degenerate pairing which 
supplies the Frobenius property reflects the non-degeneracy of 
the two-point functions of bulk fields on the sphere.


\subsection{Logarithmic CFT and finite ribbon categories}\label{LCFT+frc}

It is natural to ask whether the statements about rational CFT collected 
above, and specifically Proposition \ref{prop:rffs},
have a counterpart beyond the rational case. 

Of much interest, and particularly tractable, is the class of non-semisimple 
theories that have been 
termed \emph{logarithmic} CFTs. It appears that the categories of chiral data 
of such CFTs, while not being semisimple, still share crucial finiteness 
properties with the rational case. A relevant concept is the one of a 
\emph{finite tensor category}; this is \cite{etos} a \ko-linear
abelian rigid monoidal 
category with fi\-ni\-te-di\-men\-si\-o\-nal morphism spaces and finite set 
$\I$ of isomorphism classes of simple objects, such that each simple object 
has a projective cover and the Jordan-H\"older series of every object has 
finite length.

Unless specified otherwise, in the sequel \C\ will be assumed to be a
(strict) finite tensor category with a ribbon structure or, in short, a
\emph{finite ribbon category}.
For all such categories the coend $\Fo \eq \coendF U$ exists \cite{fuSs3}.

\begin{rem}\label{rem:squig}
Finiteness of $\I$ and existence of projective covers are, for instance, 
manifestly assumed in the conjecture \cite{qusc4,garu2} that the bulk state 
space of charge conjugation type decomposes as a left module over a single 
copy of the chiral algebra \V\ as
  \be
  \Fo \,\rightsquigarrow\,
  \bigoplus_{i\in\I} P_i^\vee \,{\otimes_\complex^{}}\, S_i \,,
  \labl{Fo_squig}
where $P_i$ the projective cover of the simple \V-module $S_i$.
\\
Note, however, that the existence of such a decomposition does by no means 
allow one to deduce the structure of \Fo\ as an object of \CbC. In particular 
there is no reason to expect that simple or projective objects of \CbC\ appear 
as direct summands of \Fo, nor that \Fo\ is a direct sum of 
`$\boxtimes$-factorizable' objects $U\boti V$ of \CbC. We do \emph{not} make 
any such assumption; our working hypothesis is solely that the bulk state space
\Fo\ for a logarithmic CFT can still be described as the coend \erf{coendCC}.
\end{rem}

\smallskip

We are now going to establish an algebra structure on the coend \erf{coendCC},
without assuming semisimplicity.
As already pointed out, the coend is not just an object, but an object
together with a dinatural family. For the bulk state space $\Fo \eq \coendF U$ 
we denote this family of morphisms by $\iF$ and its members by
  \be
  \iF_U :\quad U^\vee \boti U \to \Fo
  \ee
for $U\iN\C$. The braiding of \C\ is denoted by $c \eq (c_{U,V})$. We also need 
the canonical isomorphims that identify, for $U,V\iN\C$, the tensor product of 
the duals of $U$ and $V$ with the dual of $V\oti U$; we denote them by
  \be
  \gamma_{U,V}^{} :\quad
  U^\vee \oti V^\vee \stackrel\cong\longrightarrow (V{\otimes}U)^\vee \,.
  \labl{gammaUV}

Now we introduce a morphism $m_\Fo$ from $\Fo\oti\Fo$ to \Fo\ by setting
  \be
  m_\Fo \circ (\iF_U\oti\iF_V)
  := \iF_{V\otimes U} \circ (\gamma^{}_{U,V} \boti c^{}_{U,V})
  \labl{mFo}
for all $U,V\iN\C$. This family of morphisms from $ (U^\vee {\boxtimes} U) \oti 
(V^\vee {\boxtimes} V) \eq (U^\vee {\otimes} V^\vee)\boti(U{\otimes} V)$ to 
\Fo\ is dinatural both in $U$ and in $V$ and thereby
determines $m_\Fo$ completely, owing to the universal property of coends.

\begin{prop} \label{prop:Calgebra}~\nxl1
{\rm(i)}\, The morphism {\rm \erf{mFo}} endows the object \Fo\ with the
structure of an {\rm(}associative, unital{\rm)} algebra in \CbC.
\nxl1
{\rm(ii)}\, The multiplication $m_\Fo$ of the algebra \Fo\ is commutative.
\end{prop}

\begin{proof}
(i)\, For $U,V,W\iN\C$ we have
  \be
  \bearl
  m_\Fo \circ (\id_\Fo \oti m_\Fo) \circ (\iF_U\oti\iF_V\oti\iF_W) 
  \Nxl2 \hsp5
  = \iF_{W\otimes V\otimes U}
  \circ (\gamma^{}_{V\otimes U,W} \boti c^{}_{V\otimes U,W})
  \circ [ (\gamma^{}_{U,V} \boti c^{}_{U,V}) \oti \id_{W^\vee\boxtimes W} ]
  \qquand \Nxl4
  m_\Fo \circ (m_\Fo \oti \id_\Fo) \circ (\iF_U\oti\iF_V\oti\iF_W) 
  \Nxl2 \hsp5
  = \iF_{W\otimes V\otimes U}
  \circ (\gamma^{}_{U,W\otimes V} \boti c^{}_{U,W\otimes V})
  \circ [ \id_{U^\vee\boxtimes U} \oti (\gamma^{}_{V,W} \boti c^{}_{V,W}) ] \,.
  \eear
  \labl{mFo-ass12}
Using the braid relation
  \be
  c^{}_{V\otimes U,W} \circ (c^{}_{U,V} \oti \id_W)
  = c^{}_{U,W\otimes V} \circ (\id_U \oti c^{}_{V,W})
  \ee
and the obvious identity
  \be
  \gamma^{}_{V\otimes U,W} \circ (\gamma^{}_{U,V} \oti \id_{W^\vee})
  = \gamma^{}_{U,W\otimes V} \circ (\id_{U^\vee} \oti \gamma^{}_{V,W})
  \ee
in $\HomC(U^\vee\otimes V^\vee\otimes W^\vee,(W\otimes V\otimes U)^\vee)$, it 
follows that for any triple $U,V,W$ the two morphisms in \erf{mFo-ass12} 
coincide, and thus
  \be
  m_\Fo \circ (\id_\Fo \oti m_\Fo) = m_\Fo \circ (m_\Fo \oti \id_\Fo) \,.
  \ee
This shows associativity. Unitality is easy; the unit morphism is given by
  \be
  \eta_\Fo^{} = \iF_\one ~\in \HomCC(\one\boti\one,\Fo) \,.
  \ee
(ii)\, Denote by $c^\CbC$ the braiding in \CbC. We have
  \be
  \bearll
  m_\Fo \circ c^\CbC_{\Fo,\Fo} \circ (\iF_U\oti\iF_V) \!\!&
  = m_\Fo \circ (\iF_V\oti\iF_U) \circ (c^{-1}_{V^\vee,U^\vee} \boti c^{}_{U,V})
  \Nxl3 &
  = \iF_{U\otimes V} \circ (\gamma^{}_{V,U} \boti c^{}_{V,U})
  \circ (c^{-1}_{V^\vee,U^\vee} \boti c^{}_{U,V})
  \Nxl3 &
  = \iF_{U\otimes V} \circ [ (\gamma^{}_{V,U} \cir
  c^{-1}_{V^\vee,U^\vee}) \boti (c^{}_{V,U} \cir c^{}_{U,V}) ]
  \Nxl3 &
  = \iF_{V\otimes U} \circ [ \gamma^{}_{U,V} \boti (c_{V,U}^{-1}
  \cir c^{}_{V,U} \cir c^{}_{U,V}) ]
  \,= m_\Fo \circ (\iF_U\oti\iF_V) \,.
  \eear
  \ee
Here the crucial step is the fourth equality, in which the dinaturalness 
property of $\iF$ (together with $(c_{U,V})^{\vee} \eq c_{U^\vee,V^\vee}$) is 
used. We conclude that $m_\Fo \cir c^\CbC_{\Fo,\Fo} \eq m_\Fo$, i.e.\ the 
algebra $(\Fo,m_\Fo,\eta_\Fo)$ is commutative.
\nxl1
Let us also point out that working with a non-strict monoidal structure would make 
the formulas appearing here more lengthy, but the proof would carry over easily.
\end{proof}


\subsection{The coend as a bimodule of a factorizable Hopf algebra}\label{ssec:coendFH}

Proposition \ref{prop:Calgebra} is about as far as we can get, for now, for
general finite ribbon categories. To obtain a stronger result we
specialize to a particular subclass, consisting of categories \C\ that are
equivalent to the category \HMod\ of (\findim\ left) modules over
a \emph{factorizable Hopf algebra} $H$. Such an algebra is, in short, a \findim\
Hopf algebra $(H,m,\eta,\Delta,\eps,\apo)$ (over an algebraically closed field
\ko\ of characteristic zero) that is endowed with an R-matrix $R \iN H\otik H$ 
and a ribbon element $v \iN H$ and for which the monodromy matrix
$Q \eq R_{21}\,{\cdot}\, R$ is non-degenerate. Some more details about 
this class of algebras are supplied in Appendix \ref{facHopf}.

It is worth pointing out that categories belonging to this subclass which are 
relevant to CFT models are are well known, namely \cite{dipR,cogR,ffss} the 
semisimple representation categories of Drinfeld doubles of finite groups.
Logarithmic CFTs for which the category of chiral data is believed
to be at least very close to the type of category considered here are
the $(1,p)$ triplet models, see e.g.\ \cite{fgst,fgst4,naTs2,tswo,rugw}.

\medskip

For $H$ a \findim\ ribbon Hopf algebra, the category \HMod\ carries a natural 
structure of finite ribbon category. The monoidal structure (which again we 
tacitly take to be strictified) and dualities precisely require the algebra 
$H$ to be Hopf: the tensor product is obtained by pull-back of the $H$-action 
along the coproduct $\Delta$, the tensor unit is $\one \eq (\ko,\eps)$, and 
left and right dualities are obtained from the duality for \Vectk\ with the 
help of the antipode. The braiding $c$ on \HMod\ is given by the action of the 
R-matrix composed with the flip map $\tau$, while the twist $\theta$ is 
provided by acting with the inverse $v^{-1}$ of the ribbon element.

In a fully analogous manner one can equip the category \HBimod\ of 
finite-dimensional $H$-bi\-mo\-dules with the structure of a finite ribbon 
category as well: Pulling back both the left and right $H$-actions along 
$\Delta$ gives again a tensor product. Explicitly, the tensor product
of $H$-bi\-mo\-dules $(X,\rho_X,\ohr_X)$ and $(Y,\rho_Y,\ohr_Y)$ 
is the tensor product over \ko\ of the underlying \ko-vector spaces $X$ 
and $Y$ together with left and right actions of $H$ given by
   \be
   \bearl
   \rho_{X\otimes Y}^{} := (\rho_X \oti \rho_Y) \circ (\id_H \oti \tau_{H,X} \oti
                        \id_Y) \circ (\Delta \oti \id_X \oti \id_Y)  \qquand
   \Nxl3
   \ohr_{X\otimes Y}^{} := (\ohr_X \oti \ohr_Y) \circ (\id_X \oti \tau_{Y,H} \oti
                        \id_H) \circ (\id_X \oti \id_Y \oti \Delta) \,.
   \eear
   \labl{def-tp}
The tensor unit is the one-dimensional vector space \ko\ with both left and 
right $H$-action given by the counit, $\one_{H\text{-Bimod}} \eq (\ko,\eps,\eps)$. 
A braiding $c$ on the so obtained monoidal category is obtained by composing the
flip map with the action of the R-matrix $R$ from the right and the action of 
its inverse $R^{-1}$ from the left, and a twist $\theta$ is provided by 
  \be
  \theta_X = \rho \circ (\id_H \oti \ohr) \circ (v \oti \id_X \oti v^{-1}) \,,
  \labl{deftwist}
i.e.\ by acting with the ribbon element $v$ from the left and with its inverse
from the right.
(For a visualization of the braiding and the twist isomorphisms
in terms of the graphical calculus for the symmetric monoidal category \Vectk\
see formulas (3.3) and (4.19), respectively, of \cite{fuSs3}.)

The category \HBimod\ with this structure of ribbon category 
is of interest to us because it can be shown
to be braided equivalent to $(H{\otimes_\ko} H\op_{})$\Mod\ and thus to
the enveloping category $\CbC \eq \HMod^{\rm rev} \,{\boxtimes}\, \HMod$. 
The functors furnishing the equivalence are the identity on morphisms
and act on objects $(M,\rho^{H\Oti H\op_{}})\iN (H{\otimes_\ko} H\op_{})$\Mod,
respectively $(M,\rho^H,\ohr^{\!H})$, as \cite[Lemma\,A.4]{fuSs3}
   \eqpic{HHbimiso} {410} {36} {
   \put(0,0)   {\Includepichtft{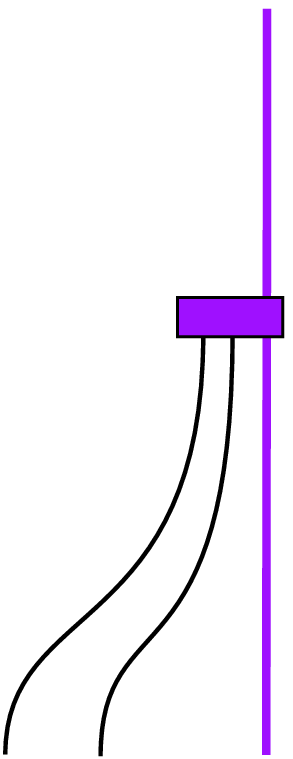}
   \put(-4,-8.5) {\sse$ H $}
   \put(6,-8)    {\sse$ H $}
   \put(24.1,-8.5) {\sse$ M $}
   \put(24.8,86) {\sse$ M $}
   \put(-5.5,48.5) {\sse$ \rho^{H\Oti H\op_{}} $}
   }
   \put(58,34)   {$ \mapsto$}
   \put(100,0) { \Includepichtft{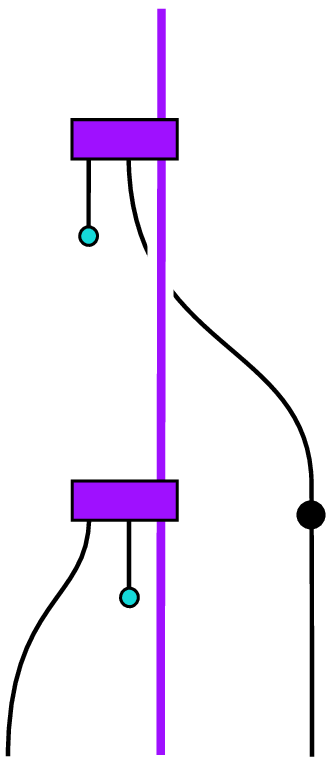}
   \put(-4,-8.5) {\sse$ H $}
   \put(12.3,-8.5) {\sse$ M $}
   \put(12.9,86) {\sse$ M $}
   \put(29,-8.5) {\sse$ H $}
   \put(36,25)   {\sse$ \apo^{-1} $}
   \put(-17.2,27.9){\sse$ \rho^{H\Oti H\op_{}} $}
   \put(-17.2,67.9){\sse$ \rho^{H\Oti H\op_{}} $}
   }
   \put(189,34)  {and}
   \put(248,0) {\Includepichtft{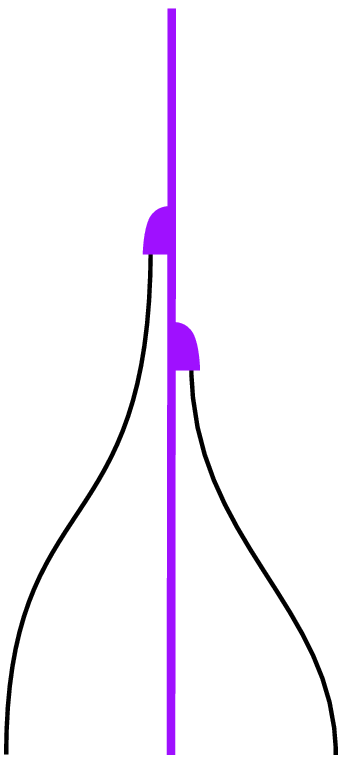}
   \put(-4,-8.5) {\sse$H $}
   \put(5.2,55.4){\sse$ \rho^H $}
   \put(13,-8.5) {\sse$ M $}
   \put(13.8,86) {\sse$ M $}
   \put(22.2,43.4){\sse$ \ohr^{\!H} $}
   \put(32,-8.5) {\sse$ H $}
   }
   \put(308,34)   {$ \mapsto$}
   \put(340,0) { \Includepichtft{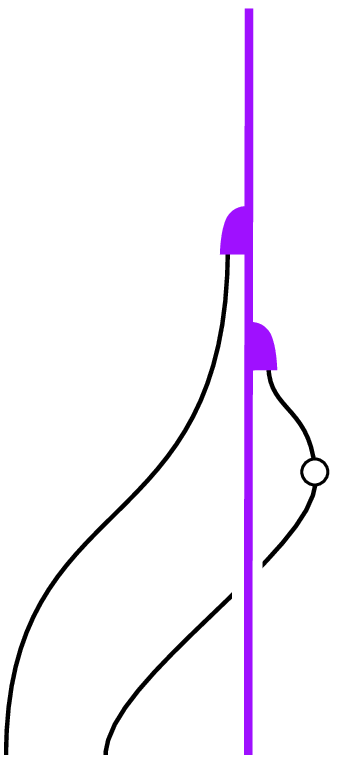}
   \put(-4,-8.5) {\sse$ H $}
   \put(7,-8.5)  {\sse$ H $}
   \put(22,-8.5) {\sse$ M $}
   \put(22.8,86) {\sse$ M $}
   \put(36,29.5) {\sse$ \apo $}
   } }
We will henceforth identify the enveloping category with \HBimod\ and
present our results in the language of $H$-bimodules. In particular we
think of the coend $\Fo \eq \coendFx U$ as an $H$-bimodule; we find

\begin{thm}\label{thm:Fo} {\rm \cite[Prop.\,A.3]{fuSs3}}\\
The coend $\coendFx U$ in the category \HBimod\ is the \emph{coregular bimodule},
that is, the vector space $\Hs\eq \Homk(H,\ko)$ dual to $H$ endowed with 
the duals of the left and right regular $H$-actions, i.e.\
   \be
   \bearl
   \rho_\Fo = (d^\ko_H\oti\idHs) \circ (\idHs\oti m\oti\idHs)
   \cir (\idHs\oti\apo\oti b^\ko_H) \circ \tauHHv \qquand
   \nxl3
   \ohr_\Fo = (d^\ko_H\oti\idHs) \circ (\idHs\oti m\oti\idHs)
   \circ (\idHs\oti\id_H\oti\tauHvH) \circ (\idHs\oti b^\ko_H \oti\apoi) \,.
   \eear
   \labl{rhorho}
together with the dinatural family of morphisms given by
  \be
  \iF_U := \big[ ( d^\ko_U \cir (\idUs \oti \rho_U) ) \oti \idHs \big] \circ
  \big[\idUs\oti ( (\tau_{U,\Hss}\oti \idHs) \cir (\id_U\oti b^\ko_\Hss) ) \big]
  \labl{iFHMod}
for any $H$-module $(U,\rho_U)$.
\end{thm}

Here $d^\ko$ and $b^\ko$ are the evaluation and coevaluation maps for the 
duality in \Vectk, respectively. Thus in particular \erf{iFHMod} describes 
$\iF_U$ in the first place only as a linear map from $U^*\otik U$ to \Hs. But it
can be checked \cite[Lemma\,A.2]{fuSs3} that if $\Hs\ \eq \Homk(H,\ko)$ is given
the structure of the coregular $H$-bimodule and $U^*\otik U$ the $H$-bimodule 
structure 
   \be
   (U^*\otik U,\rho,\ohr) :=
   \big( U^*\otik U\,,\, \rho_{U^\vee_{}}\oti\id_U\,,\, \idUs
   \oti(\rho_U\cir \tau_{U,H}^{} \cir (\id_U\oti\apo^{-1})) \big)
   \ee
that is implied by the equivalence between $\HMod^{\rm rev}\,{\boxtimes}\,\HMod$
and \HBimod, then $\iF_U$ is actually a morphism in \HBimod.

Now due to our finiteness assumptions, $H$ has an integral and cointegral. 
Denote by $\Lambda \iN H$ and $\lambda \iN \Hs$ the integral and cointegral,
respectively, normalized according to the convention \erf{Lambdalambda}.
We henceforth denote the coend $\coendFx U$ again by \Fo\ and set
  \be
  \begin{array}{ll}
  m_\Fo^{} := {\Delta^{\phantom:}}^{\!\!\!*_{}} :\quad \Fo\oti \Fo \To \Fo\,,
  \qquad\quad
  \eta_\Fo := \eps^* :\quad \one \To \Fo\,, 
  \Nxl4
  \Delta_\Fo := {[ (\mbox{\sl id}_H \,{\otimes}\, (\lambda\,{\circ}\, m))
  \cir (\mbox{\sl id}_H\,{\otimes}\,\mbox{\sc s}\,{\otimes}\,\mbox{\sl id}_H)
  \cir (\Delta\,{\otimes}\,\mbox{\sl id}_H) ]}^* :\quad \Fo \To \Fo\oti \Fo
  \qquad~{\rm and}
  \Nxl4
  \eps_\Fo := {\Lambda^{\phantom:}}^{\!\!\!*} :\quad \Fo \To \one \,.
  \end{array}
  \labl{Ffro}
Again these are introduced as linear maps between the respective underlying 
vector spaces, but are actually morphisms of $H$-bimodules, as indicated.
This way \Fo\ is endowed with a Frobenius algebra structure, as befits 
the bulk state space of a conformal field theory:

\begin{thm} {\rm \cite[Thm.\,2]{fuSs4}}\label{thmfuSs4}\\
For $H$ a factorizable Hopf algebra, the bimodule morphisms {\rm \erf{Ffro}} 
endow the coend \Fo\ with a natural structure of a commutative symmetric 
Frobenius algebra with trivial twist in the ribbon category \HBimod.
\end{thm}

We note that this assertion is in full agreement with the result of Proposition
\ref{prop:Calgebra} which holds for general finite ribbon categories. Indeed, 
by implementing the explicit form \erf{iFHMod} of the dinaturality morphisms, 
the product $m_\Fo$ that we defined in \erf{mFo} for any finite ribbon category
\C\ reproduces the expression for $m_\Fo$
in \erf{Ffro}, and likewise for $\eta_\Fo$.

It will be convenient to work with the pictorial expressions for the maps
\erf{Ffro} that are furnished by the graphical calculus for tensor categories.
They are
   \Eqpic{pic-Hb-Frobalgebra} {440} {47} {
   \put(0,45)      {$ m_\Fo ~= $}
     \put(50,0)  {\Includepichtft{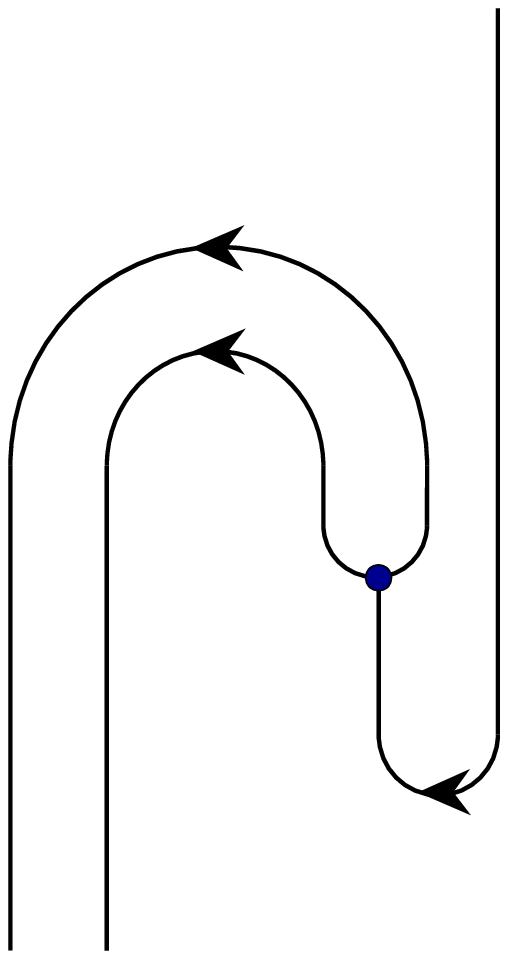}
   \put(-5.9,-8.8) {\sse$ \Hss $}
   \put(6.5,-8.8)  {\sse$ \Hss $}
   \put(31.5,34.7) {\sse$ \Delta $}
   \put(49.7,106.8){\sse$ \Hss $}
   }
   \put(144,45)    {$ \eta_\Fo ~= $}
     \put(192,24) {\Includepichtft{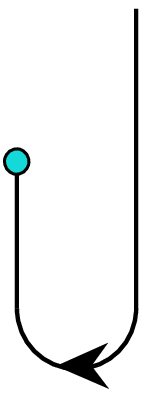}
   \put(-5.4,23.6) {\sse$ \eps $}
   \put(10.7,44.1) {\sse$ \Hss $}
   }
   \put(250,45)    {$ \Delta_\Fo ~= $}
     \put(300,0) {\Includepichtft{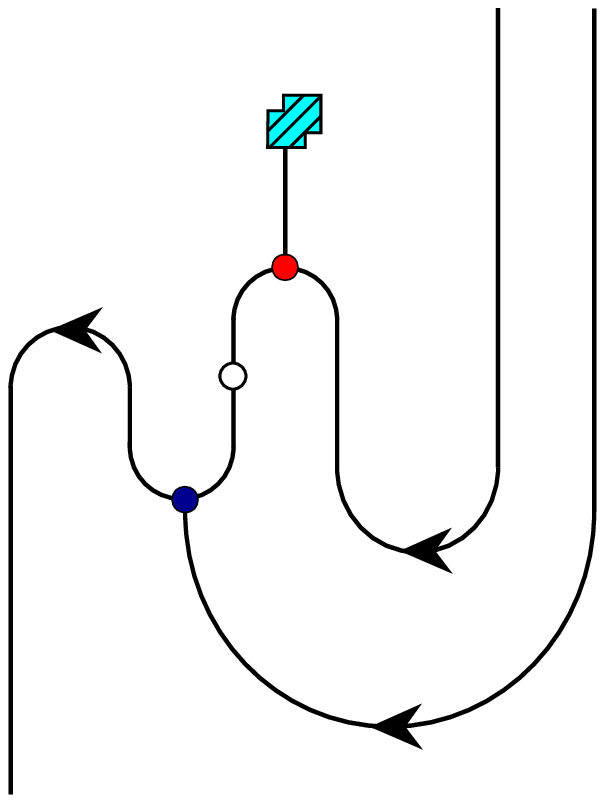}
   \put(-4.3,-8.8) {\sse$ \Hss $}
   \put(11.1,26.8) {\sse$ \Delta $}
   \put(17.7,43.8) {\sse$ \apo $}
   \put(21.5,71)   {\sse$ \lambda $}
   \put(32.8,59.3) {\sse$ m $}
   \put(48.2,89.4) {\sse$ \Hss $}
   \put(61.1,89.4) {\sse$ \Hss $}
   }
   \put(403,45)    {$ \eps_\Fo ~= $}
     \put(447,24) {\Includepichtft{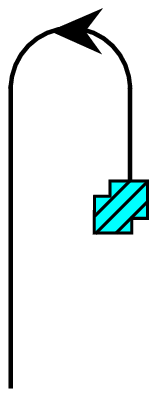}
   \put(-4.3,-8.8) {\sse$ \Hss $}
   \put(15.8,16.3) {\sse$ \Lambda $}
   } }
(Such pictures are to be read from bottom to top.)

The result just described generalizes easily from \Fo\ to the automorphism-twisted
versions $\Fomega$ as defined in \erf{Fomega}. Namely \cite[Prop.\,6.1]{fuSs3},
for any Hopf algebra automorphism $\omega$ of $H$ the $H$-bimodule that is obtained
from the coregular bimodule by twisting the right $H$-action by $\omega$, i.e.\ 
  \be
  \Fomega = \big( \Hs,\rho_\Fo^{},\,\ohr^{}_\Fo\cir(\idHs\oti\omega) \big) \,,
  \labl{def:Fomega}
carries the structure of a Frobenius algebra, which is commutative, symmetric 
and has trivial twist. The structural morphisms for the Frobenius structure are
again given by \erf{Ffro}, i.e.\ as linear maps they are identical with those 
for \Fo.

We also note that if $H$ is semisimple, then $\Fomega$ carries the structure 
of a Lagrangian algebra in the sense of \cite[Def.\,4.6]{dmno}.


\section{Handle algebras}\label{sec:hndl}

Besides in the description of the bulk state space, there is another issue in 
CFT in which one needs to perform a sum over all states for the full local 
theory, namely when one wants to specify the relation between correlators on 
world sheets that are obtained from each other by \cite{sono2}
\emph{sewing} (respectively, looking at the process from the other end,
by \emph{factorization}) as `summing over all intermediate states'.

Let us first consider this relationship for rational CFT, and at the level 
of spaces of chiral blocks. A rational CFT furnishes a modular functor, and this
functor is representable \cite[Lemma\,5.3.1]{BAki}. Accordingly the space 
$V(E)$ of chiral blocks for a Riemann surface $E$ is isomorphic to the 
morphism space $\HomC(U_E,\one)$ for a suitable object $U_E \iN \C$; in 
particular, if $E$ has genus zero and $n$ ingoing (say) chiral insertions 
$U_1,U_2,...\,,U_n$, then $U_E \,{\cong}\, U_1 \oti U_2 \oti{\cdots}\oti U_n$.
Let now the Riemann surface $E^1$ be obtained from the connected 
Riemann surface $E^\circ$ by removing two disjoint open disks $D_\pm$ and gluing 
the resulting boundary circles to each other, whereby the genus increases by 1.
Then there is an isomorphism 
  \be
  \bigoplus_{i\in\I}\, V(E^\circ_{i\bar i}) \stackrel{\cong}{\longrightarrow} V(E^1)
  \ee
between the space $V(E^1)$ and the direct sum of all
spaces $V(E^\circ_{i\bar i})$, where the surface $E^\circ_{i\bar i}$ is obtained
from $E^\circ$ by introducing chiral insertions $S_i$ and $S_{\bar i} \,{\cong}\,
S_i^\vee$, respectively, in the disks $D_\pm$ (see e.g.\ 
\cite[Def.\,5.1.13(iv)]{BAki}). In terms of morphism spaces of \C\ this amounts to
  \be
  V(E^1) \,\cong\,
  \bigoplus_{i\in\I} \HomC( S_i^\vee \oti S_i^{} \oti U_{E^\circ_{}} ,\one)
  \,\cong\, \HomC(\LR \oti U_{E^\circ_{}} ,\one) \,,
  \ee
where on the right hand side we introduced the object 
  \be
  \LR := \bigoplus_{i\in\I} S_i^\vee \oti S_i^{} ~\in \C \,.
  \labl{ssiL}
By induction, the space of chiral blocks for a genus-$g$ 
surface with ingoing field insertions $U_1$, $U_2,,...\,,U_n$ is then
  \be
  V(E) \,\cong\,
  \HomC( \LR^{\otimes g} \oti U_1 \oti U_2 \oti{\cdots}\oti U_n ,\one) \,,
  \labl{VE}
i.e.\ the object $\LR$ appears to a tensor power given by the number of handles.
Also, as we will point out soon, the object $\LR$ of \C\ carries a natural 
structure of a Hopf algebra internal to \C; it is therefore called the 
(chiral) \emph{handle Hopf algebra}.

Invoking holomorphic factorization, the correlation function $\mathrm{Cor}(\Sigma)$
for a world sheet
$\Sigma$ is an element in the space of chiral blocks for the \emph{complex double}
$\widehat\Sigma$ of the world sheet. Taking $\Sigma$ to be orientable and
with empty boundary and all field insertions on $\Sigma$ to
be the whole bulk state space $F$, one has
  \be
  \mathrm{Cor}(\Sigma) \,\in\, V(\widehat\Sigma)
  \,\cong\, \HomCC(\KR^{\otimes g} \oti F^{\otimes m},\one)
  \labl{CorSigma}
with
  \be
  \KR := \LR \boxtimes \LR = \bigoplus_{i,j\in\I}\, (S_i^\vee \boti S_j^\vee) \otimes
  (S_i^{} \boti S_j^{}) ~\in \CbC \,.
  \labl{ssiK}
$\KR$ is called the \emph{bulk handle Hopf algebra} of the rational CFT based 
on the modular tensor category \C.
Note that each of the objects $S_i^{} \boti S_j^{}$ with $i,j\iN\I$ is simple
and together they exhaust the set of simple objects of \CbC, up to isomorphism.

\medskip

It is tempting to think of the gluing of a handle to a Riemann surface as a
means for inserting a complete set of intermediate states. 
In view of the isomorphism \erf{VE} then immediately the question arises why 
it is precisely the object $\LR$ that does this job. And again the categorical 
notion of a coend turns out to provide the proper answer. Indeed, just like 
the bulk state space $\Fo\eq\FoR$, the objects $L \eq \LR \iN\C$ and $K \eq 
\KR \iN\CbC$ can be recognized as the coends of suitable functors, namely as
  \be
  L = \int^{U\in\C}\! U^\vee \oti U \qquand
  K = \int^{X\in\CbC}\!\! X^\vee \oti X \,,
  \labl{defLK}
respectively (together with corresponding families of dinatural 
transformations, whose explicit form we do not need at this point).  

Moreover, just like in the discussion of \Fo\
the description of $L$ and $K$ as the coends \erf{defLK} remains valid 
beyond the semisimple setting: These coends exist not only when \C\ is a modular
tensor category, i.e.\ for rational CFT, but also for more general categories,
and in particular for any finite tensor category \C. Note that here the 
statement for $K$ is redundant, as it is just obtained from the from the one 
for $L$ by replacing \C\ with \CbC, and \CbC\ inherits all relevant
structure and properties from \C. This applies likewise to other issues, like 
e.g.\ the Hopf algebra structure on these objects, and accordingly we will 
usually refrain from spelling them out for $L$ and $K$ separately.

As already announced, we have

\begin{thm} {\rm \cite{lyub6,kerl5}} \\
For \C\ a finite ribbon category, the coend $L\eq \int^{U\in\C} U^\vee \oti U$
carries a natural structure of a Hopf algebra in $\mathcal C$. It has an 
integral $\Lambda_L \iN \HomC(\one,L)$ and a Hopf pairing 
$\varpi_L \iN \HomC(L\oti L,\one)$.
\end{thm}

The structural morphisms of $L$ as a Hopf algebra are given by
  \be
  \bearl
  m_L \circ (\iota_U\oti \iota_V) := \iota_{V\otimes U} \circ (\gamma_{U,V}
  \oti\id_{V\otimes U}) \circ (\id_{U^\vee} \oti c_{U,V^\vee\otimes V}) \,,
  \Nxl3
  \eta_L := \iota_\one \,,
  \Nxl3
  \Delta_L \circ \iota_U := (\iota_U\oti \iota_U)
  \circ (\id_{U^\vee} \oti b_U \oti \id_U)  \,,
  \Nxl3
  \eps_L \circ \iota_U := d_U \,,
  \Nxl3
  \apo_L \circ \iota_U := (d_U \oti \iota_{U^\vee})
  \circ (\id_{U^\vee} \oti c_{U^{\vee\!\vee}_{}\!,U}
  \oti \id_{U^\vee}) \circ (b_{U^\vee} \oti c_{U^\vee\!,U}) \,,
  \eear
  \ee
and the Hopf pairing is
  \be
  \varpi_L \circ (\iota_U\oti \iota_V) := (d_U \oti d_V) \circ \big[
  \id_{U^\vee} \oti (c^{}_{V^\vee,U} \cir c^{}_{U,V^\vee} \oti \id_V) \big] \,.
  \ee
Here $d$ and $b$ are the evaluation and coevaluation morphisms for the (right)
duality in \HMod\ and $c$ is the braiding of \HMod, while $\iota$ is the 
dinatural family of the coend $L$, and the isomorphisms $\gamma_{U,V}$ are 
the ones defined in \erf{gammaUV}.

In terms of graphical calculus in \C,
  \eqpic{p7} {395}{102} {
    \put(13,105) {
  \put(0,0)   {\Includepichtft{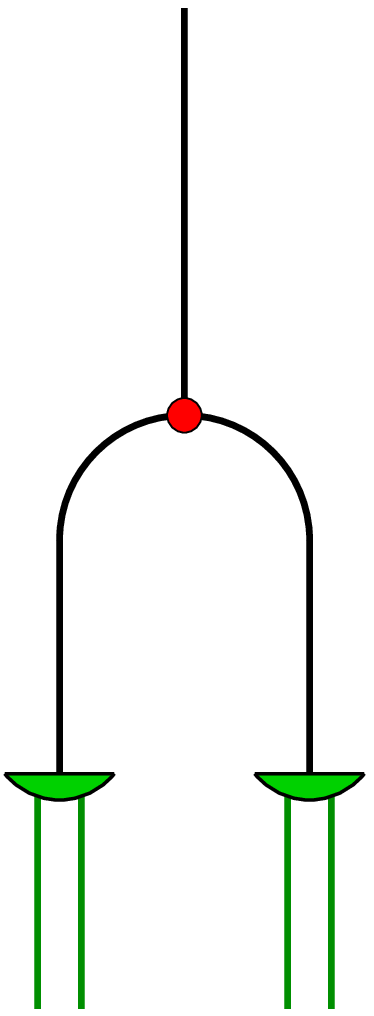}
  \put(-6,-9)   {$\ssg U^{\!\vee} $}
  \put(7,-9)    {$\ssg U $}
  \put(16.9,115){$\ssg L $}
  \put(22,69)   {$\ssg m_L $}
  \put(23,-9)   {$\ssg V^{\!\vee} $}
  \put(36,-9)   {$\ssg V $}
  \put(-9.6,23) {\sse $\iota_{\!U}^{} $}
  \put(41,23)   {\sse $\iota_{\!V}^{} $}
  \put(60,50)   {$ = $}
   }
  \put(82,0) { {\Includepichtft{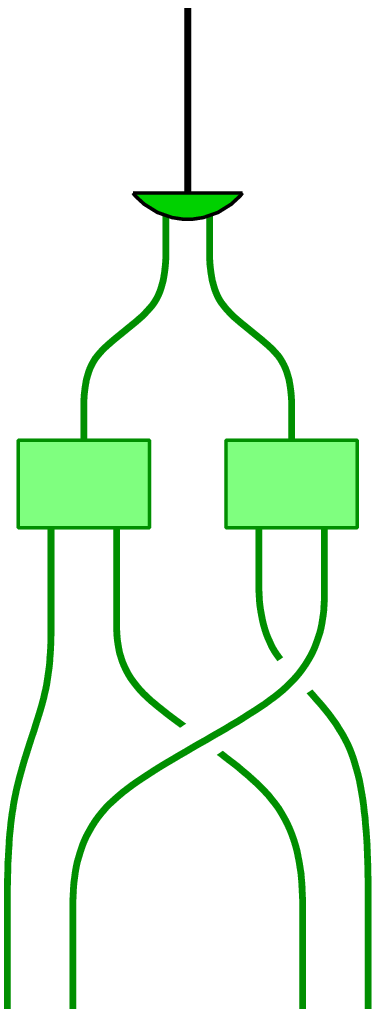}}
  \put(-10.8,66){$\ssg \gamma_{U,V} $}
  \put(34,65.2) {$\ssg \idsm_{V\otimes U} $}
  \put(-7,-8)   {$\ssg U^{\!\vee} $}
  \put(6,-8)    {$\ssg U $}
  \put(18,115)  {$\ssg L $}
  \put(19.4,23) {$\ssg c $}
  \put(25.6,-9) {$\ssg V^{\!\vee} $}
  \put(35.3,36.3) {$\ssg c $}
  \put(38.3,-8) {$\ssg V $}
  \put(-11.8,79){$\ssg (V{\otimes}U)^{\!\vee}_{} $}
  \put(25.5,79) {$\ssg V{\otimes}U $}
   }
    \put(213,8) {
  \put(-1,0) { {\Includepichtft{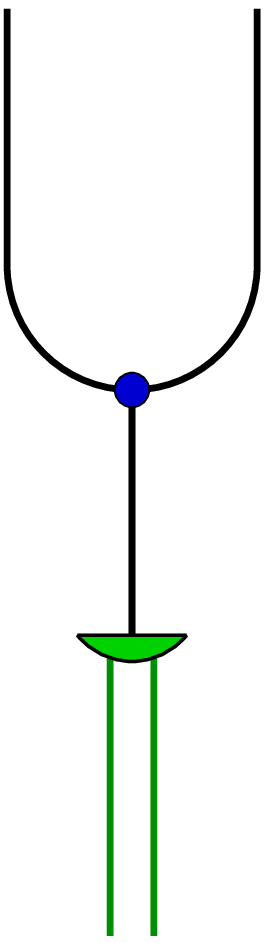}}
  \put(-2.5,107){$\ssg L $}
  \put(24.3,107){$\ssg L $}
  \put(2.6,-9)  {$\ssg U^{\!\vee} $}
  \put(15.3,-9) {$\ssg U $}
  \put(15.9,53.5) {$\ssg \Delta_L $}
  \put(55,50)   {$ = $}
  \put(84,0) { {\Includepichtft{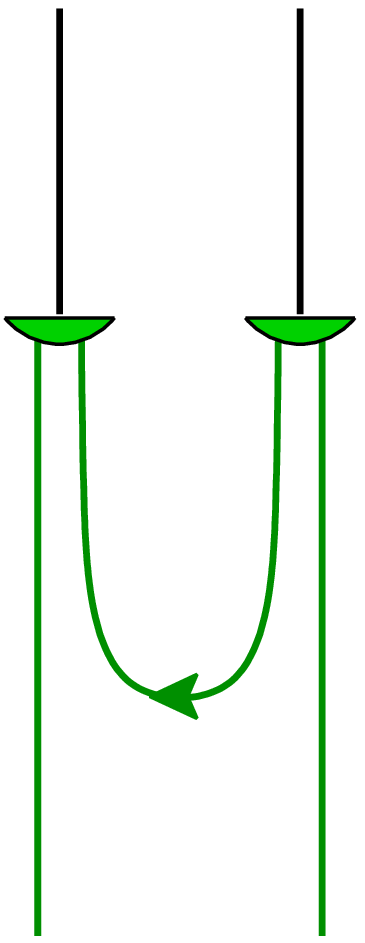}}
  \put(2.8,107) {$\ssg L $}
  \put(30.2,107){$\ssg L $}
  \put(-1,-9)   {$\ssg U^{\!\vee} $}
  \put(30.5,-9) {$\ssg U $}
   } } }
   }
    \put(13,2) {
  \put(0,17) { {\Includepichtft{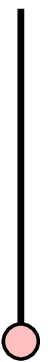}}
  \put(-.5,42) {$\ssg L $}
  \put(5.7,1)   {$\ssg \eta_L^{} $}
  \put(25,23)   {$ = $}
  \put(52,-18) {\Includepichtft{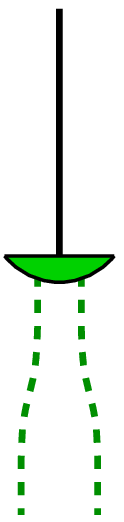}}
  \put(56,42)   {$\ssg L $}
   }
  \put(114,4) { {\Includepichtft{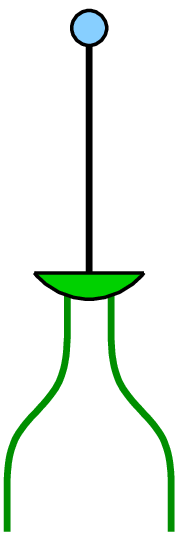}}
  \put(13.7,54.8) {\sse$ \eps_L^{} $}
  \put(-5.2,-9) {$\ssg U^{\!\vee} $}
  \put(15,-9)   {$\ssg U $}
  \put(38,30)   {$ = $}
  \put(67,10)   {\Includepichtft{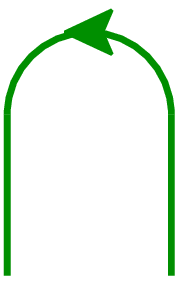}}
  \put(62.4,1)  {$\ssg U^{\!\vee} $}
  \put(82.3,1)  {$\ssg U $}
   }
  \put(260,-13) { {\Includepichtft{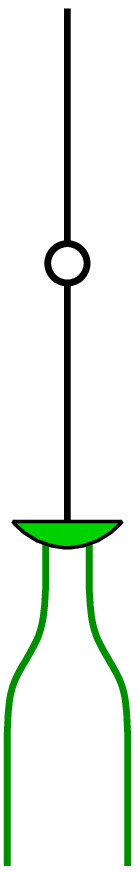}}
  \put(11.9,64.5) {$\sse \apo_L $}
  \put(4.4,98)  {$\ssg L $}
  \put(-5.2,-9) {$\ssg U^{\!\vee} $}
  \put(11.5,-9) {$\ssg U $}
  \put(34,44)   {$ = $}
  \put(63,0)   {\Includepichtft{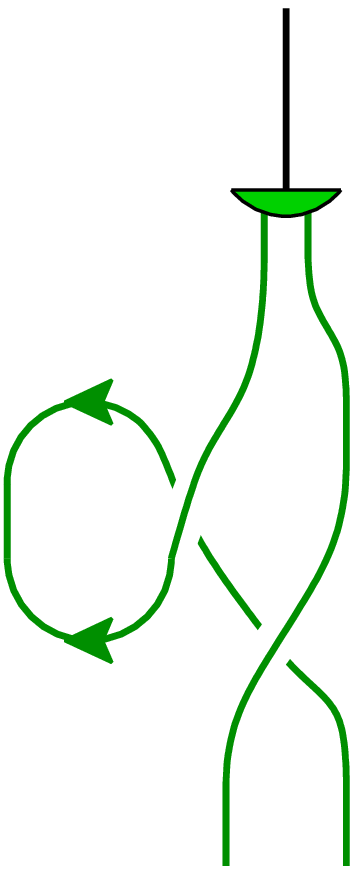}}
  \put(91.3,98) {$\ssg L $}
  \put(76.5,62) {$\ssg U^{\!\vee\!\vee} $}
  \put(84.9,38.4) {$\ssg c $}
  \put(91.6,16) {$\ssg c $}
  \put(99.9,64) {$\ssg U^{\!\vee} $}
  \put(80.9,-9) {$\ssg U^{\!\vee} $}
  \put(97.7,-9) {$\ssg U $}
   } } }
and
  \eqpic{def.hopa} {120}{31} {
  \put(0,4) { {\Includepichtft{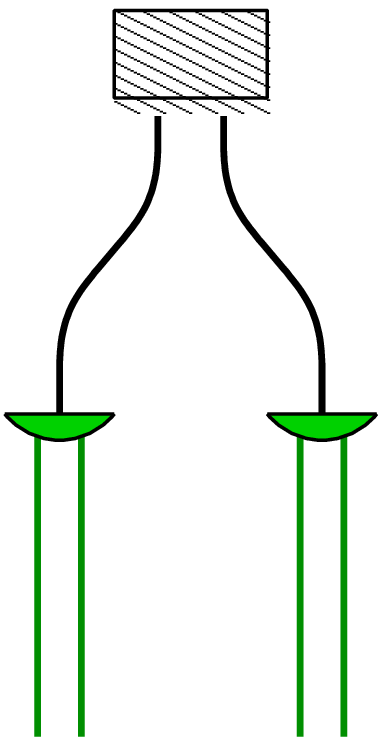}}
  \put(-5.5,-9) {$\ssg U^{\!\vee} $}
  \put(7.5,-9)  {$\ssg U $}
  \put(24.5,-9) {$\ssg V^{\!\vee} $}
  \put(37.5,-9) {$\ssg V $}
  \put(14.7,74.9) {$\sse \omega_L$}
  \put(61,38)   {$ = $}
  \put(95,0)  {\Includepichtft{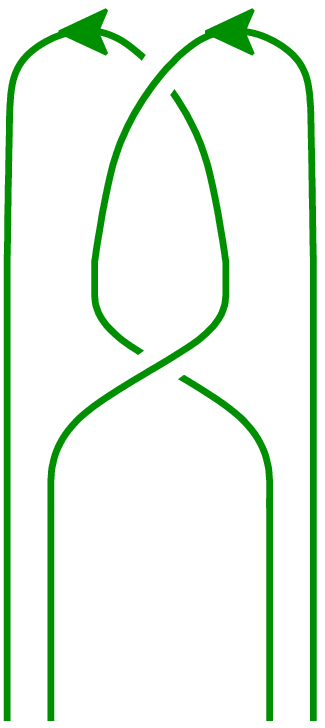}}
    \put(91,0){
  \put(-5.3,-9) {$\ssg U^{\!\vee} $}
  \put(7.7,-9)  {$\ssg U $}
  \put(18.8,33.1) {$\ssg c $}
  \put(18.8,62.9) {$\ssg c $}
  \put(24,-9)   {$\ssg V^{\!\vee} $}
  \put(37,-9)   {$\ssg V $}
    } } }

\begin{rem}\label{rem:Cobetc}~\\[2pt]
(i)\, That the object $L$ that is associated with the creation of handles 
carries a Hopf algebra structure is by no means a coincidence. Indeed, Hopf 
algebras are ubiquitous in constructions with three-dimensional cobordisms,
and specifically the \emph{handle}, i.e.\ a torus with an open disk removed, 
is a Hopf algebra 1-morphism in the bicategory \Cob\ of three-dimensional 
cobordisms with corners \cite{crye2,yett6}. Moreover, there exists a 
surjective functor from the braided monoidal category freely generated by 
a Hopf algebra object to \Cob\ \cite{kerl9}.
\\[2pt]
(ii)\, The Hopf algebra $L \eq L(\C)$ is directly associated with the
category \C, and thereby with the CFT having \C\ as its category of chiral data,
and analogously for $K$: for given \C\ there is a uniquely (up to isomorphism) 
determined chiral handle Hopf algebra, and likewise a unique bulk handle Hopf 
algebra. This is in contrast to the bulk state space: for given \C\ there is 
typically more than one possibility. Specifically, in rational CFT the 
different bulk state spaces are in bijection with Morita classes of simple 
symmetric special Frobenius algebras in \C, see formula \erf{FZA} above.
\\[2pt]
(iii)\, Via its integrals and Hopf pairing, the Hopf algebra $L$ gives rise to 
three-manifold invariants as well as to representations of mapping class groups 
(see \cite{lyub6,lyub8,kerl5,vire4}, and \cite[Sects.\,4.4\,\&\,4.5]{fuSc17} 
for an elementary introduction). Even though for non-semisimple \C\ these 
cannot be normalized in such a way that they fit together to furnish a 
three-dimensional topological field theory, one may still hope that this 
hints at a close relationship with three-dimensional topology even in the 
non-semisimple case.
\\[2pt]
(iv)\, Obviously, the object $L$ of \C\ is obtained from the object \Fo\ of 
\CbC\ by applying the functor that on $\boxtimes$-factorizable objects of 
\CbC\ acts as $U \boti V \,{\mapsto}\, U \oti V$. This functor is called the 
\emph{diagonal restriction functor} in \cite{lyub10}.
\end{rem}
 
It will be relevant to us that the algebra $K$ acts (and coacts) on the object 
\Fo\ of \CbC, and in fact on any object of \CbC. We formulate the relevant 
statements directly for objects in \CbC; with $K$ replaced by $L$, 
they apply analogously in \C. 

For $Y\iN \CbC$ set
  \be
  \delta^K_Y := (\id_Y \oti \iK_Y) \circ (b_Y \oti \id_Y) ~\in \HomCC(Y,Y\oti K)
  \labl{deltaK}
as well as
  \be
  \req_Y^K := (\eps_K\oti \id_Y) \cir \QB_Y ~\in \HomCC(K\oti Y,Y) \,,
  \labl{req}
where in the latter formula the morphism $\QB_Y$ is the \emph{partial monodromy}
between $K$ and $Y$, defined as
   \eqpic{QHX} {105} {44} { \put(0,1){
   \put(0,0)  {\begin{picture}(0,0)(0,0)
         \scalebox{.38}{\includegraphics{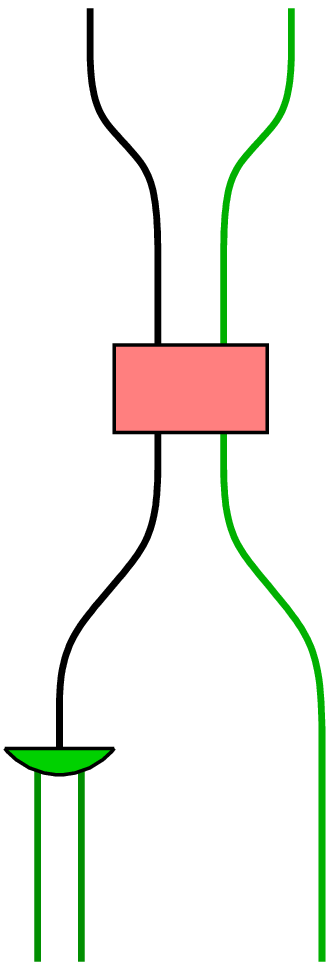}}\end{picture}
   \put(-6.6,60.5){$ \QB_Y $}
   \put(-4,-9.2)  {\sse$ X^{\!\vee} $}
   \put(6.8,109)  {\sse$ K $}
   \put(7,-9.2)   {\sse$ X $}
   \put(31.3,-9.2){\sse$ Y $}
   \put(30.4,109) {\sse$ Y $}
   \put(14.3,21.5){\sse$ \iK_X $}
   }
   \put(60,50)    {$ := $}
   \put(94,0) {\begin{picture}(0,0)(0,0)
         \scalebox{.38}{\includegraphics{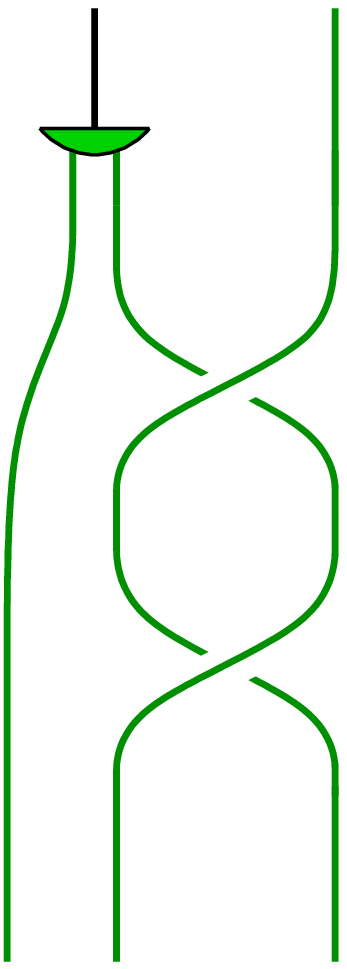}}\end{picture}
   \put(-6,-9.2)  {\sse$ X^{\!\vee} $}
   \put(7.4,109)  {\sse$ K $}
   \put(8,-9.2)   {\sse$ X $}
   \put(22.2,26.9){\sse$ c $}
   \put(22.2,58)  {\sse$ c $}
   \put(34,-9.2)  {\sse$ Y $}
   \put(34.3,109) {\sse$ Y $}
   \put(17.8,90)  {\sse$ \iK_X $}
   } } }

\begin{prop} {\rm \cite[Fig.\,7]{lyub8}}\\
For any object $Y\!$ of a finite ribbon category \C, the morphism 
{\rm \erf{deltaK}} endows $Y$ with the structure of a right $K$-comodule.
\end{prop}

\begin{prop} {\rm \cite[Rem.\,2.3]{fuSs5}}\label{prop:YD}\nxl1
{\rm (i)}\, For any object $Y\!$ of a finite ribbon category \C, the morphism 
{\rm \erf{req}} endows $Y$ with the structure of a left $K$-module.
\nxl1
{\rm (ii)}\, The module and comodule structures {\rm \erf{req}} and 
{\rm \erf{deltaK}} fit together to the one of a left-right Yetter-Drinfeld 
module over $K$. This affords a fully faithful embedding of \CbC\ into the 
ca\-te\-gory of left-right Yetter-Drinfeld modules over $K$ internal to \CbC.
\end{prop}

Since the crucial ingredient of $\req_Y^K$ is a double braiding, we refer to
$\req_Y^K$ as the \emph{partial mono\-dro\-my action} of $K$ on $Y$.
 
\begin{rem}
(i)\, The second part of Proposition \ref{prop:YD} fits nicely with the 
result \cite[Thm.\,3.9]{yett6} that every 1-morphism of the cobordism 
bicategory \Cob\ carries a structure of left-right Yetter-Drinfeld module 
over the one-holed torus (compare Remark \ref{rem:Cobetc}(i)).
\\[2pt]
(ii)\, The category $\C_L$ of $L$-modules in \C\ is braided equivalent to the 
monoidal center $\Z\C$, see Theorem 8.13 of \cite{brVi5}.
\\[2pt]
(iii)\, The full subcategory $\C_L^{\mathcal Q}$ of $\C_L$ consisting of the 
modules $(U,\req_U^L)$ for $U\iN\C$, with $\req_U^L$ defined analogously as 
in \erf{req}, is a monoidal subcategory: by the definition of $\Delta_L$ and
the functoriality of the braiding, one has
$\req_{U\otimes V}^L \eq (\req_U^L \oti \req_V^L) \cir \Delta_L $.
\end{rem}

\medskip

Next let us specialize to the situation considered in Section 
\ref{ssec:coendFH}, i.e.\ that the finite ribbon category \C\ is equivalent to 
the category \HMod\ for $H$ a factorizable Hopf algebra. Then the chiral
handle Hopf algebra $L \iN \HMod$ is the vector space \Hs\ dual to $H$ endowed 
with the coadjoint $H$-action
  \eqpic{def_lads} {120} {41} {
  \put(0,39)    {$ \rho\coa ~:= $}
  \put(53,-1) { {\Includepichopfsm{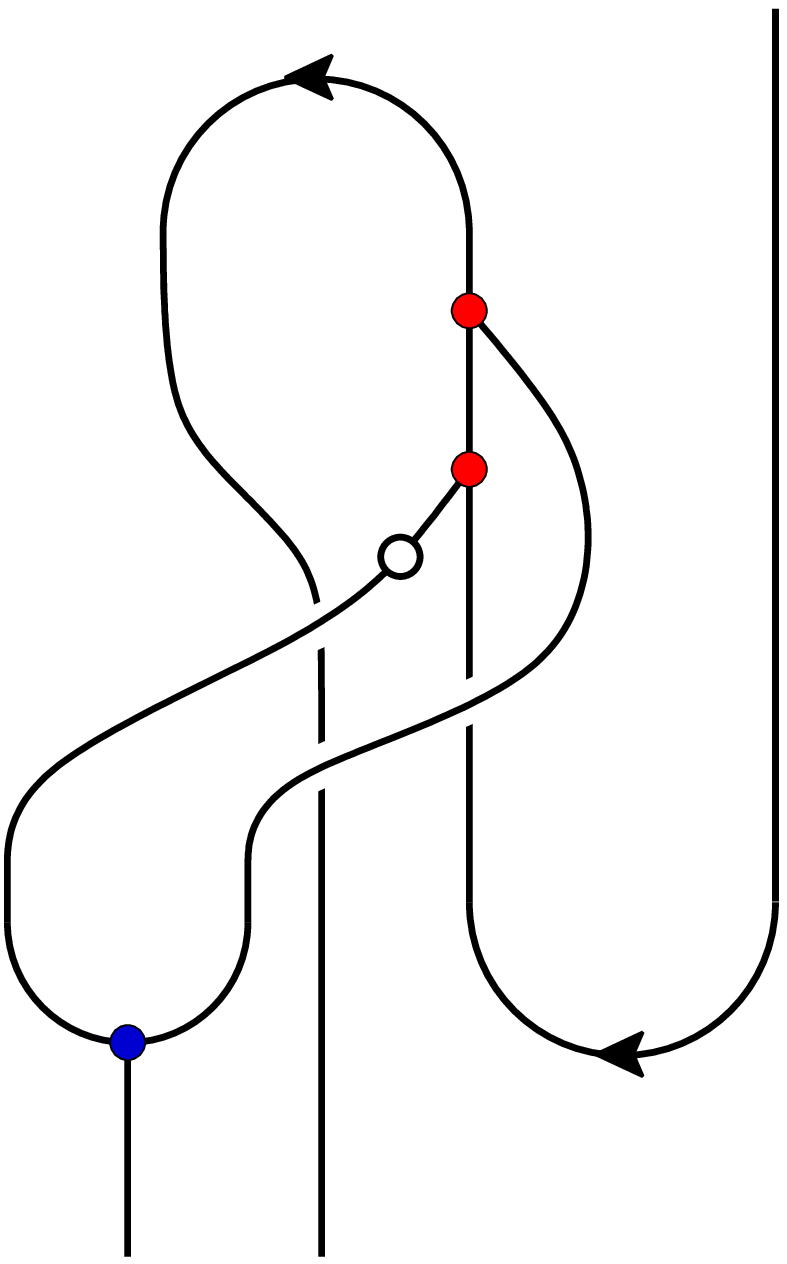}}
  \put(3.6,-9.2){$\ssg \Hss $}
  \put(20.1,-9.2){$\ssg \Hss $}
  \put(25,55.5) {$\aposm $}
  \put(27,63)   {$\ssg m $}
  \put(37.8,74) {$\ssg m $}
  \put(55,100)  {$\ssg \Hss $}
  } }
and the members of the dinatural family $\iL$ are the linear maps
(see \cite[Lemma\,3]{kerl5} and \cite[Sect.\,4.5]{vire4})
  \eqpic{p8} {165}{34} {
  \put(30,36)   {$ \iL_U = $}
  \put(80,0)  {\Includepichtft{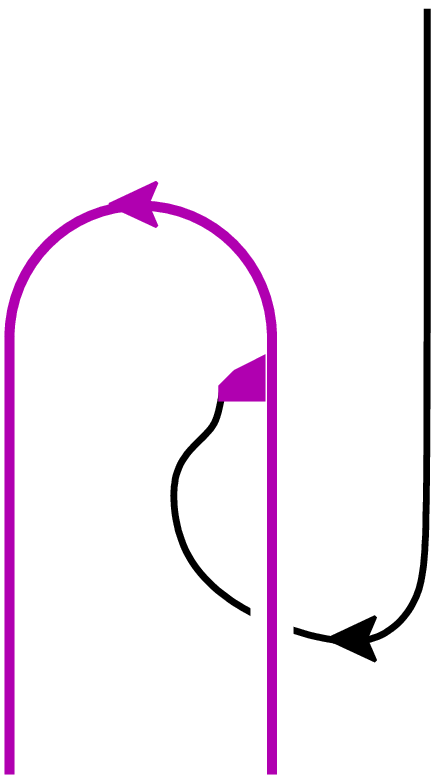}}
  \put(111,43)  {$\sss \rho_{\!U}^{} $}
  \put(75,-9)   {$\ssg U^* $}
  \put(104.6,-9){$\ssg U $}
  \put(122.5,89){$\ssg \Hss $}
  }
When expressed in terms of vector space elements, these morphisms
are nothing but the matrix elements of left multiplication in $H$.

The unit, counit and coproduct of the Hopf algebra $L$, as given by \erf{p7}
for the case of general finite ribbon categories, now read
  \be
  \bearl
  \eta_L = (\eps_H)^* \equiv (\eps_H \oti \idHs) \circ b_H^\ko \,,
  \Nxl4
  \eps_L = (\eta_H)^* \equiv d_H^\ko \circ (\idHs \oti \eta_H) \qquand
  \Delta_L = (m_H\op)^*_{} \,,
  \eear
  \ee
while two equivalent descriptions of the product are
  \eqpic{Hmod_prod_res} {300} {67} {
  \put(0,62)    {$ m_L ~= $}
  \put(60,0) { \Includepichtft{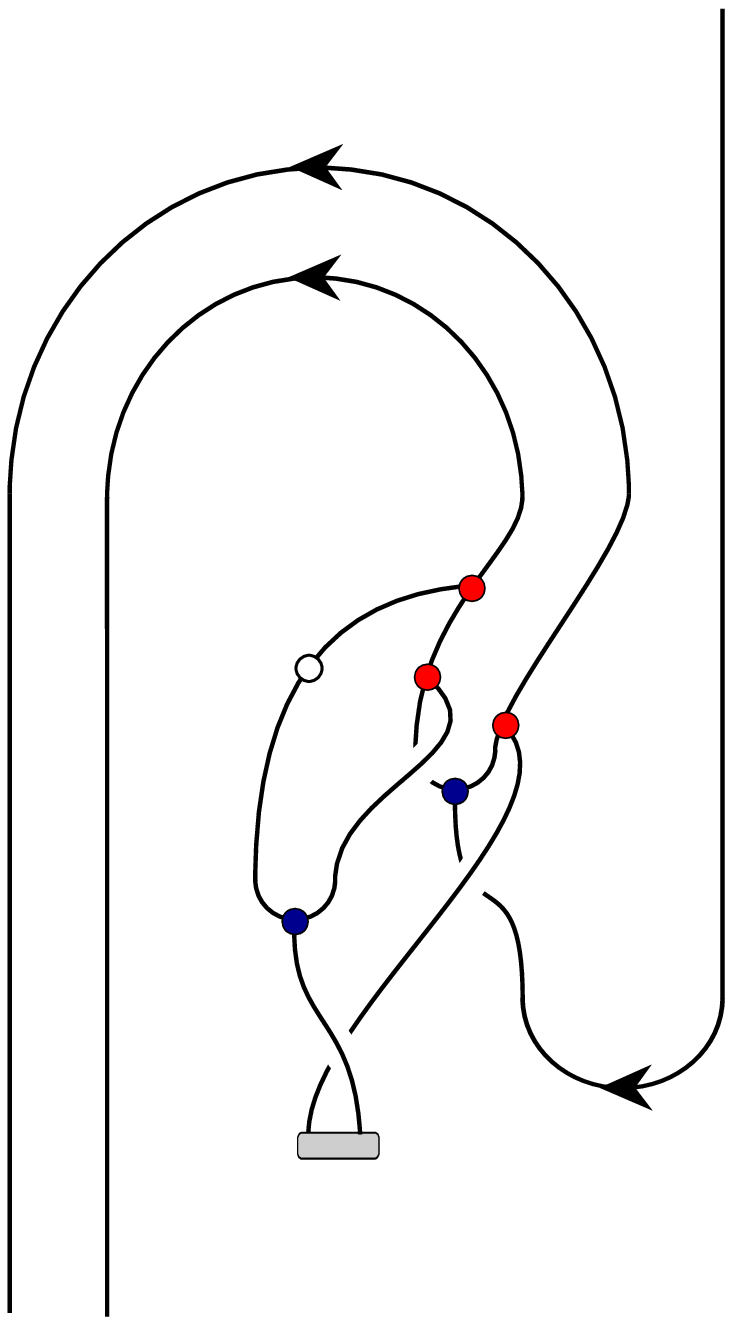}
  \put(-7,-8)   {\sse$ \Hss $}
  \put(7,-8)    {\sse$ \Hss $}
  \put(26.9,73){\sse$ \apo $}
  \put(31.8,9.8){\sse$ R $}
  \put(75,148)  {\sse$ \Hss $}
  }
  \put(174,62)  {$ =$}
  \put(211,0) { \Includepichtft{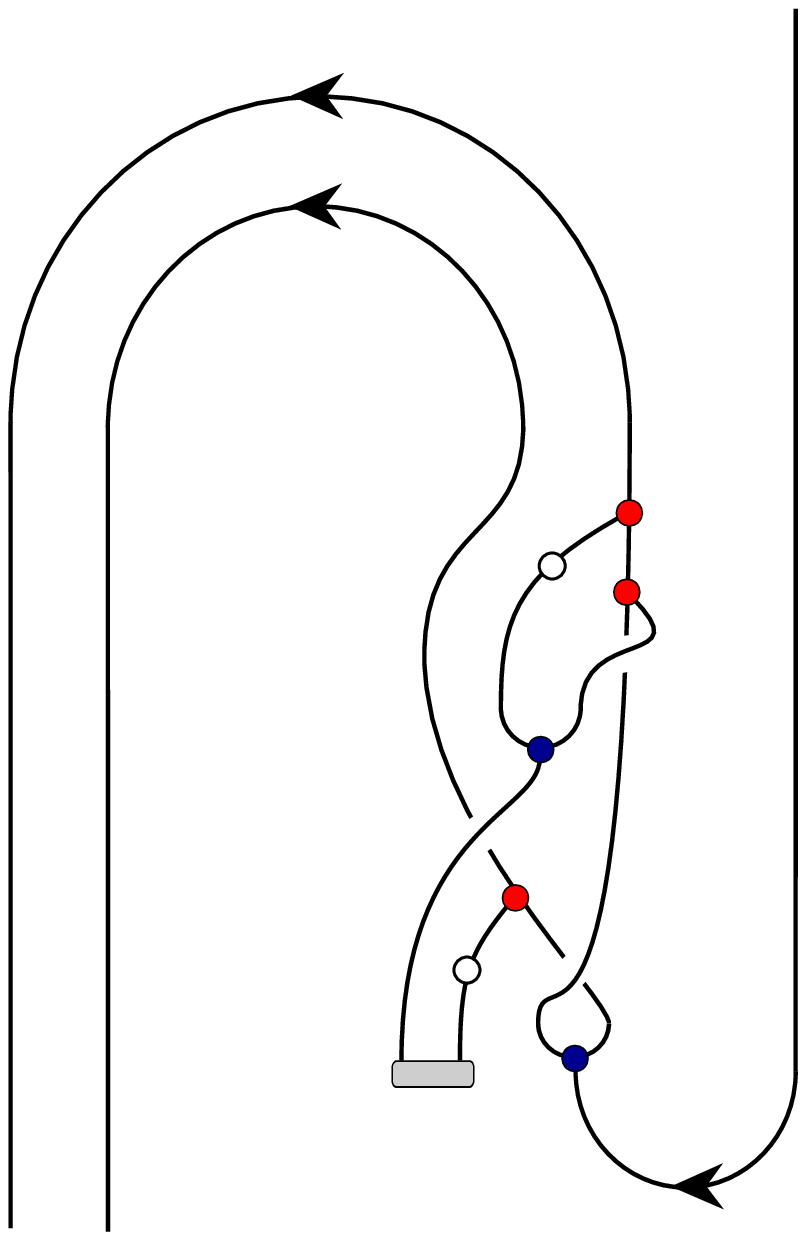}
  \put(-7,-8)   {\sse$ \Hss $}
  \put(7,-8)    {\sse$ \Hss $}
  \put(42.9,8.6){\sse$ R $}
  \put(52.8,27) {\sse$ \apo $}
  \put(58.7,66.4){\sse$ \apo $}
  \put(82,139)  {\sse$ \Hss $}
  } }
and the antipode is
  \eqpic{apo_Haa_A} {105} {42} {
  \put(0,43)    {$ \apo_\bicoaa ~= $}
  \put(52,0) {\Includepichtft{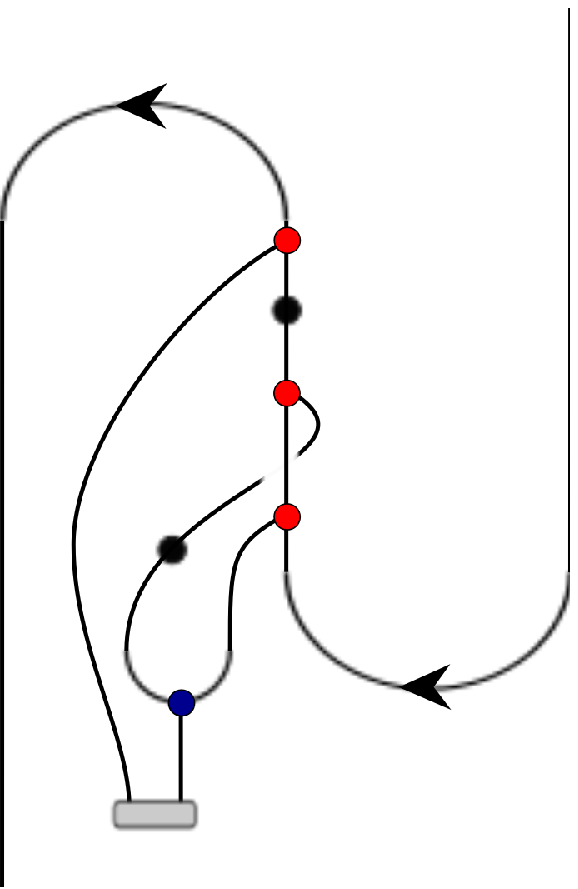}
  \put(-5,-8.5) {\sse$ \Hss $}
  \put(11.1,39) {\sse$ \apoi $}
  \put(12.7,-.9){\sse$ R $}
  \put(59,100.4){\sse$ \Hss $}
  } }

Further, our finiteness assumptions imply that $L$ now comes with an integral 
and a cointegral, given by
  \be
  \Lambda_L = \lambda^* \qquand \lambda_L = \Lambda^* ,
  \ee
respectively. Both of them are two-sided, even though the cointegral $\lambda$
of $H$ in general is only a right cointegral.

Similarly, identifying, as in Section \ref{ssec:coendFH}, the enveloping 
category $\HMod^{\rm rev} \,{\boxtimes}\, \HMod$ with the category \HBimod\ of 
bimodules (with the ribbon structure presented there), the coend $K \iN \HBimod$
is the \emph{coadjoint bimodule}, that is, the tensor product $\Hs \otik \Hs$ 
of two copies of the dual space \Hs\ endowed with the coadjoint left $H$-action 
\erf{def_lads} on the first tensor factor and with the coadjoint right 
$H$-action on the second factor, with dinatural family
   \eqpic{def_iHaa_X} {150} {45} {
   \put(18,44)    {$ \iK_X ~:= $}
   \put(75,0) { \Includepichtft{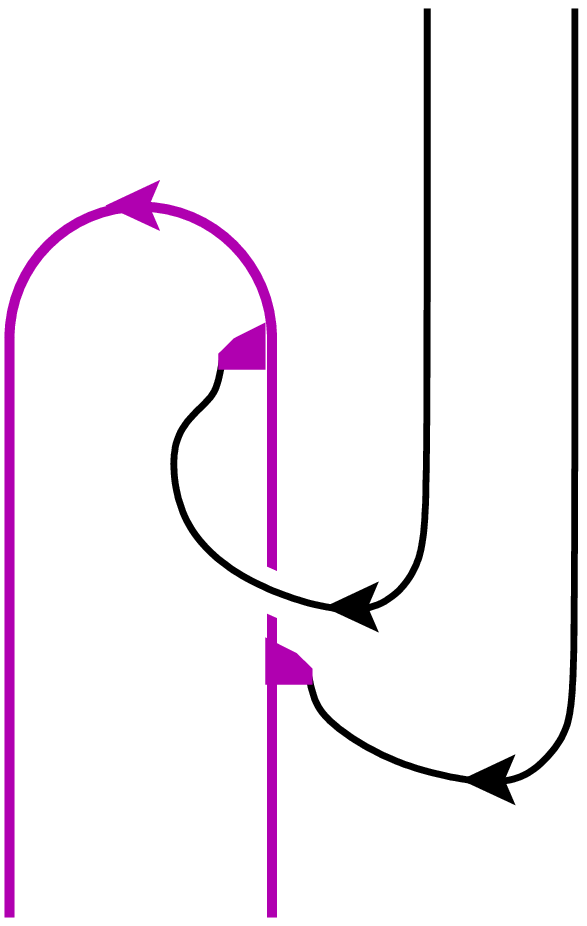}
   \put(-4.1,-8.5){\sse$ \Xs $}
   \put(19.5,28)  {\sse$ \ohr_{\!X}^{} $}
   \put(25.3,-8.5){\sse$ X $}
   \put(31.2,62)  {\sse$ \rho_{\!X}^{} $}
   \put(42.6,104) {\sse$ \Hss $}
   \put(59.4,104) {\sse$ \Hss $}
   } }
for any $H$-bimodule $X \eq (X,\rho_{\!X}^{},\ohr_{\!X}^{})$. The structural 
morphisms of $K$ as a Hopf algebra and its integral and cointegral are 
straightforward analogues of the expression given for $L$ above; for 
explicit formulas we refer to (A.32)\,--\,(A.36) of \cite{fuSs3}.

The partial monodromy action \erf{req} of $K$ on an $H$-bimodule $(Y,\rho_Y,
\ohr_Y)$ is given in terms of the monodromy matrix $Q$ and its inverse by
  \eqpic{Lyubact_HKH} {120} {46} {
  \put(0,47)       {$ \req^K_Y=~ $}
    \put(50,0) { \Includepichtft{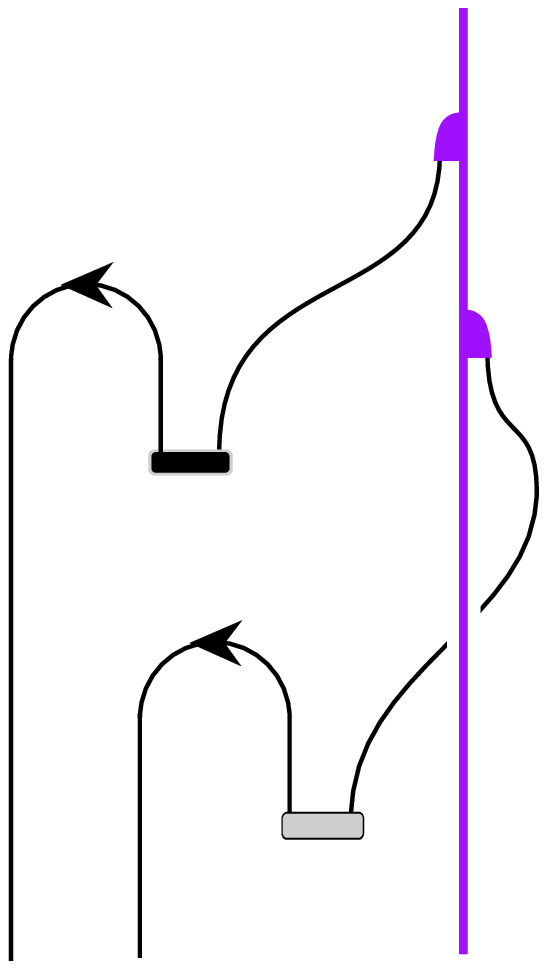}
  \put(-5.4,-8.5)  {\sse$ \Hss $}
  \put(10,-8.5)    {\sse$ \Hss $}
  \put(46.4,-8.5)  {\sse$ Y $}
  \put(31.6,7.3)   {\sse$ Q $}
  \put(13,46.8)    {\sse$ Q^{-1} $}
  \put(54.7,68.8)  {\sse$ \ohr_Y^{} $}
  \put(35.2,90.5)  {\sse$ \rho_Y^{} $}
  \put(47.7,108.9) {\sse$ Y $}
  } }
i.e.\ the natural $K$-action is nothing but the $H$-bimodule action composed
with variants of the Drinfeld map \erf{def-drin}.

\medskip

\begin{rem}\label{rem:factorizable} ~\\[2pt]
(i)\,
For $H$ a ribbon Hopf algebra, the Hopf pairing \erf{def.hopa} of the handle 
Hopf algebra $L$ is non-degenerate iff $H$ is factorizable. It is thus natural 
to call more generally a finite ribbon category \C\ \emph{factorizable} iff 
the Hopf pairing \erf{def.hopa} of $L(\C)$ is non-degenerate. 
\\[2pt]
(ii)\,
Factorizability implies e.g.\ that the integral of $L(\C)$ is two-sided and 
that $L(\C)$ also has a two-sided cointegral (Prop.\ 5.2.10 and Cor.\ 5.2.11 
of \cite{KEly}). If \C\ is semisimple, then being factorizable is equivalent 
to being modular. Thus factorizability may be seen as a generalization of
modularity to non-semisimple categories (the authors of \cite{KEly} 
even use the qualification `modular' in place of `factorizable').
\\[2pt]
(iii)\,
A quasitriangular Hopf algebra $H$ is factorizable iff its Drinfeld double 
$D(H)$ is isomorphic, in a particular manner, to a two-cocycle twist of 
the tensor product Hopf algebra $H \oti H$ \cite[Thm.\,4.3]{schne8}, and 
thus \cite[Rem.\,4.3]{etno2} iff the functor
that acts on objects $U\boti V$ of the enveloping category of \HMod\ as
  \be
  U\boti V \,\longmapsto\, (U\oti V,z_{U\otimes V}^{})
  \qquad{\rm with} \qquad
  z_{U\otimes V}^{}(W)
  := (c_{U,W}^{} \oti \id_V^{}) \circ (\id_U^{} \oti c_{W,V}^{-1})
  \labl{funHbH2ZH}
furnishes a monoidal equivalence
  \be
  \HMod^{\rm rev} \,{\boxtimes}\, \HMod \,\stackrel\simeq\longrightarrow\,
  \Z\HMod 
  \labl{ZH=HbH}
between the enveloping category and the monoidal center of \HMod.
\\[2pt]
(iv)\,
Now the bulk state space in conformal field theory is an object in \CbC. Thus 
if we want to be able to describe the bulk state space, in line with the 
semisimple case \erf{F=ZA}, as a full center, we should better be allowed to 
regard the full center $Z(A)$ of an algebra $A$ in a factorizable finite ribbon 
category \C, which by definition is an object in $\Z\C$, also as an object in 
\CbC, and thus want $\Z\C$ and \CbC\ to be monoidally equivalent.\,%
 \footnote{~In fact, in \cite{etno2} this property is used to \emph{define}
  factorizability for braided monoidal categories that are not necessarily
  ribbon.}
\\
We do not know whether this requirement is satisfied
for \emph{all} factorizable finite ribbon categories.
On the other hand, for the condition to be satisfied it is certainly not 
required that \C\ is ribbon equivalent to \HMod\ for a ribbon Hopf algebra $H$.
Specifically, the notion of factorizability can be 
extended from Hopf algebras to weak Hopf algebras \cite[Def.\,5.11]{nitv}, and 
again a weak Hopf algebra $H$ is factorizable iff the functor \erf{funHbH2ZH} 
is a monoidal equivalence \cite[Rem.\,4.3]{etno2}.
This covers in particular the case of all semisimple \C, because every semisimple 
finite tensor category is equivalent to the representation category of some 
semisimple \findim\ weak Hopf algebra \cite[Thm.\,4.1\,\&\,Rem.\,4.1(iv)]{ostr}.
\end{rem}


\section{The torus partition function}\label{sec:pf}

\subsection{The partition function as a character}

By definition, the torus partition function $Z$ of a CFT, whether rational or 
not, is the character of the bulk state space $F$. Here the term character 
refers to $F$ as a module over the tensor product of the left and right copies 
of the chiral algebra \V. That is, the character is a real-analytic function of
the modulus $\tau$ of the torus, which takes values in the complex upper half 
plane, and it is the generating function for dimensions of homogeneous 
subspaces of $\V{\otimes_\complex}\V$-modules. As such, $Z$ is a sum of 
characters of simple $\V{\otimes_\complex}\V$-modules, even though $F$ is, 
in general, not fully reducible. 

Referring to the chiral algebra \V\ is not necessary, though. Rather, as for 
our purposes we are allowed to work at the level of \repV\ as an abstract 
factorizable ribbon category, we can regard $F$ just as an object of 
$\CbC \,{\simeq}\, \mathcal{R}ep(\V{\otimes_\complex}\V)$. Indeed, we know 
from \erf{CorSigma} that the torus partition function -- the zero-point 
correlator on the torus $\mathrm T$ -- is an element of the space
  \be
  V(\mathrm T\,{\sqcup}\,{-}\mathrm T)
  \,\cong\, \HomC(L,\one) \otic \HomC(L,\one) \,\cong\, \HomCC(K,\one) 
  \ee
of chiral blocks. Now the morphism space $\HomCC(K,\one)$ contains in particular
the \emph{characters} of the algebra $K$. An immediate conjecture for the torus 
partition function is thus the character $\chii^K_F$ of the bulk state space as 
a module (with action $\req_F^K$ as defined in \erf{req}) over the bulk handle 
Hopf algebra $K$. As we will see in Section \ref{sec:corfus} below, $\chii^K_F$ 
is in fact just the particular member $(g,n) \eq (1,0)$ of a family of morphisms 
that are natural candidates for correlation functions at any genus $g$ and with 
any number $n$ of bulk insertions.

In this description the term character now refers to $F$ as a $K$-module.
The notion of the character of a module over an associative \ko-algebra
is standard and is explained in detail in Appendix \ref{algchar}.
For an algebra $A$ in a monoidal category \C\ one can set up 
representation theory in much the same way as for a \ko-algebra, i.e.\ for
an algebra in \Vectk. The notion of character then still makes sense provided 
that \C\ is sovereign, which for the categories of our interest is the case.\,%
 \footnote{Any ribbon category is sovereign, i.e.\ (see e.g.\
 Def.\,2.7 of \cite{drab3}) the left and right dualities are connected by a
 monoidal natural isomorphism.}
Concretely, the formula \erf{chiiMA} for the character of a module $M$ over an
algebra $A$ in \Vectk\ gets modified to
  \be
  \chii_M^A = \mathrm{tr}_M(\rho) = \tilde d_M \circ
  (\rho \oti \pi_M) \circ (\id_A\oti b_M) ~\in \HomC(A,\one) \,,
  \labl{chiiMA-C}
with $\pi_M\colon M^\vee \To {}^{\vee}\!M$ the sovereignty isomorphism between 
the right and left duals of $M$.

In the case at hand the relevant algebra is the bulk handle Hopf algebra $K$, 
and its action is given by \erf{req}. We thus have
  \eqpic{Qrep2} {100} {43} {
  \put(-70,51) {$ \chii^K_Y \circ \iK_X ~= $}
  \put(0,0)  {\Includepichtft{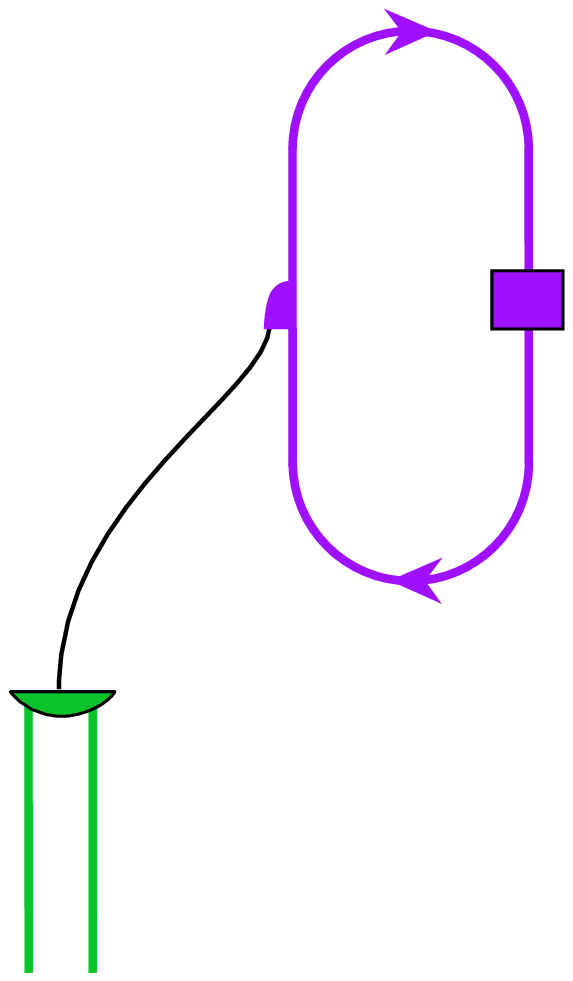}
  \put(-5,-8.5){\sse $X^{\!\vee}$}
  \put(7,-8.5) {\sse $X$}
  \put(13,30)  {\sse $\iK_X$}
  \put(17.5,73){\sse $\req^K_Y$}
  \put(33,58)  {\sse $Y$}
  \put(63,73)  {\sse $\pi_Y$}
  }
  \put(81,51)  {$ = $}
  \put(110,0) {\Includepichtft{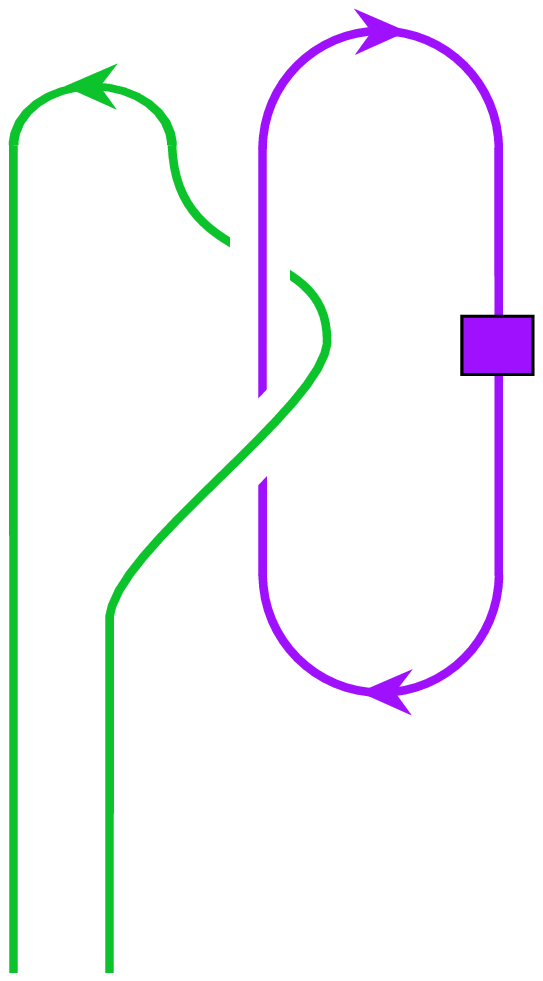}
  \put(-6,-8.5){\sse $X^{\!\vee}$}
  \put(7,-8.5) {\sse $X$}
  \put(30,44)  {\sse $Y$}
  \put(21,59)  {\sse $c$}
  \put(30,80)  {\sse $c$}
  \put(60,68)  {\sse $\pi_Y$}
   } }
for any $K$-module $(Y,\req^K_Y)$ in \CbC, and analogously for the character 
$\chii^L_U$ of an $L$-module $(U,\req^L_U)$ in \C.

\begin{rem}\label{rem:dual}
Since $L$ and $K$ are Hopf algebras, there are natural notions of left and 
right dual modules. The character of the $L$-module $U^\vee$ right dual to 
$U \eq (U,\req^L_U)$ is given by the same morphism as the one for $\chii^L_U$, 
except that the braidings in \erf{Qrep2} get replaced by inverse braidings.
\end{rem}

\medskip

It is worth being aware that so far the coend $F \eq \Fo$ \erf{coendCC}
(respectively, $F \eq \Fomega$ \erf{Fomega}) is only conjecturally the bulk 
state space of a conformal field theory, and similarly the morphism $\chii^K_F$ 
is merely a candidate for the torus partition function $Z$ of that CFT. But 
just like we could verify that, for the case $\C \,{\simeq}\, \HMod$ (and 
$F \eq \Fomega$ for any ribbon automorphism $\omega$ of $H$) the coend has 
the desired properties of being a commutative symmetric Frobenius algebra, 
we will see below that in this case $\chii^K_F$ has the desired property of 
being a bilinear combination of suitable chiral characters with non-negative 
integral coefficients. Moreover, these coefficients turn out to be quantities 
naturally associated with the category \C.

The status of $\chii^K_F$ can be corroborated further by
establishing modular invariance. Indeed, this follows as a corollary from the
mapping class group invariance of general correlation functions that we
will present in Section \ref{sec:corfus} below.
The partition function should in addition be compatible with sewing. At this 
point we have no handle on this property yet. Thus, while we can prove modular 
invariance at any genus, as far as sewing is concerned the state of affairs 
bears some similarity with the situation in rational CFT prior to the 
development of the TFT construction \cite{fuRs} of correlators: While modular 
invariance is a crucial property of the partition function, it is only 
necessary, but in general not sufficient, and indeed there are plenty of 
modular invariants which are incompatible with sewing. On the other hand, 
for all rational CFTs the charge conjugation 
modular invariant \emph{is} compatible with sewing \cite{fffs2}, and
accordingly we do expect that, for any factorizable finite ribbon category
\C, at least for $F \eq \Fo$ the character $\chii^K_F$ does provide the 
torus partition function of a CFT with $\C \,{\simeq}\, \repV$.

\begin{rem}
As already mentioned, the categories \C\ of chiral data for the logarithmic
$(1,p)$ triplet models should be closely related to finite tensor categories 
that are representation categories of a finite-dimensional complex Hopf algebra.
Indeed a finite-dimensional Hopf algebra $H_p$, the restricted quantum group 
$\overline U_q(\mathfrak{sl}_2)$ with a primitive $2p$th root $q$ of unity,
has been proposed \cite{fgst} for the $(1,p)$-model; this Hopf algebra is not 
quasitriangular.
It has been established \cite{naTs2} that $H_p$\Mod\ is equivalent to the category 
of chiral data for the $(1,p)$ triplet model at the level of abelian categories.
The situation becomes more subtle once braidings and monodromies are involved, but
still the monoidal category $H_p$\Mod\ admits a monodromy matrix. Accordingly,
typical aspects of these models for which only monodromies, but not the 
braiding morphisms themselves are involved, should
be covered by our analysis. Specifically, the modular transformations of
characters of $H$-modules considered here should reproduce the modular group
\rep\ on the center of $\overline U_q(\mathfrak{sl}_2)$.
Explicit calculation \cite{fgst} shows that the latter representation coincides
with the modular transformations of the characters of the chiral algebra
of the logarithmic $(1,p)$ triplet models.
\end{rem}


\subsection{Chiral decomposition}
  
The simple modules of $\V\otic\V$, i.e.\ the simple objects of 
$\CbC \,{\simeq}\, \mathcal{R}ep(\V\otic\V)$, are of the form $S_i \boti S_j$
with $S_i$, for $i\iN\I$, the simple \V-modules. For rational CFT, i.e.\ for
semisimple \C, the category \CbC\ is semisimple, too, so that in particular
the bulk state space $F$ decomposes as in formula \erf{F=ZA} into a 
direct sum of simple objects $S_i \boti S_j$ for appropriate $i,j\iN\I$.
When \C\ is non-semisimple, this is no longer the case. Moreover, for
non-semisimple \C\ one even cannot, in general, write $F$ as a direct sum of
$\boxtimes$-factorizable objects, i.e.\ of objects of the form $U\boti V$.

Nevertheless, since characters split over exact sequences \cite[Sect.\,1.5]{loren},
a chiral decomposition analogous to the one in rational CFT does exist for 
the torus partition function. Specifically,
if \C\ is a finite tensor category, for which the index set $\I$ is finite, 
the torus partition function can be written as a finite sum
  \be
  Z = \sum_{i,j\in\I} Z_{ij}\, \chii^\V_i \otic \chii^\V_j
  \labl{Z=sumij}
with $Z_{ij} \iN \zet_{\ge0}$.

For non-rational CFT the space of zero-point chiral blocks for the torus is 
\emph{not} exhausted by the characters -- that is, the characters of \V-modules
in the vertex algebra description, respectively by the characters $\chii^L_U$, 
for $U\iN\C$, of the $L$-modules $(U,\req^L_U)$. Rather, this space also 
includes linear combinations of so-called \emph{pseudo-characters} 
\cite{fgst,flga,gaTi,arNa}. Specifically, for any $C_2$-cofinite vertex algebra 
these functions can be constructed with the help of symmetric linear functions 
on the endomorphism spaces of suitable decomposable projective modules 
\cite{miya8,arik3}. The existence of an expression of the form \erf{Z=sumij} 
thus means in particular that the pseudo-characters do not contribute to the 
torus partition function. This certainly fits nicely with the physical idea of 
counting states; mathematically it is a non-trivial statement
that a decomposition of the form \erf{Z=sumij} exists, even without 
requiring integrality of the coefficients.

In a purely categorical setting, the analogue of the space of zero-point 
blocks for the torus is the space $\HomC(L,\one)$,   
One should expect that in analogy with \erf{Z=sumij} the character $\chii^K_F$ 
satisfies
  \be
  \chii^K_F \,\in\, 
  \HomC(L,\one) \otic \HomC(L,\one) \cong \HomCC(K,\one)
  \ee
and thus decomposes into products of simple $L$-characters $\chii^L_k$ as
  \be
  \chii^K_F = \sum_{k,l} x_{kl}^{}(F)\, \chii^L_k \otik \chii^L_l
  \labl{X=sumij}
with $x_{kl} \iN \zet_{\ge0}$.
We will now establish that this is indeed true in the case that
$\C \eq \HMod$ and $F \eq \Fomega$.


\subsection{The Cardy-Cartan modular invariant and its relatives}

Let us thus specialize again to the case that $\C \eq \HMod$ 
for some factorizable Hopf algebra $H$.

The sovereign structure for the ribbon categories \HMod\ and \HBimod\ is given by 
  \eqpic{pic-piV} {320} {36} {
  \put(0,39)    {$ \pi_U^\HModsm ~= $}
  \put(70,0)  {\Includepichtft{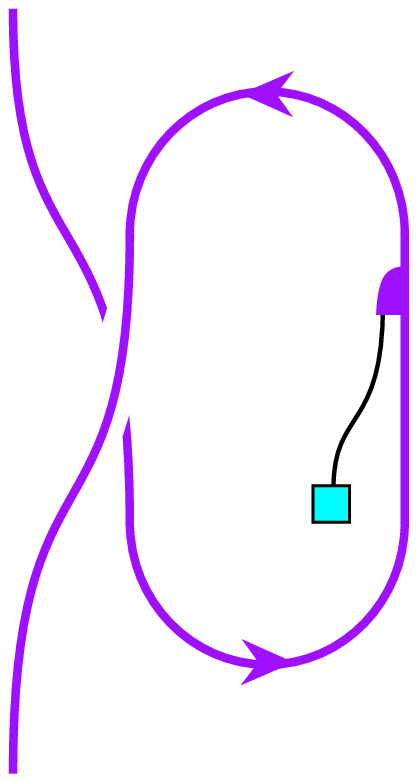}
  \put(-4.5,-8) {\sse$ U^*_{} $}
  \put(-3,88)   {\sse$ U^*_{} $}
  \put(28,28)   {\sse$ t $}
  \put(45.4,53) {\sse$ \rho_{\!U}^{} $}
  }
  \put(155,39) {and $\qquad~ \pi_X^\HBimodsm ~= $}
  \put(281,0)  {\Includepichtft{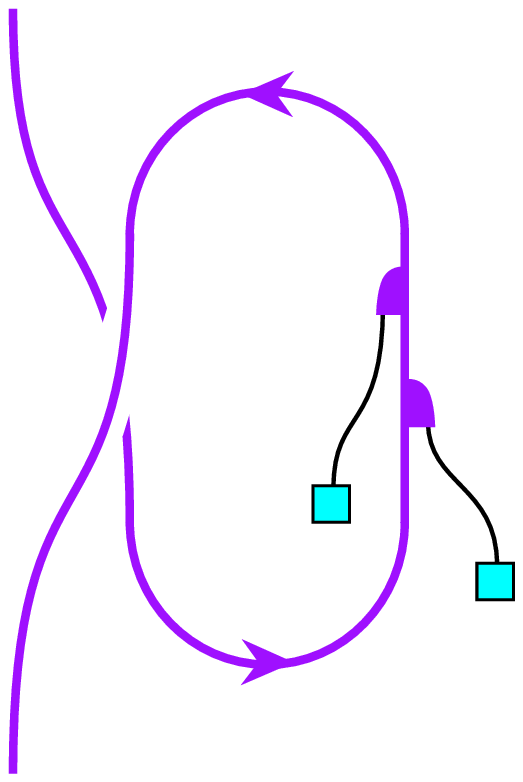}
  \put(-4.5,-8) {\sse$ \Xs $}
  \put(-3,88)   {\sse$ \Xs $}
  \put(28,28)   {\sse$ t $}
  \put(45,53.5) {\sse$ \rho_X^{} $}
  \put(48,42)   {\sse$ \ohr_X^{} $}
  \put(58,20)   {\sse$ t $}
  } }
respectively, with $t$ an invertible group-like element of $H$ obtained as
the product of the Drinfeld element $u$
\erf{def:uDrinfeld} and the inverse of the ribbon element of $H$,
  \be
  t = u\,v^{-1} .
  \ee
Using the formulas \erf{p8} and \erf{def_iHaa_X} for the dinatural families
$\iL$ and $\iK$ of the coends $L$ and $K$, the characters of $L$-modules 
$(U,\req_U^L)$ with $U \eq (U,\rho^H_{U}) \iN \HMod$ and those of $K$-modules 
$(X,\req_X^K)$ with $X \eq (X,\rho^H_{\!X},\ohr^H_{X}) \iN \HBimod$ -- as 
described, for the case of $K$, in \erf{Qrep2} -- can then be written as
  \eqpic{char-drin} {370} {33} {
  \put(0,41)     {$ \chii^L_U ~= $}
  \put(50,0) { \Includepichtft{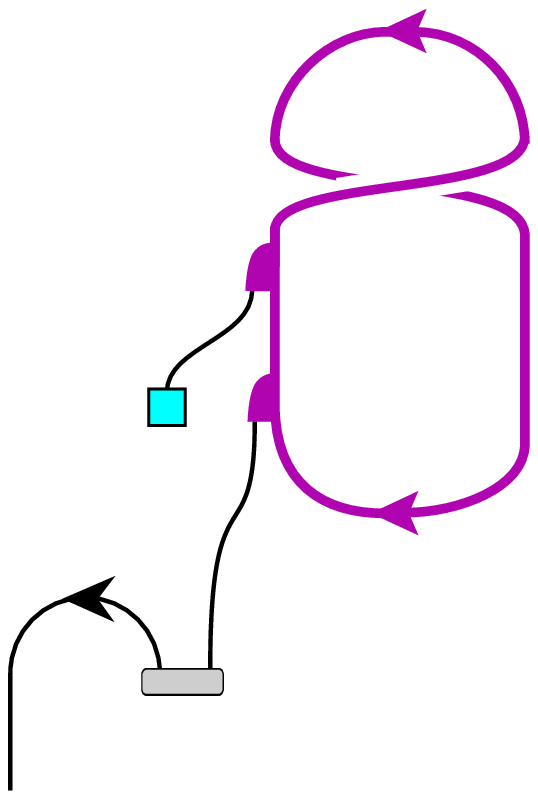} 
  \put(-4.4,-8.5){\sse$ \Hss $}
  \put(10.8,40.5){\sse $t$ }
  \put(16.2,3.8) {\sse $Q$ }
  \put(31,42.5)  {\sse$ \rho^H_{\!U} $}
  \put(31,56.5)  {\sse$ \rho^H_{\!U} $}
  \put(58,76)    {\sse$ U$ }
  \put(83,41)    {$ =~ \chii^H_{U} \circ m \circ (t \oti f_Q)
                   ~=~ \chii^H_{U} \circ m \circ (f_Q \oti t) $}
  } }
and as
  \eqpic{KX_char} {290} {37} {
  \put(-52,40)   {$ \chii^K_{\!X} ~= $}
  \put(0,2)  {\Includepichtft{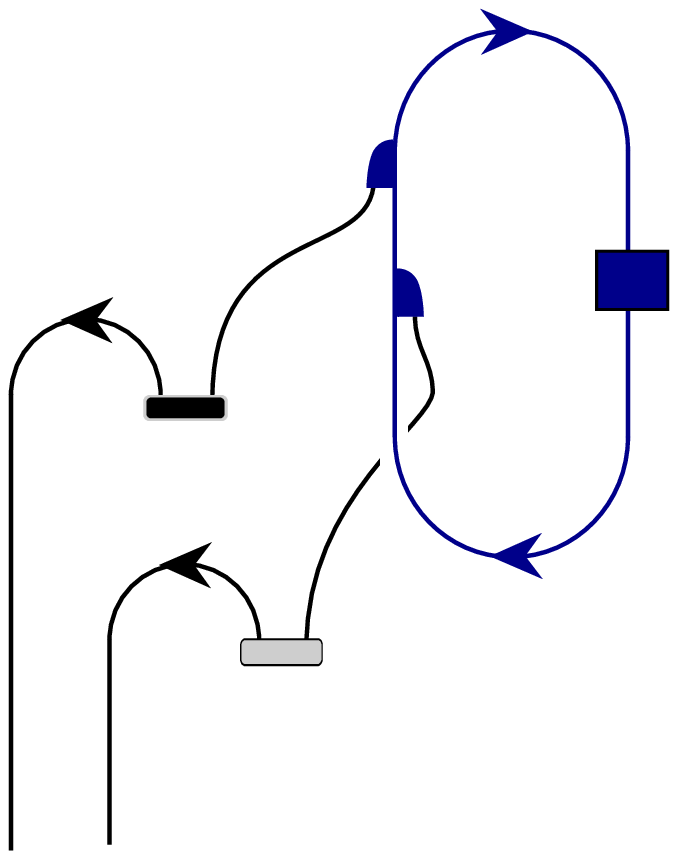}
  \put(12.3,40.7){\sse$Q^{-1}$}
  \put(26.7,13.3){\sse$Q$}
  \put(-6,-8.5)  {\sse$H^*$}
  \put(7,-8.5)   {\sse$H^*$}
  \put(29.3,75.2){\sse$ \rho^H_{\!X} $}
  \put(45.6,62.2){\sse$ \ohr^H_{\!X} $}
  \put(74,61.5)  {\sse $\pi_X^\HBimodsm $}
  }
  \put(111,40)   {$ =~ \chii^{H\otimes H\op}_X \circ (m \oti m) \circ
                       ( t \oti f_{Q^{-1}} \oti f_Q \oti t) \,, $}
  }
respectively, with $f_Q$ the Drinfeld map \erf{def-drin} and $f_{Q^{-1}}$ the 
analogous morphism in which the monodromy matrix $Q$ is replaced by its 
inverse. (In \erf{KX_char}, each of the two occurrences of the element $t$ in
$\pi_X^\HBimodsm$ can be treated analogously as the single $t$ in 
\erf{char-drin}; for details see Lemmas 6 and 8 of \cite{fuSs4}.)

The result \erf{KX_char} is actually a rather direct corollary of
\erf{char-drin}: the categories of $H\,{\otimes}\,H$-modules and of 
$H$-bimodules are ribbon equivalent (an equivalence functor has been given
explicitly in equation \erf{HHbimiso} above), and this equivalence maps the
$H\,{\otimes}\,H$-module $L \otik L$ and the $H$-bimodule $K$ to one another. 

\begin{rem}\label{rem:simple}
Since $H$ is by assumption factorizable, the Drinfeld map $f_Q$ is invertible. 
The group-like element $t$ is invertible as well. As a consequence the 
result \erf{char-drin} implies that the set $\mathcal X \eq 
\{ \chii_{S_i}^L \,|\, i\iN\I \}$ of characters is linearly independent and 
that the character of any $L$-mo\-du\-le of the form $(U,\req_U^L)$ is an 
integral linear combination of the characters in $\mathcal X$. It 
follows that the simple objects, up to isomorphisms, of the full monoidal 
subcategory $\HMod_L^{\mathcal Q}$ of $\HMod_L$ that consists of the modules 
$(U,\req_U^L)$ are precisely the modules $(S_i,\req_{S_i}^L)$ with 
$\{S_i\,|\,i\iN\I\}$ the simple $H$-modules. 
As a consequence, in the chiral decomposition \erf{X=sumij} the simple 
$L$-characters are $ \chii^L_k \eq \chii_{S_k}^L$ and the summation 
extends over the same index set $\I$ as the summation in e.g.\ \erf{Z=sumij}.
\end{rem}

Next we note that a \findim\ Hopf algebra $H$ in \Vectk\ carries a natural 
structure of a Frobenius algebra
and thus is in particular self-injective. According to \erf{A-bimod-char-2} the 
character of $H$ as the \emph{regular bimodule} (i.e., with regular left and 
right actions) over itself can thus be written as
  \be
  \chii_H^{H\otimes H\op}
  = \sum_{i,j\in\I} c_{i,j}^{} \, \chii_i^H \oti \chii_j^H 
  \label{H-bimod-char}
  \ee
with $c_{i,j}^{}$ the entries of the Cartan matrix of the category \HMod,
i.e.\ $c_{i,j}^{} \eq [\,P_i \,{:}\, S_j \,]$ is the multiplicity of the simple
$H$-module $S_j$ in the Jordan-H\"older series of the projective cover $P_i$
of the simple $H$-module $S_i$. Equivalently, $c_{i,j}^{}$ is the dimension of
the space of intertwiners between the projective covers,
$c_{i,j}^{} \eq \dimk(\HomH(P_i,P_j))$.
If $H$ is factorizable, then the \emph{coregular} bimodule \Fo\ (see Theorem 
\ref{thm:Fo}) is isomorphic to the regular bimodule, with an intertwiner
given by the \emph{Frobenius map}
  \be
  \Phi := ((\lambda\,{\circ}\, m) \oti \idHs) \circ (\apo \oti b_H^\ko) \,,
  \ee
and hence the character of \Fo\ decomposes like in \erf{H-bimod-char},
  \be
  \chii_F^{H\otimes H\op}
  = \sum_{i,j\in\I} c_{i,j}^{} \, \chii_i^H \oti \chii_j^H \,.
  \label{F-H-bimod-char}
  \ee

Now compose the equality \erf{F-H-bimod-char} with 
$(m \oti m) \cir (t \oti f_{Q^{-1}} \oti f_Q \oti t)$. Then by comparison with 
\erf{KX_char} we learn that
  \be
  \chii^K_X = \sum_{i,j\in\I} c_{i,j}\,
  \big[ \chii^H_i \,{\circ}\, m \,{\circ}\, (t \oti f_{Q^{-1}}) \big] \otimes
  \big[ \chii^H_j \,{\circ}\, m \,{\circ}\, (f_Q \oti t) \big] \,.
  \label{chiKF=cij..}
  \ee
Here the second tensor factor equals $\chii^L_j$ as given in \erf{char-drin}.
For the first factor, the presence of $f_{Q^{-1}}$ instead of $f_Q$ amounts
to replacing the braiding in $\req^L_{S_i}$ by its inverse, and thus 
according to Remark \ref{rem:dual} we deal with the $L$-character the dual 
module. We conclude that \cite[Thm.\,3]{fuSs4}
  \be
  \chii^K_F = \sum_{i,j\in\I} c_{i,j}\, \chii^L_{\overline i} \oti \chii^L_j
  = \sum_{i,j\in\I} c_{\overline i,j}^{}\, \chii^L_i \otimes \chii^L_j \,,
  \label{chiKX-chiL-chiL}
  \ee
where $\chii^L_i \eq \chii_{S_i}^L$ is the character of the simple 
$L$-module $(S_i,\req_{S_i}^L)$.
This is the desired chiral decomposition, of the form \erf{Z=sumij}.

\begin{rem}
(i)\,
By definition (see \erf{cij}) the numbers $c_{i,j}$ are non-negative integers. 
And they are naturally associated with the category $\C \,{\simeq}\, \HMod$ 
-- they depend only on \C\ as an abelian category.
\\[2pt]
(ii)\,
Among the simple objects of \C\ is in particular the tensor unit 
$\one \,{\cong}\, S_0$. In general, the corresponding diagonal coefficient
$c_{0,0}$ in \erf{chiKX-chiL-chiL} is larger than 1.
This is \emph{not} in conflict with the uniqueness of the vacuum -- it
just accounts for the fact that for non-semisimple \C\ the tensor unit
has non-trivial extensions and is in particular not projective. 
\\[2pt]
(iii)\,
The result \erf{chiKX-chiL-chiL} fits well with predictions for the bulk 
state space of concrete classes of logarithmic CFTs, namely \cite{garu2} 
the $(1,p)$ triplet models and \cite{qusc4} WZW models with
supergroup target spaces, compare Remark \ref{rem:squig}.
\end{rem}

\medskip

We refer to the character \erf{F-H-bimod-char} as the \emph{Cardy-Cartan modular
invariant}, because in the semisimple case, for which $c_{i,j} \eq \delta_{i,j}$,
the expression \erf{F-H-bimod-char} reduces to the charge conjugation modular 
invariant, which in the context of studying compatible conformally invariant 
boundary conditions of the CFT is also known as the ???Cardy case???.

Next we generalize the Cardy-Cartan modular invariant to the situation that
we perform a twist by a ribbon Hopf algebra automorphism of $H$. This is
achieved as follows.
First note that an automorphism $\omega$ of $H$ induces an endofunctor 
$G_\omega\colon \HMod \To \HMod$. If $\omega$ is a \emph{Hopf algebra 
automorphism}, i.e.\ both an algebra and a coalgebra automorphism and commuting
with the antipode, then the functor $G_\omega$ is rigid monoidal, and if 
$\omega$ is a \emph{ribbon Hopf algebra automorphism}, i.e.\ in addition 
satisfies
  \be
  (\omega\oti\omega)(R) = R \qquand \omega(v) = v \,, 
  \labl{omegaRv}
then $G_\omega$ is is even a ribbon functor. Given two automorphisms $\omega$ 
and $\omega'$, one has $ G_\omega\cir G_{\omega'} \eq G_{\omega\omega'}$, as a 
strict equality of functors. It follows that $G_{\omega}$ has 
$G_{\omega^{-1}_{}}$ as an inverse and is thus an equivalence of categories.
In particular, $\omega$ induces a bijection $\overline\omega$ from the index set
$\I$ to itself, in such a way that 
$ \{ S_{\overline\omega(i)} \,|\, i\,{\in}\,\I \} $ is again a full set of 
representatives of the isomorphism classes of simple $H$-modules.

With this information we are in a position to establish

\begin{thm}\label{thm:Fomega}
For $\omega$ a ribbon Hopf algebra automorphism of a factorizable Hopf algebra
$H$, the character of the automorphism-twisted coregular bimodule 
$ \Fomega\eq (\Hs,\rho_\Fo^{},\ohr^{}_\Fo\cir(\idHs\oti\omega))$
{\rm (}see {\rm \erf{def:Fomega}}{\rm)} has the chiral decomposition 
  \be
  \chii^K_ \Fomega = \sum_{i,j\in\I} c_{\overline i,\overline\omega(j)}^{}\,
  \chii^L_i \otimes \chii^L_j \,.
  \labl{pic5csp}
\end{thm}

\begin{proof}
Pictorially, \erf{chiKF=cij..} reads
  \eqpic{Z_graph} {330} {35} {
  \put(0,-2) {\Includepichtft{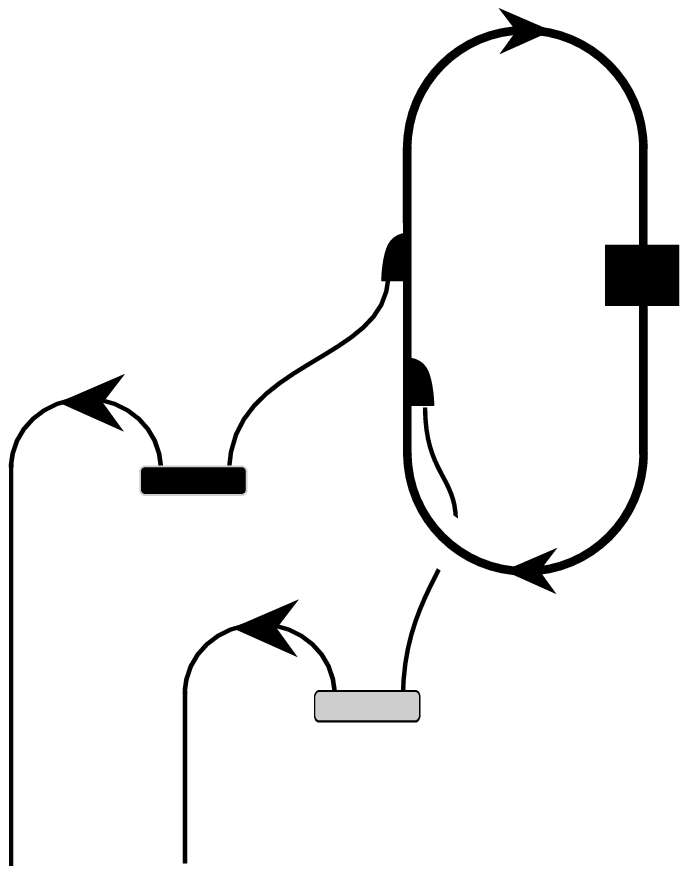}
  \put(-4.5,-8.8) {\sse $H^*$}
  \put(15,-8.8) {\sse $H^*$}
  \put(34,8.8)  {\sse $f_Q$}
  \put(10,35.7) {\sse $f_{Q^{-1}}$}
  \put(46,78)   {\sse $F$}
  \put(29,65)   {\sse $\rho_{F}^H$}
  \put(49,52)   {\sse $\ohr_{F}^H$}
  \put(76,64)   {\sse $\pi_F^\HBimodsm$}
  }
  \put(100,39)    {$\dsty = \quad \sum_{i,j\in\I}~ c_{\bar i,j}$}
  \put(180,-2) {\Includepichtft{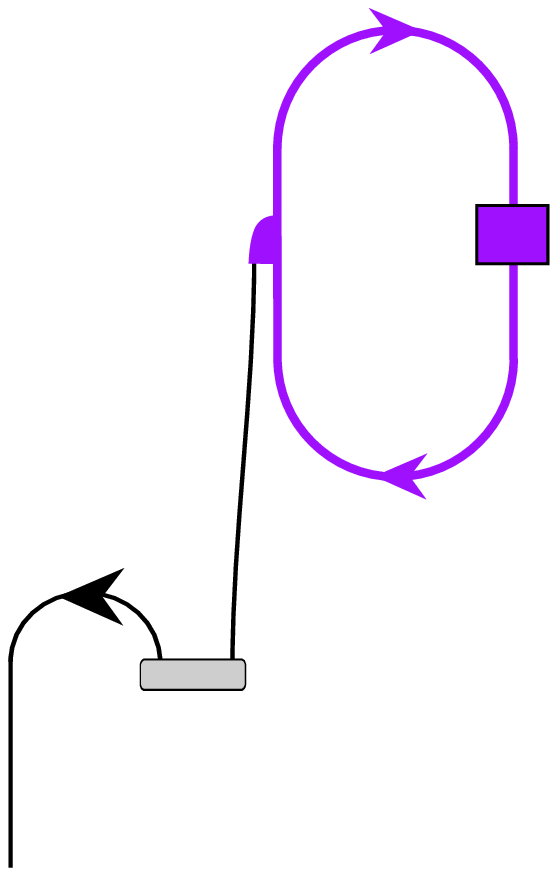}
  \put(-4,-8.8) {\sse $H^*$}
  \put(15.6,12) {\sse $f_Q$}
  \put(33,81)   {\sse $S_i$}
  \put(16,68)   {\sse $\rho_{S_i}^H$}
  \put(61,68.8) {\sse $\pi_{S_i}^\HModsm$}
  }
  \put(260,-2) {\Includepichtft{142a}
  \put(-4,-8.8) {\sse $H^*$}
  \put(15.6,12) {\sse $f_Q$}
  \put(33,81)   {\sse $S_j$}
  \put(16,68)   {\sse $\rho_{S_j}^H$}
  \put(61,68.8) {\sse $\pi_{S_j}^\HModsm$}
  } } 
Now compose this equality with $\id_\Hs \oti (\omega^{-1})^*$ and use that,
by the first equality in \erf{omegaRv}, $(\omega\oti\omega)(Q) \eq Q$, so that
the automorphism $\omega^{-1}$ can be pushed through the Drinfeld map on both
sides of the equality. This yields
  \be
  \chii^K_ \Fomega = \sum_{i,j\in\I} c_{\overline i,j}^{}\,
  \chii^L_i \otimes \chii^L_{S_j^{\omega^{-1}}}
  = \sum_{i,j\in\I} c_{\overline i,j}^{}\,
  \chii^L_i \otimes \chii^L_{S_{\overline\omega(j)^{-1}}} \,.
  \ee
A relabeling of the summation index $j$ then gives \erf{pic5csp}.
\end{proof}


\section{Correlation functions} \label{sec:corfus}

As already pointed out, the conjecture that the character $\chii^K_F$ gives 
the torus partition function of a full CFT with bulk state space $F \eq \Fo$ 
constitutes a special case of a proposal for general correlation functions 
$\Coro gn$ of bulk fields, for orientable world sheets of arbitrary genus $g$ 
and with an arbitrary number $n$ of insertions of the bulk state space. This 
proposal \cite{fuSs3,fuSs5} is based on the idea that it should be possible to 
express correlators entirely and very directly through the basic structures of 
their ingredients -- that is, the topology of the world sheet and the structure
of the bulk state space as a symmetric Frobenius algebra and as a module 
$(F,\req^K_F)$ over the bulk handle Hopf algebra.

Let us first see how this works in the case of the torus partition 
function. To this end we note the equalities
  \eqpic{piF-epsFetaF} {335} {38} { \setulen 70
    \put(0,0) {\includepichtftsm{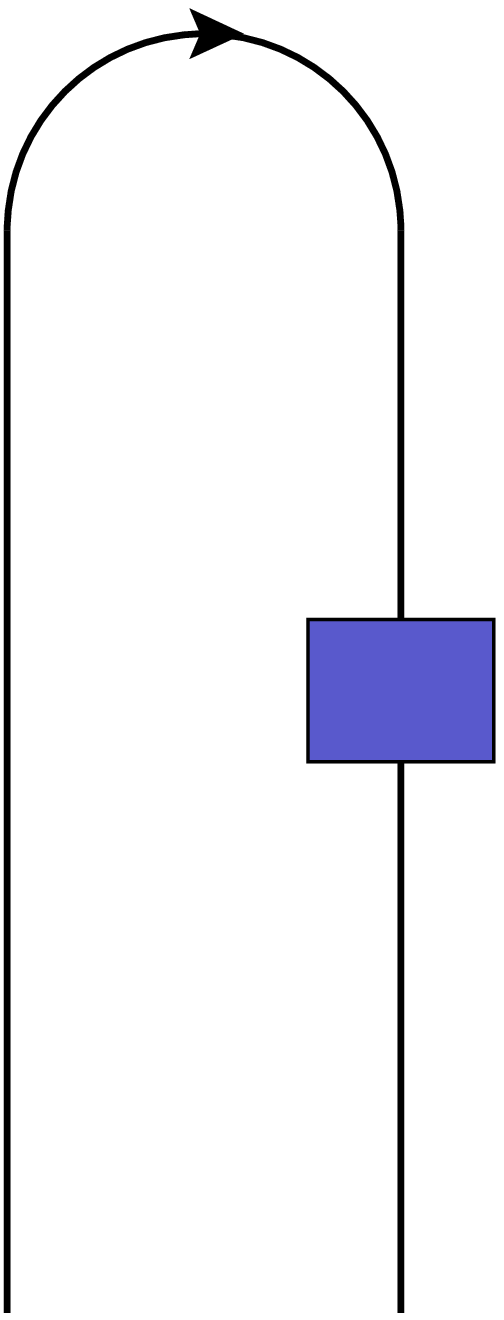}
  \put(-4.4,-12.2) {\sse$ F $}
  \put(38.8,-12.2) {\sse$ F $}
  \put(48.1,80.2)  {$ \pi_F^\HBimodsm $}
  }
  \put(119,61)      {$ = $}
    \put(165,0) {\includepichtftsm{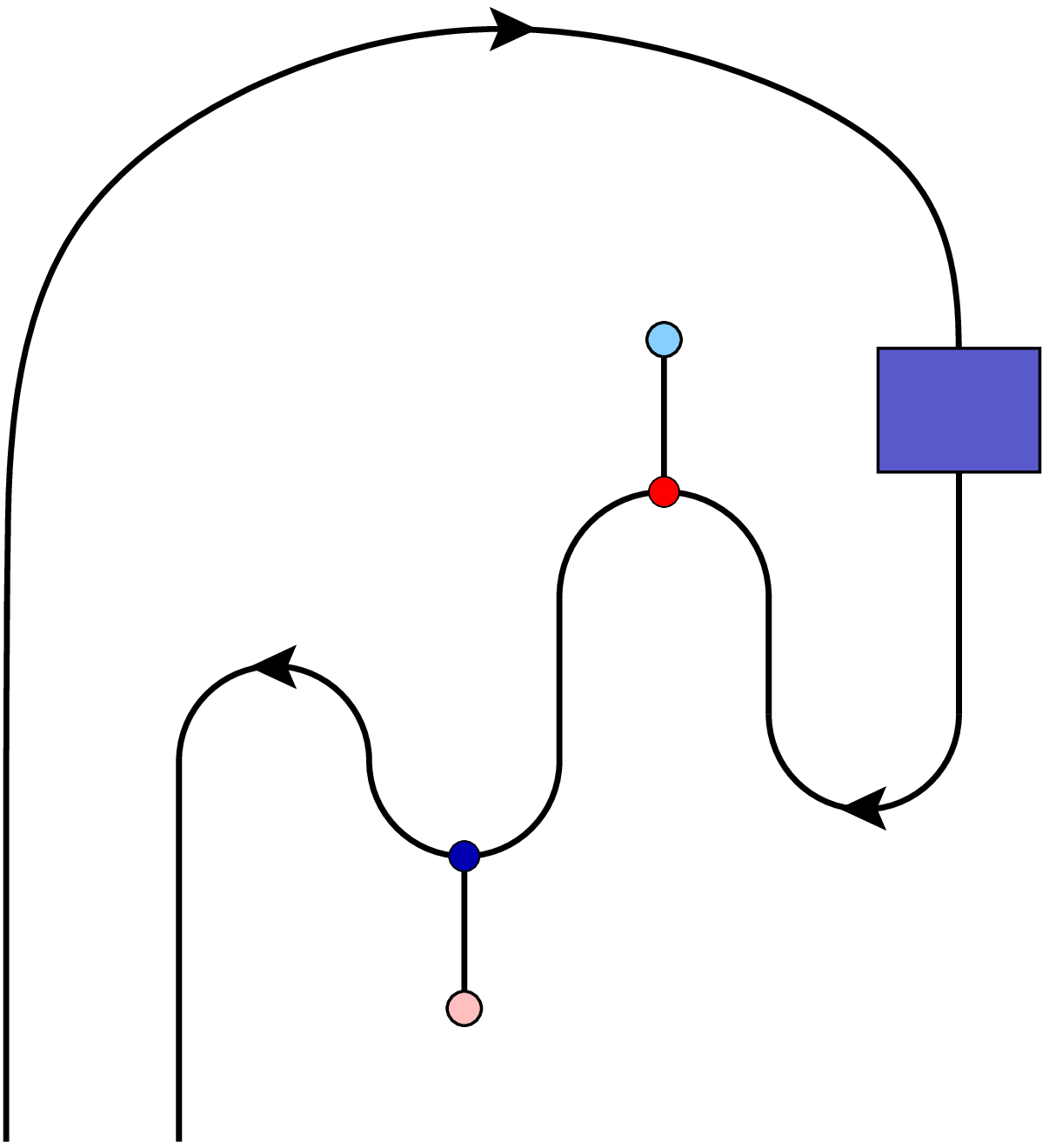}
  \put(-4.4,-12.2) {\sse$ F $}
  \put(17.6,-12.2) {\sse$ F $}
  \put(50.3,44.1)  {\sse$ \Delta_F $}
  \put(61.9,16.1)  {\sse$ \eta_F^{} $}
  \put(66.9,96.5)  {\sse$ \eps_F^{} $}
  \put(73.9,72.7)  {\sse$ m_F^{} $}
  \put(125,104.2)  {$ \pi_F^\HBimodsm $}
  }
  \put(358,61)      {$ = $}
    \put(400,0) {\includepichtftsm{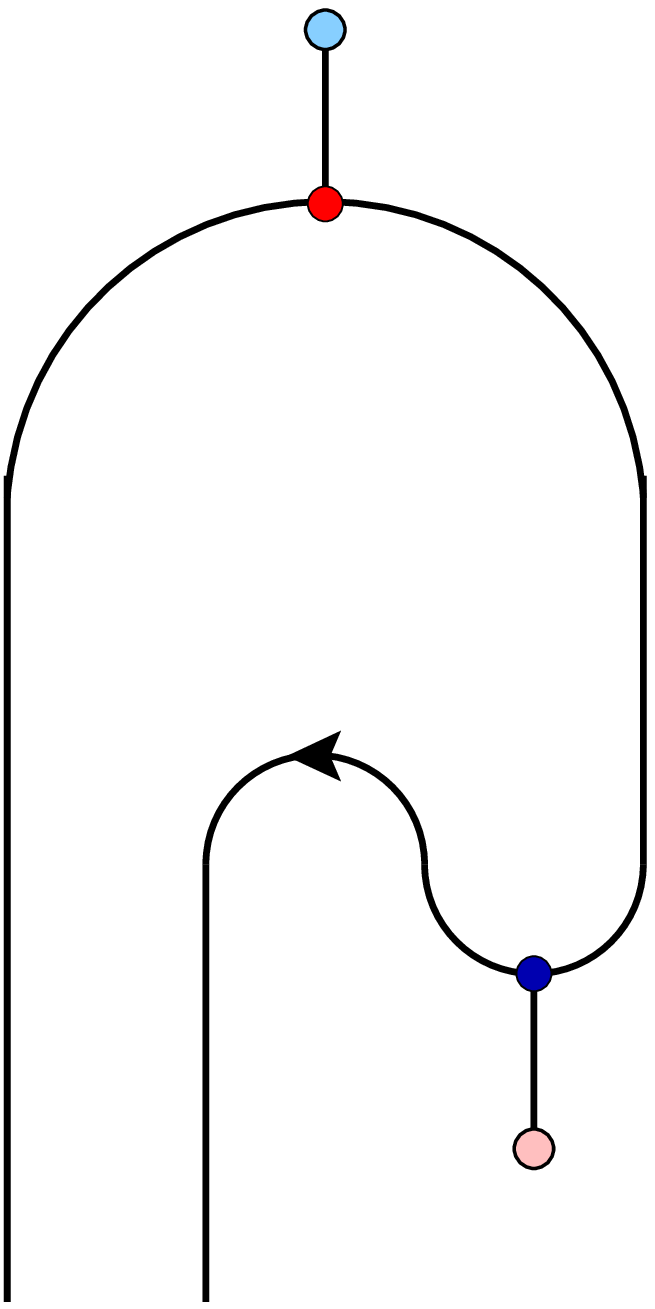}
  \put(-4.4,-12.2) {\sse$ F $}
  \put(17.7,-12.2) {\sse$ F $}
  } } 
where we first use the Frobenius property and then the symmetry of $F$.
The equality of the left and right hand sides of \erf{piF-epsFetaF} allows
us to rewrite the expression \erf{Qrep2} for $\chii^K_F$ as
  \eqpic{htft132p} {110} {44} { \setlength\unitlength{1.5pt} \put(0,-5){
  \put(0,36)      {$ \chii^K_F=~ $}
    \put(30,0) {\INcludepichtft{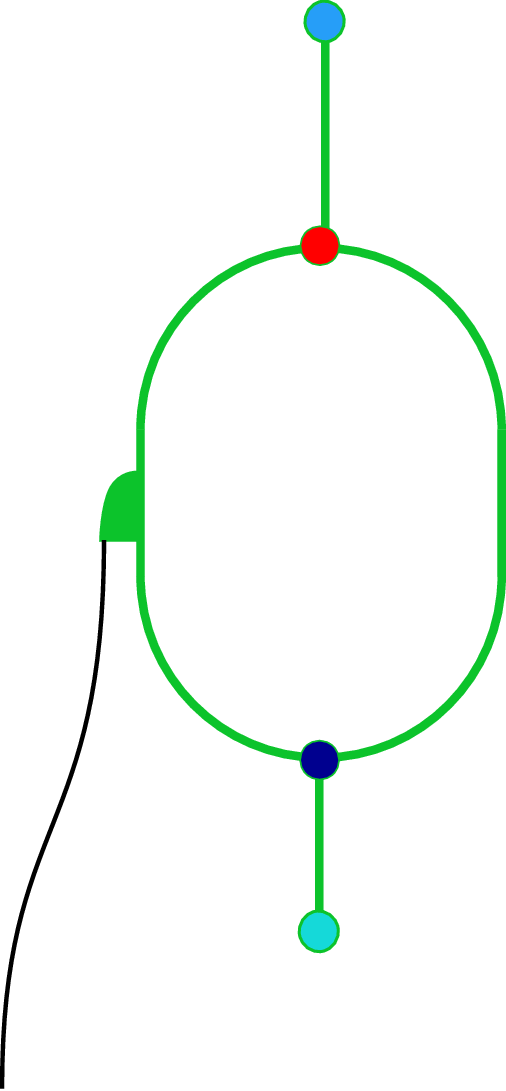}{38}
  \put(-2.9,-5.2)  {\sse$ K $}
  \put(38.3,42)    {\sse$ F $}
  \put(11.4,41)    {\sse$ \req^K_F $}
  \put(25.5,64.4)  {\sse$ m_F^{} $}
  \put(25.4,19.4)  {\sse$ \Delta_F $}
  \put(26.5,11.1)  {\sse$ \eta_F^{} $}
  \put(27.3,78.3)  {\sse$ \eps_F^{} $}
  } } }

Thus, basically the morphism $\chii^K_F$ consists of an $F$-`loop' combined with
the action of the handle Hopf algebra $K$. This can be seen as a manifestation 
of the fact that we deal with a world sheet having one handle.  
In a similar vein, for the correlator $\Coro g0$, i.e.\ the partition function 
of an orientable world sheet $\Sigma$ of arbitrary genus $g$, we are lead to the
following construction:

\def\leftmargini{1.49em}~\\[-2.52em]\begin{itemize}\addtolength{\itemsep}{-6pt}
    \item[\nx]
Select a skeleton $\Gamma$ for $\Sigma$ and label each edge of the skeleton
by the Frobenius algebra $F$. 
    \item[\nx]
Orient the edges of $\Gamma$ in such a manner that each vertex of $\Gamma$
has either one incoming and two outgoing edges or vice versa.
Label each of these three-valent vertices either with the coproduct 
$\Delta_F$ of $F$ or with the product $m_F$,
depending on whether one or two of its three incident edges are incoming.
    \item[\nx]
To avoid having to introduce any duality morphisms (analogously as in the
description \erf{Qrep2} of $\Cor 10 \eq \chii^K_F$), when implementing the 
previous part of the construction allow for adding further edges that 
connect one three-valent and one uni-valent vertex, the latter being 
labeled by the unit $\eta_F$ or counit $\eps_F$ of $F$.
    \item[\nx]
For each handle of $\Sigma$ attach one further edge,
labeled by the handle Hopf algebra $K$, to the corresponding
loop of the skeleton, and label the resulting new trivalent vertex
by the \rep\ morphism $\req^K_F$.
    \item[\nx]
The so obtained graph defines a morphism in $\HomCC(K^{\otimes g},\one)$. 
\end{itemize}

At genus $g\eq1$ this prescription precisely reproduces the morphisms
\erf{htft132p} in  $\HomCC(K,\one)$.
At higher genus several different choices for the skeleton $\Gamma$ are 
possible, but with the help of the symmetry and Frobenius property of $F$ 
one sees that they all yield one and the same morphism in \CbC, namely
  \eqpic{htft132o} {180} {108} { \put(0,20){
  \put(0,101)      {$\Cor g0 =~$}
    \put(60,0) {\Includepichtft{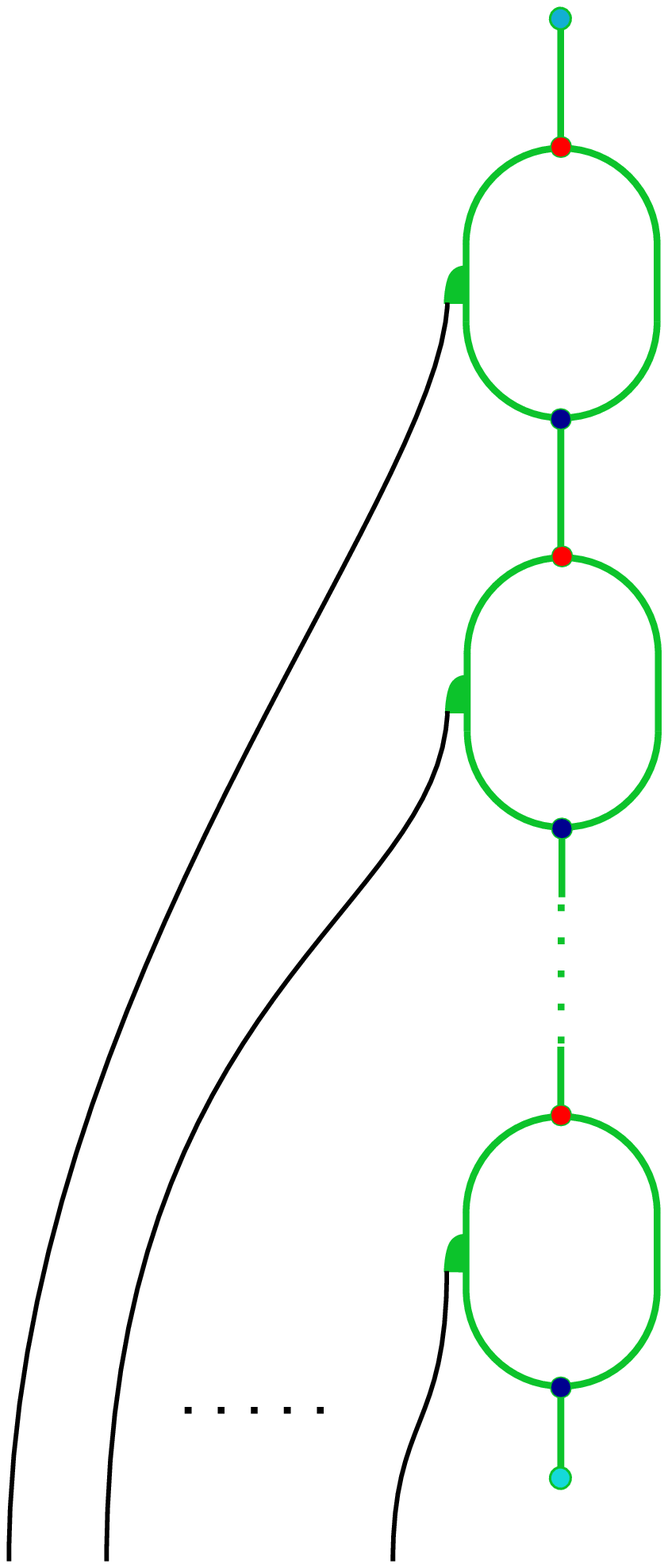}
  \put(-8.7,-12)   {\sse$ \underbrace{\hspace*{7.1em}}_{g~ \rm factors} $}
  \put(-5,-9.2)    {\sse$ K $}
  \put(8.8,-9.2)   {\sse$ K $}
  \put(27,-9.2)    {\sse$ \cdots $}
  \put(49.2,-9.2)  {\sse$ K $}
  \put(91.6,120)   {\sse$ F $}
  \put(65,176)     {\sse$ \req^K_F $}
  \put(78.9,198.3) {\sse$ m_F^{} $}
  \put(79.8,151)   {\sse$ \Delta_F $}
  \put(80.1,11.3)  {\sse$ \eta_F^{} $}
  \put(80.9,214)   {\sse$ \eps_F^{} $}
  } } }

Our ansatz generalizes easily to world sheets with bulk field insertions: For 
$n$ outgoing (say) insertions of the bulk state space, just replace the counit 
$\eps_F \iN \HomCC(F,\one)$ in \erf{htft132o} with an $n$-fold coproduct 
$\Delta_F^{(n)} \iN \HomCC(F,F^{\otimes n})$.
When doing so, the order of taking coproducts is immaterial owing to 
coassociativity of $\Delta_F$, and the order of factors in $F^{\otimes n}$
does not matter due to cocommutativity of $\Delta_F$; with any choice of
ordering, the resulting morphism in $\HomCC(K^{\otimes g},F^{\otimes n})$ equals
  \eqpic{Sk_morph} {180} {142} { \put(0,13){
  \put(-19,133)      {$\Cor gn(F) ~=$}
    \put(60,0) {\Includepichtft{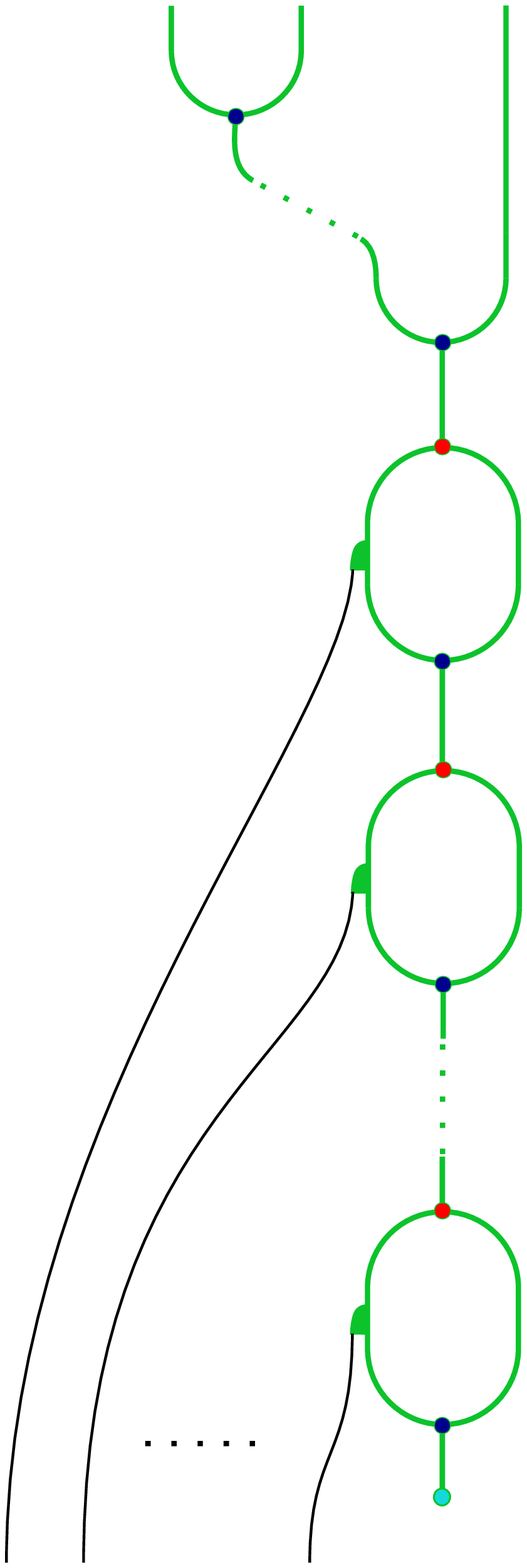}
  \put(-8.7,-12)   {\sse$ \underbrace{\hspace*{7.1em}}_{g~ \rm factors} $}
  \put(28.3,288)   {\sse$ \overbrace{\hspace*{7.6em}}^{n~ \rm factors} $}
  \put(-5,-9.2)    {\sse$ K $}
  \put(8.8,-9.2)   {\sse$ K $}
  \put(27,-9.2)    {\sse$ \cdots $}
  \put(49.2,-9.2)  {\sse$ K $}
  \put(27.2,278)   {\sse$ F $}
  \put(49.5,278)   {\sse$ F $}
  \put(66,278)     {\sse$ \cdots $}
  \put(86,278)     {\sse$ F $}
  \put(65,176)     {\sse$ \req^K_F $}
  \put(78.9,198.3) {\sse$ m_F^{} $}
  \put(71.1,219.3) {\sse$ \Delta_F $}
  \put(79.8,12.2)  {\sse$ \eta_F^{} $}
  } } }
for $n \,{>}\, 0$.

Likewise one can generalize the ansatz to correlators $\Corro gpq$ with any 
numbers $q$ of incoming and $p$ of outgoing insertions. The incoming insertions 
are incorporated by replacing the unit $\eta_F \iN \HomCC(\one,F)$ in 
\erf{htft132o} with a $q$-fold 
product $m_F^{(q)} \iN \HomCC(F^{\otimes q},F)$.
Furthermore, the case of genus zero is included by just omitting the 
$F$-`loop'. Altogether the prescription can be summarized as follows:
  \be
  \begin{array}{l}
  \Corr 011 := \mbox{\sl id}_F \,, 
  \Nxl4
  \Corr 111 := m_F \circ (\rho^K_F \,{\otimes}\, \mbox{\sl id}_F)
    \circ (\mbox{\sl id}_K \,{\otimes}\, \Delta_F) \,,
  \Nxl4
  \Corro g11 := \Corr 111 \circ (\mbox{\sl id}_K \,{\otimes}\,
    \mathrm{Cor}_{g-1;1,1}) \quad~ {\rm for}~~ g\,{>}\,1 \,,
  \Nxl4
  \Corro gpq := \Delta_F^{(p)} \circ \Corr g11 \circ
  \big( \mbox{\sl id}^{}_{K^{\otimes g}_{}} \,{\otimes}\, m_F^{(q)} \big) \,.
  \end{array}
  \labl{corrs}

\begin{rem} ~\\[2pt]
(i)\, One may be tempted to work with ribbons instead of with edges. But since
$F$ has trivial twist, $\theta_F \eq \id_F$,
the framing does not matter and can be neglected in our discussion.
\\[2pt]
(ii)\, Our ansatz results from the description \erf{htft132p} of the torus 
partition function $\Corro 100$ and the knowledge that $\Corro gpq$ must be 
an element of the morphism space $\HomCC(K^{\otimes g}\oti F^{\otimes q},
F^{\otimes p})$. It would be much more elegant to derive the prescription 
from a three-dimensional approach, which in the case of rational CFT should be 
related by a kind of folding trick to the TFT construction of \cite{fuRs}.
\end{rem}

What enters in the expressions for correlation functions above is only the 
structure of \C\ as a factorizable finite tensor category and of $F$ as a bulk 
state space, carrying the structure of a Frobenius algebra that is commutative 
and symmetric and has trivial twist. Again we can be more explicit for the case
that \C\ is equivalent to the category \HMod\ of \findim\ modules over
some factorizable Hopf algebra $H$ and that $F \eq \Fomega$ for any ribbon 
automorphism $\omega$ of $H$. 
Let us present the correlator $\Corro gpq$ for the case that $p \eq q \eq 1$, 
the extension to $p,q\,{>}\,1$ being easy, and first take $F$ to be the 
coregular $H$-bimodule \Fo. Then by inserting the expressions 
\ref{pic-Hb-Frobalgebra} for the structural morphisms of the Frobenius algebra 
\Fo\ and writing out the braiding of the category 
$\HMod^{\rm rev} \,{\boxtimes}\, \HMod \,{\simeq}\, \HBimod$ (which appears in 
the representation morphism $\req^K_F$), after a few rearrangements one obtains
  \eqpic{CorrgnH} {300} {118} { \put(0,18){ \setulen90
  \put(-16,114)    {$\Corr g11(\Fo) ~=$}
  \put(97,0) {\INcludepichtft{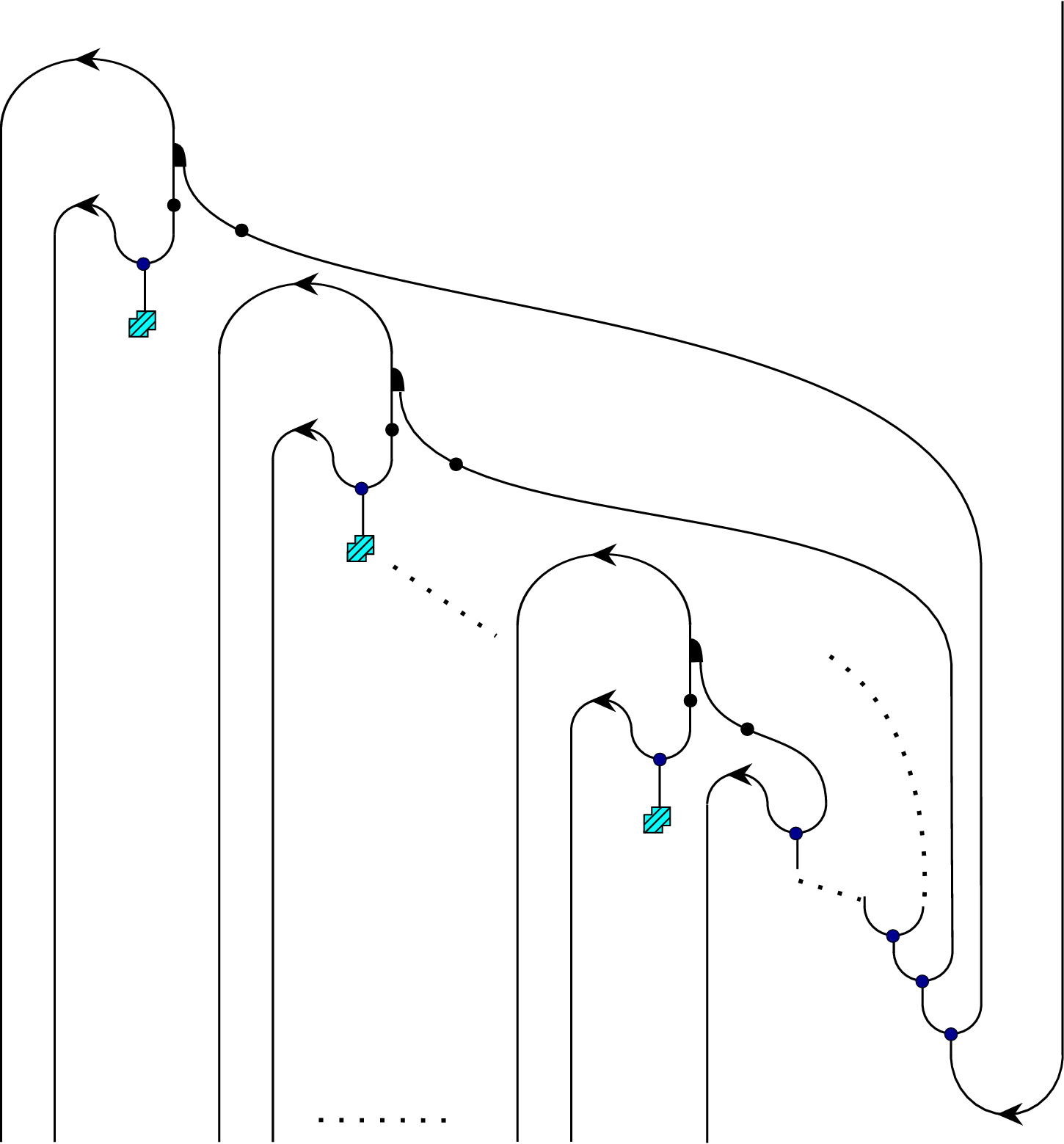}{342}
  \put(-3.9,-13)   {$ \underbrace{\hspace*{10.5em}}
                    _{g~ {\rm factors~of}~\Hs{\otimes}\Hs} $}
  \put(-5,-9.6)    {\sse$ \Hss $}
  \put(8.2,-9.6)   {\sse$ \Hss $}
  \put(44.5,-9.6)  {\sse$ \Hss $}
  \put(56.5,-9.6)  {\sse$ \Hss $}
  \put(76.5,-9.6)  {\sse$ \dots\dots$}
  \put(111.5,-9.6) {\sse$ \Hss $}
  \put(123.5,-9.6) {\sse$ \Hss $}
  \put(152.1,-9.6) {\sse$ \Hss $}
  \put(233,260)    {\sse$ \Hss $}
  \put(36,180.5)   {\sse$ \Lambda$}
  \put(70,130.5)   {\sse$ \Lambda$}
  \put(137,69.5)   {\sse$ \Lambda$}
  \put(43,220)     {\sse$ \ohrad $}
  \put(91.4,169.4) {\sse$ \ohrad $}
  \put(158.2,108.5){\sse$ \ohrad $}
  } } }
with $\ohrad$ the right-adjoint action of $H$ on itself.

For general $\Fomega$ the result differs from \erf{CorrgnH} only by a few
occurences of the automorphism $\omega$ (recall formula \erf{def:Fomega} 
and that the structural morphisms of the
Frobenius algebra $\Fomega$ coincide with those of \Fo\ as linear maps):
  \eqpic{CorrgnHw} {300} {110} {\setulen90
  \put(-16,120) {$\Corr g11(\Fomega) ~= $}
  \put(90,0) { \INcludepichtft{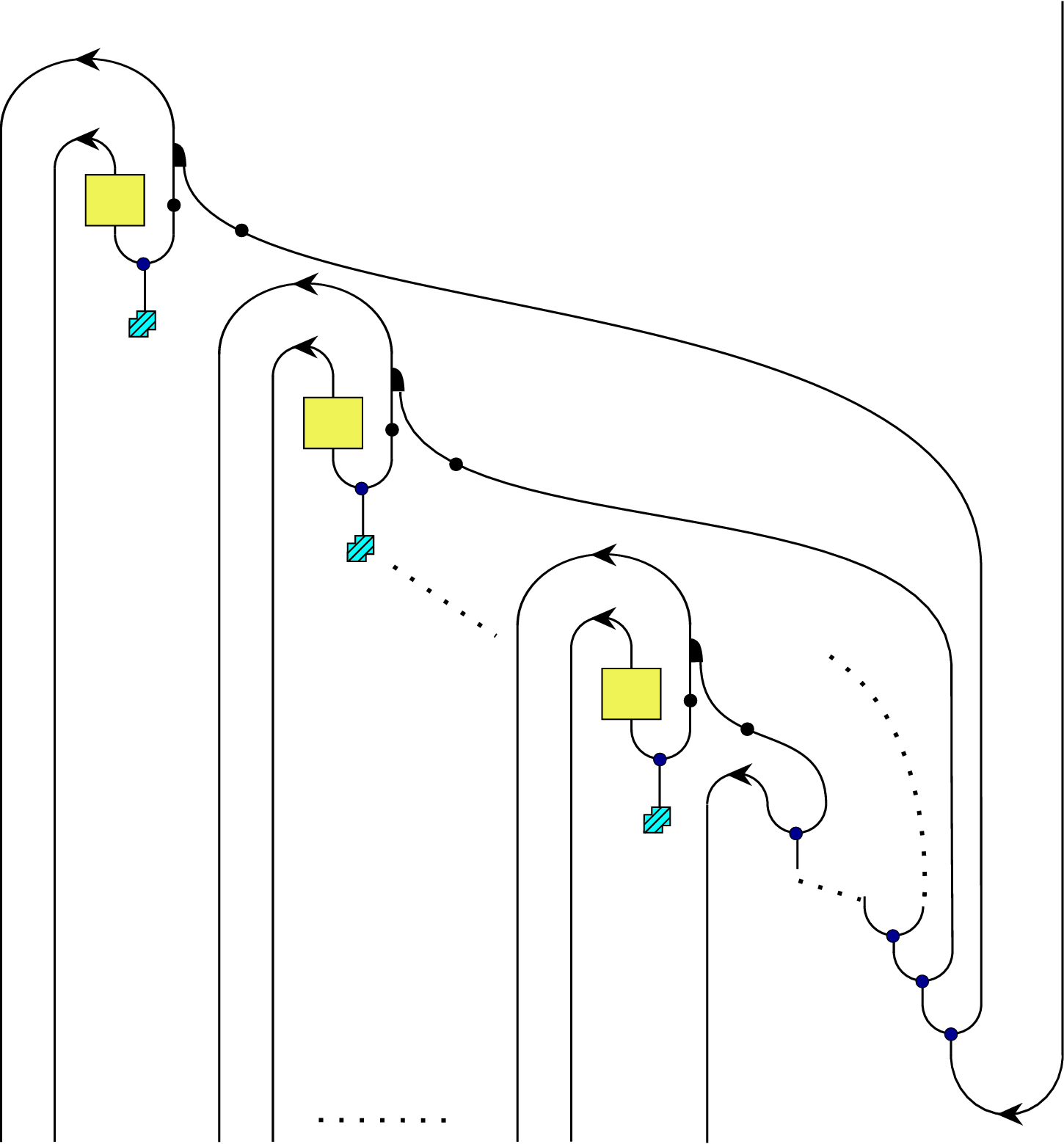}{342}
  \put(-5,-9.6)    {\sse$ \Hss $}
  \put(8.2,-9.6)   {\sse$ \Hss $}
  \put(44.5,-9.6)  {\sse$ \Hss $}
  \put(56.5,-9.6)  {\sse$ \Hss $}
  \put(76.5,-9.6)  {\sse$ \dots\dots$}
  \put(111.5,-9.6) {\sse$ \Hss $}
  \put(123.5,-9.6) {\sse$ \Hss $}
  \put(152.5,-9.6) {\sse$ \Hss $}
  \put(234,260)    {\sse$ \Hss $}
  \put(36,180.5)   {\sse$ \Lambda$}
  \put(70,130.5)   {\sse$ \Lambda$}
  \put(137,69.5)   {\sse$ \Lambda$}
  \put(19.5,208.5) {\sse$ \omega^{\!\!-1}$}
  \put(68.7,158.5) {\sse$ \omega^{\!\!-1}$}
  \put(135.3,97.5) {\sse$ \omega^{\!\!-1}$}
  \put(43,220)     {\sse$ \ohrad $}
  \put(91.4,169.4) {\sse$ \ohrad $}
  \put(158.2,108.5){\sse$ \ohrad $}
  } }

\medskip

The correlators of a rational conformal field theory must be invariant under an 
action of the mapping class group \Mapgn\ of closed oriented surfaces of genus 
$g$ with $n$ boundary components, where $g$ is the genus of the world sheet and 
$n \eq p\,{+}\,q$ is the number of (incoming plus outgoing) field insertions. 
One expects that this can still be consistently imposed for logarithmic CFTs.
And indeed we are able to establish mapping class group invariance of the
ansatz for correlators that we presented above, i.e. of the morphisms 
\erf{corrs} for the case that $\C \,{\simeq}\, \HMod$ and $F \eq \Fomega$.

Recall that the morphism $\Corro gpq$ is an element of the space
$\HomCC(K^{\otimes g}\oti F^{\otimes q},F^{\otimes p})$. A natural action 
\pigppq\ of \Mapgppq\ on this morphism space has been found in 
\cite{lyub8,lyub11}. \pigppq\ is described in some detail in Appendix 
\ref{app:mpg}; here we just note that \Mapgn\ is generated by suitable Dehn 
twists, that most of them are represented by pre-composing with an endomorphim 
of $K^{\otimes g}$ or $F^{\otimes q}$ or post-composing with an
endomorphim of $F^{\otimes p}$, and that at genus 1 the relevant 
endomorphims of $K \eq \Hs \otik \Hs$ are 
  \eqpic{S_KH-TKH} {410}{45} { \put(0,4){
  \put(0,39)    {$ \SK ~= $}
  \put(50,-6) { \Includepichtft{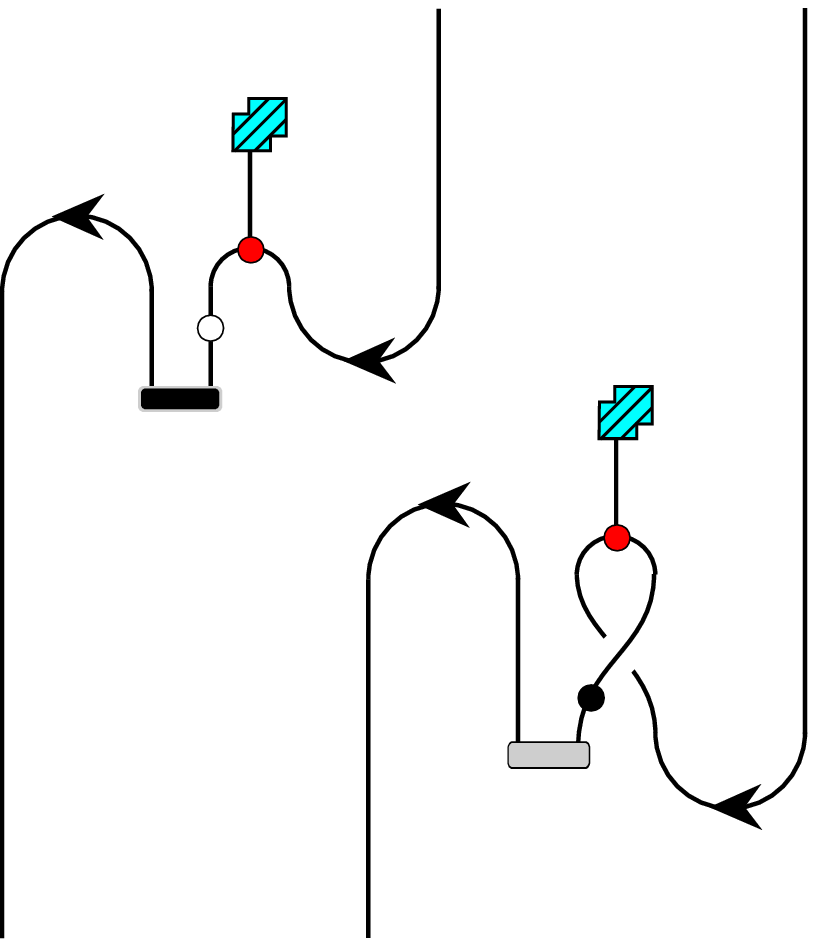}
  \put(-4.5,-9.2) {\sse$ \Hss $}
  \put(13,51.4)   {\sse$ Q^{-1} $}
  \put(36.1,-9.2) {\sse$ \Hss $}
  \put(32.5,87)   {\sse$ \lambda $}
  \put(44.6,106)  {\sse$ \Hss $}
  \put(57.1,11.8) {\sse$ Q $}
  \put(73.5,56)   {\sse$ \lambda $}
  \put(85.6,106)  {\sse$ \Hss $}
  }
  \put(180,39)    {and}
  \put(230,39)    {$ \TK ~= $}
  \put(278,0)  { \Includepichtft{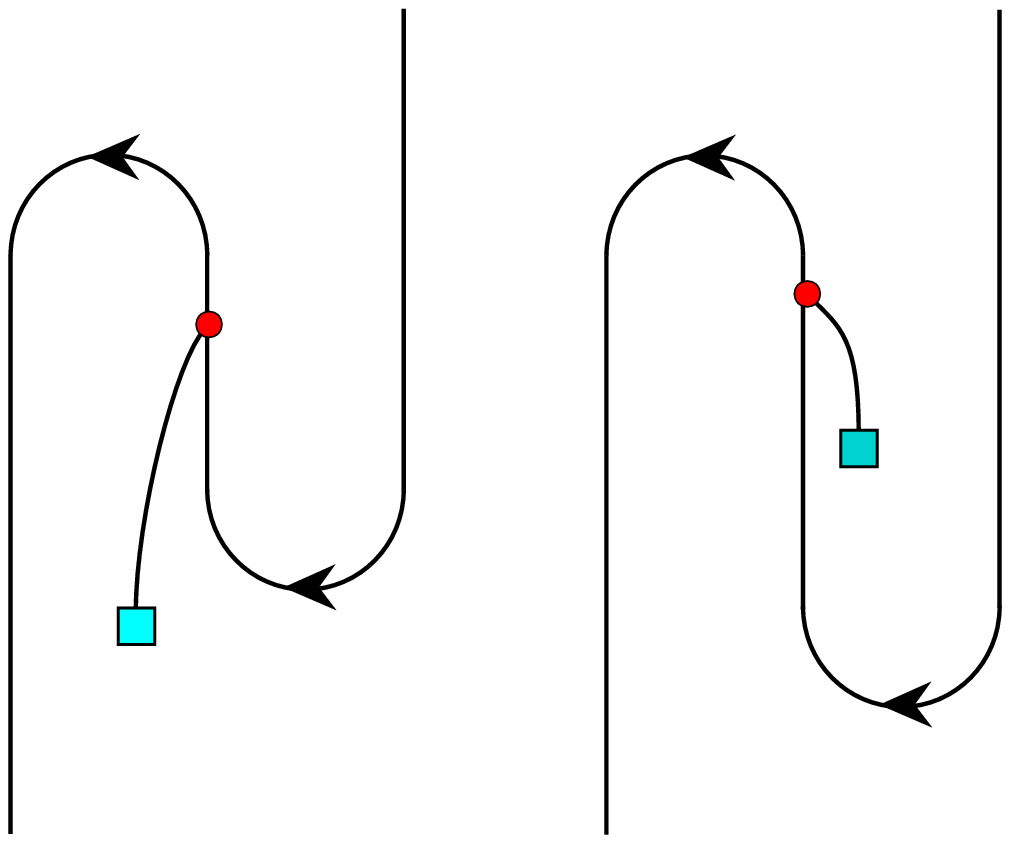}
  \put(-4.4,-9.2) {\sse$ \Hss $}
  \put(12,15)     {\sse$ v $}
  \put(39.6,94)   {\sse$ \Hss $}
  \put(61.2,-9.2) {\sse$ \Hss $}
  \put(89.7,34.5) {\sse$ v^{-1} $}
  \put(105.3,94)  {\sse$ \Hss $}
  } } }
which amount to an S- and T-transformation, respectively.

\begin{rem}
In general, the mapping class group action considered in \cite{lyub8,lyub11} 
is only projective. But owing to the fact that the category relevant to us is 
an enveloping category \CbC, with a ribbon structure in which the two factors 
are treated in an opposite fashion, in the situation at hand the action is in 
fact a genuine linear representation (compare Remark 5.5 of \cite{fuSs3}).
\end{rem}

Denote by \Mapgpq\ the subgroup of \Mapgppq\ that leaves the subsets of 
incoming and outgoing insertions separately invariant, and by \pigpq\ the 
\rep\ of \Mapgpq\ that is obtained (compare \erf{piXYgpq}) from \pigppq.
The following is the main result of \cite{fuSs5} (Theorem 3.2,
Remark 3.3 and Theorem 6.7):

\begin{thm}\label{thm:Mapgppq}
For $H$ a factorizable ribbon Hopf algebra and $\omega$ a ribbon automorphism 
of $H$, and for any triple of integers $g,p,q \,{\ge}\,0$, the morphism 
$\Corro gpq(\Fomega)$ is invariant under the action \pigpq\ of the group \Mapgpq.
\end{thm} 

\begin{rem}
Besides invariance under the action of \pigpq, the other decisive property
of correlation functions is compatibility with sewing. 
That is, there are sewing relations at the level of chiral blocks, and 
the correlators must be such that the image of a correlator,
as a specific vector in a space of chiral blocks, under these given chiral 
relations, is again a correlator. Compatibility with sewing allows one to 
construct all correlation functions by starting from a small set of 
fundamental corrleators and thereby amounts to a kind of locality property.
 \\
Like mapping class group invariance, compatibility with sewing is a 
requirement in rational CFT, and again one expects that one can consistently 
demand it for logarithmic CFTs as well. For now, checking compatibility of our 
ansatz with sewing is still open.
It is in fact fair to say that already for rational CFT the study of
sewing \cite{fjfrs,fjfs} still involves some brute-force arguments.
Understanding sewing in a way suitable for logarithmic CFT may require
(or, amount to) deeper insight into the nature of sewing.
\end{rem}

\newpage

\appendix \section{Appendix}

\subsection{Coends} \label{Acoend}

For a category \C, the \emph{opposite} category $\C\op$ is the one with
the same objects, but reversed morphisms, i.e.\ a morphism $f\colon U\To V$
in \C\ is taken to be a morphism $V\To U$ in $\C\op$. Given \ko-li\-ne\-ar 
abelian categories \C\ and \D\ and a functor $G$ from \CopC\
to \D, a \emph{dinatural transformation} from $G$ to an object $B\,{\in}\,\D$ 
is a family $\varphi \,{=}\, \{ \varphi_U\colon G(U,U)\To B \}_{\!U\in\C}^{}$ 
of morphisms with the property that the square
  \bee4010{
  \xymatrix @R+8pt{ & G(V,U) \ar^{G(\idsm_V,f)}[dr]\ar_{G(f,\idsm_U)}[dl]\\
  ~~~~G(U,U) \ar_{\varphi_U^{}}[dr] && G(V,V)~~~~ \ar^{\varphi_V^{}}[dl] \\ & B\,
  } }
of morphisms commutes for all $f\,{\in}\,\Hom(U,V)$.

A \emph{coend} $(D,\iota)$ for the functor $G$ is an initial object among
all such dinatural transformations, that is, it is an object $D\,{\in}\,\D$ 
together with a dinatural transformation $\iota$ such that for any dinatural 
transformation $\varphi$ from $G$ to any $B\iN\D$ there exists a unique morphism
$\kappa \,{\in}\, \Hom_\D(D,B)$ such that $\varphi_U^{} \,{=}\, \kappa \cir 
\iota_U^{}$ for every object $U$ of \C. In other words, given a diagram
  \bee4010{
  \xymatrix @R+8pt{ & G(V,U) \ar^{G(\idsm_V,f)}[dr]\ar_{G(f,\idsm_U)}[dl]\\
  ~~G(U,U) \ar@/_2pc/_{\varphi_U^{}}[ddr] \ar_{\iota_U^{}}[dr] &&
  G(V,V)~~ \ar@/^2pc/^{\varphi_V^{}}[ddl] \ar^{\iota_V^{}}[dl]
  \\ & D\, \ar@{-->}^\kappa[d] \\ & B\,
  } }
with commuting inner and outer squares for any morphism $f \iN \HomC(U,V)$, 
there exists a unique morphism $\kappa$
such that also the triangles in the diagram commute for all $U,V \iN \C$.

If the coend exists, then it is unique up to unique isomorphism. The 
underlying object, which by abuse of terminology is referred to as the coend of
$G$ as well, is denoted by an integral sign,
  \be
  D = \coen {U\in\C} G(U,U) \,.
  \ee
The finiteness properties of the categories we are working with in this paper 
guarantee the existence of all coends we need. Specifically, the bulk state 
space \Fo\ is the coend \erf{coendCC} of the functor 
$G^\C_\boxtimes\colon \C\op\Times\C \To \CbC$ that acts on objects as 
$(U,V)\,{\mapsto}\,U^\vee\boti V$, while the chiral and full handle Hopf 
algebras $L$ and $K$ are the coends \erf{defLK} of the functors 
$G^\C_\otimes\colon \C\op\Times\C \To \C$ and
$G^\CbC_\otimes\colon (\CbC)\op\Times(\CbC) \To \CbC$ that act on objects as
$(U,V)\,{\mapsto}\,U^\vee\oti V \iN \C$ and as
$(X,Y)\,{\mapsto}\,X^\vee\oti Y \iN \CbC$, respectively.

If the category \C\ is cocomplete, then an equivalent description of the coend of
$G$ (see e.g.\ section V.1 of \cite{MAy4}) is as the coequalizer of the morphisms
  \be
  \xymatrix{
  {\dsty\coprod_{f\colon V\to W}} G(V,W)\,\, \ar@<3.5pt>[r]^{~~s}
  \ar@<-2.5pt>[r]_{~~t} & \,{\dsty\coprod_{U\in\,\C} G(U,U)}
  } \labl{coendalizer}
whose restrictions to the `$f$th summand' are $s_f \eq F(f,\id)$ and 
$t_f \eq F(\id,f)$, respectively. Thus, morally, the coend of $G$ is the
universal quotient of $\coprod_U G(U,U)$ that enforces the two possible
actions of $G$ on any morphism $f$ in \C\ to coincide.


\subsection{The full center of an algebra} \label{Acenter}

Given an object $U$ of a monoidal category \C, a \emph{half-braiding}
$z \eq z(U)$ on $U$ is a natural family of isomorphisms $z_V\colon U \oti V 
\To V \oti U$, for all $V \iN \C$, such that (assuming \C\ to be strict) 
$z_\one \eq \id_U$ and
  \be
  (\id_V \oti z_W) \circ (z_V \oti \id_W) = z_{V\otimes W}
  \ee
for all $V,W \iN \C$. The \emph{monoidal center} $\Z\C$ is the category which
has as objects pairs $(U,z)$ consisting of objects of \C\ and of 
half-braidings, while its morphisms are morphisms $f$ of \C\ that are naturally 
compatible with the half-braidings of the source and target of $f$. The 
category $\Z\C$ is again monoidal, with tensor product
  \be
  (U,z) \otimes (U',z')
  := (U\oti U', (z_V{\otimes}\id_{U'})\cir(\id_U{\otimes}z'_V) )
  \ee
and tensor unit $(\one,\id)$, and with respect to this tensor product it is 
braided, with braiding isomorphisms $c_{(U,z),(U',z')} \,{:=}\, z_{U'}$.
The forgetful functor $\FF$ from $\Z\C$ to \C, acting on objects as
  \be
  \FF: \quad (U,z)\,\longmapsto\,U \,,
  \ee
is faithful (but in general neither full nor essentially surjective) and 
monoidal.

For $A \eq (A,m,\eta)$ a (unital, associative) algebra in \C, we say that an 
object $(U,z)$ of $\Z\C$ together with a morphism $r\iN \HomC(U,A)$ is 
\emph{compatible} with the product of $A$ iff
  \be
  m \circ (\id_A \oti r) \circ z_A^{} = m \circ (r \oti \id_A) 
  \labl{compatiblepair}
in $\HomC(U\oti A,A)$.
Given the algebra $A$ in \C, the \emph{full center} $Z(A)$ of $A$ is 
\cite{davy20} a pair consisting of an object in $\Z\C$ -- by abuse of notation
denoted by $Z(A)$ as well -- and a morphism $\zeta_A \iN \HomC(\FF(Z(A)),A)$
that is terminal among all pairs $((U,z),r)$ in $\Z\C$ that are compatible with 
the product of $A$. That $Z(A)$ is terminal among compatible pairs means that 
for any such pair $((U,z),r)$ there exists a unique morphism 
$\kappa \iN \Hom_{\Z\C}((U,z),Z(A))$ such that the equality
  \be
  \zeta_A \circ \FF(\kappa) = r
  \ee
holds in $\HomC(U,A)$.

For the categories \C\ relevant to us in this paper, the full center of any 
algebra in \C\ exists. Being defined by a universal property, $Z(A)$ unique 
up to unique isomorphism. Further, $Z(A)$ has a unique structure of a (unital, 
associative) algebra in $\Z\C$ such that $\zeta_A$ is an algebra morphism in 
\C, and this algebra structure is commutative \cite[Prop.\,4.1]{davy20}.
Furthermore, if $A$ and $B$ are Morita equivalent algebras in \C, then the 
algebras $Z(A)$ and $Z(B)$ in $\Z\C$ are isomorphic \cite[Cor.\,6.3]{davy20}.

\medskip

If the category \C\ is braided, then there is also a more familiar notion of 
center of an algebra inside \C\ itself, albeit there are two variants (unless 
the braiding is symmetric), the left center and the right center. The 
\emph{left center} of an algebra $A$ in \C\ is obtained with the help of an 
ana\-logue of the compatibility condition \erf{compatiblepair} in which the 
half-braiding is replaced by the braiding $c$ of \C, according to
  \be
  m \circ (\id_A \oti q) \circ c_{U,A}^{} = m \circ (q \oti \id_A) \,.
  \labl{compaleft}
Again one considers pairs of objects $U$ in \C\ together with morphisms 
$q \iN\HomC(U,A)$ obeying \erf{compaleft}, and defines the left center 
$C_l(A) \,{\equiv}\, (C_l(A),\zeta^l_A)$ to be terminal among such compatible 
pairs. The right center $C_r(A)$ is defined analogously. 
$C_l(A)$ has a unique structure of an algebra in \C\ such that the morphism
$\zeta^l_A \iN \HomC(C_l(A),A)$ is an algebra morphism. This algebra structure 
is commutative (Prop.\,2.37(i) of \cite{ffrs} and Prop.\,5.1 of \cite{davy20});
clearly, if already $A$ is commutative, then $C_l(A) \eq A \eq C_r(A)$.
If \C\ is ribbon and $A$ is Frobenius,
then $C_l(A)$ has trivial twist \cite[Lemma\,2.33]{ffrs}.

We mention the left center here because it was instrumental for the construction
by which the full center $Z(A)$ was introduced originally, for modular tensor 
categories \cite[Eq.\,(A.1)]{rffs}.
For the more general categories of our interest, there is the following variant,
which also makes use of the functor $\RR\colon \C\To\Z\C$ that is right adjoint 
to the forgetful functor $\FF$. Let us assume that \C\ is a factorizable finite
ribbon category for which $\Z\C$ is monoidally equivalent 
to the enveloping category \CbC\ (compare Remark \ref{rem:factorizable}(iv)).
Then the right adjoint functor $\RR$ exists \cite[Thm.\,3.20]{rugw}. Moreover,
$\RR$ is lax monoidal, implying that for any algebra $A$ in \C,
the object $\RR(A)\iN\Z\C$ is again an algebra. Also, the natural 
transformations $\eps\colon \FF\cir\RR \,{\Rightarrow}\,\Id_\C$ and
$\eta\colon \Id_\CbC \,{\Rightarrow}\, \RR\cir \FF$ of the adjunction are
monoidal \cite[Lemma.\,5.3]{davy20}. And further, provided that the natural 
transformation $\eps$ is epi, the full center of an algebra $A\iN\C$ 
can be expressed as ($\!$\cite[Thm.\,5.4]{davy20} and \cite[Thm.\,3.24]{rugw})
  \be
  Z(A) = C_l(\RR(A)) \qquad{\rm with}\qquad
  \zeta_A = \eps_A \,\circ\, \zeta^l_{\RR(A)} \,.
  \labl{ZA=Cl}
The functor $\RR$ can be given explicitly; once we identify $\Z\C$ with \CbC,
it acts on objects $U\iN\C$ as \cite[Eq.\,(3.43)]{rugw}
  \be
  \RR(U) = (U \boti \one) \otimes \RR(\one) \,.
  \labl{RRU=}
Moreover, the algebra $\RR(\one)$ in $\Z\C$ is commutative
\cite[Lemma\,3.25]{rugw} and hence equals $Z(\one)$. Together with the
formulas \erf{Z(1)} and \erf{coendFo} for the bulk state space \Fo\
this shows that $\RR(\one)$ can be obtained as a  coend,
  \be
  \RR(\one) \,\cong \coendF U \,.
  \labl{RR1=coendF}
Hereby for modular tensor categories, for which \Fo\ is given by the finite 
direct sum \erf{Fo}, the formula \erf{ZA=Cl} for the full center reduces to 
(A.1) of \cite{rffs}.
As will be described elsewhere, an isomorphism between the objects
on the left and right hand sides of \erf{RR1=coendF} in fact exists
for any factorizable finite tensor category \C.


\subsection{Algebras and characters} \label{algchar}

Let $A \,{=}\, (A,m,\eta)$ be a (unital, associative) \findim\ algebra over a 
field \ko, and let $M \,{=}\, (M,\rho)$ be a \findim\ left $A$-module. The 
\emph{character} 
$\chii_M^A$ of the module $M$ is defined to be the partial trace of the \rep\
 morphism $\rho$, with the trace taken in the sense of linear maps. This means
  \be
  \chii_M^A = \mathrm{tr}_M(\rho) = \tilde d^{\,\ko}_M \circ
  (\rho \oti \id_{M^\vee}) \circ (\id_A\oti b^{\,\ko}_M) ~\in \Hom(A,\ko) \,,
  \labl{chiiMA}
with $b^{\,\ko}_M$ the (right) coevaluation and $\tilde d^{\,\ko}_M$ the (left) 
evaluation map of \Vectk. Now the map $\tilde d^{\,\ko}_M \,{\in} 
         $\linebreak[0]$
\Homk(M \otik M^*,\ko)$ can be expressed 
through the right evaluation map $d^{\,\ko}_M \,{\in}\, \Homk(M^* \otik M,\ko)$ 
as $\tilde d^{\,\ko}_M \eq d^{\,\ko}_M \cir \tau_{M,M^*_{\phantom|}}$ with 
$\tau$ the flip map (and similarly for the left and right coevaluations). Thus
two equivalent descriptions of the character are, pictorially,
  \eqpic{def_char} {200} {27} { \put(0,3){
  \put(0,36)    {$ \chii_M^A ~= $}
  \put(45,3){ \begin{picture}(0,0)(0,0)
      \scalebox{.38}{\includegraphics{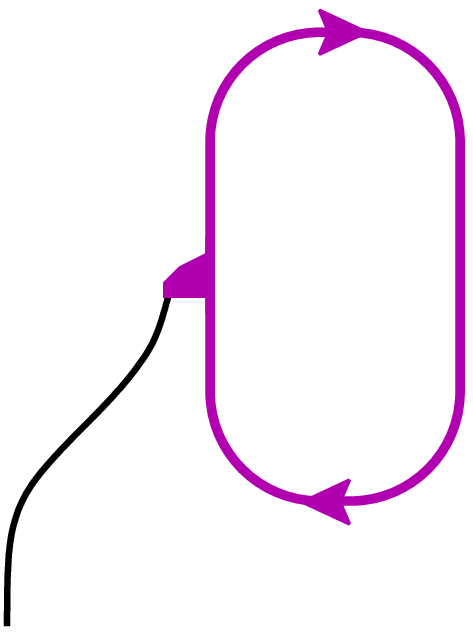}}\end{picture}
  \put(-3.9,-9) {\sse$ A $}
  \put(14.6,59) {\sse$ M $}
  }
  \put(125,36)  {$ = $}
  \put(150,3){ \begin{picture}(0,0)(0,0)
      \scalebox{.38}{\includegraphics{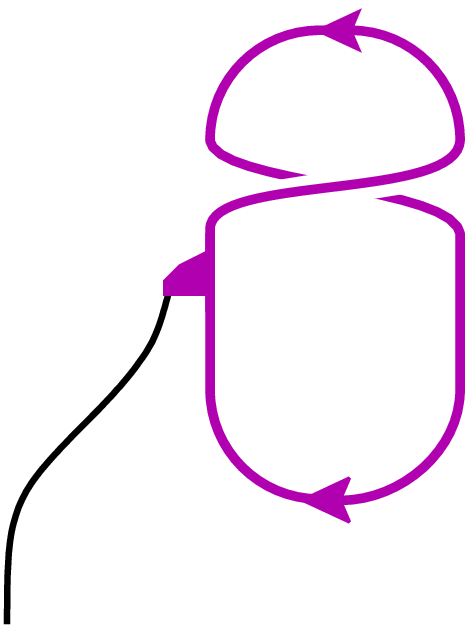}}\end{picture}
  \put(-3.9,-9) {\sse$ A $}
  \put(51,60)   {\sse$ M $}
  } } }

Characters are class functions, i.e.\ satisfy
$\chii_M^A \,{\circ}\, m \,{=}\, \chii_M^A \,{\circ}\, m\,{\circ}\, \tau_{A,A}$.
$A$ is semisimple iff the space of class functions is already exhausted by 
linear combinations of characters of $A$-modules \cite[Cor.\,2.3]{coWe6}.
Furthermore, the characters of the non-isomorphic simple $A$-modules are
linearly independent over \ko\ \cite[Thm.\,1.6(b)]{loren}, and
characters behave additively under short exact sequences. As a 
consequence, taking $ \{ S_i \,|\, i\,{\in}\,\I \} $ to be a full set of 
representatives of the isomorphism classes of simple $A$-mo\-du\-les, with 
characters $\chii_i^A \,{\equiv}\, \chii_{S_i^{}}^A$, and writing 
$[\,M\,{:}\, S_i\,]$ for the multiplicity of $S_i$ in the Jordan-H\"older 
series of $M$, one has
  \be
  \chii_M^A = \sum_{i\in\I}\, [\,M \,{:}\, S_i\,] \, \chii_i^A \,.
  \label{chiAV}
  \ee

The simple modules $S_i$ have projective covers $P_i$, from which they can be
recovered as the quotients $S_i \eq P_i/J(A)\,P_i$ with $J(A)$ the Jacobson 
radical of $A$. The modules $ \{ P_i \,|\, i\,{\in}\,\I \} $ constitute a full 
set of representatives of the isomorphism classes of indecomposable projective 
left $A$-modules. There is a (non-unique) collection $ \{ e_i \,{\in}\, A 
\,|\, i\,{\in}\,\I \} $ of primitive orthogonal idempotents such that 
$P_j \,{=}\, A\,e_j$ for all $j\iN\I$, as well as $Q_j \,{=}\, e_j\,A$ for 
a full set of representatives of the isomorphism classes of indecomposable 
projective right $A$-modules.
The algebra $A$ decomposes as a left module over itself (with the
\emph{regular} action, given by the product $m$) as
  \be
  _AA \,\cong\, \bigoplus_{i\in\I}\, P_i \otimes_\ko \ko^{\dim(S_i)}_{} .
  \labl{AA=}

\smallskip

Of particular interest to us is $A$ regarded as a \emph{bi}module over itself,
with regular left and right actions. The decomposition of this bimodule into 
indecomposables is considerably more involved than the decomposition as a module
and cannot be expressed in a `model-independent' manner analogous to \erf{AA=}. 
But we can use that the structure of an $A$-bimodule is equivalent to the one of
a left $A{\otimes}A\op$-module. Accordingly, by the \emph{character of $A$ as an
$A$-bimodule} we mean its character as an $A{\otimes}A\op$-module.

Now if \ko\ has characteristic zero, then for any two \findim\ \ko-algebras 
$A$ and $B$ a complete set of simple modules over the tensor product algebra 
$A{\otimes}B$ is \cite[Thm.\,(10.38)]{CUre1} given by 
$\{S_i^A \otik S_j^B \,|\, i\,{\in}\,\I_A\,,\, j\,{\in}\,\I_B\}$.
In view of \eqref{chiAV}, the character of any $A{\otimes}B$-module $X$ can 
therefore be written as the bilinear combination
  \be
  \chii_X^{A\otimes B} = \sum_{i\in\I_A,\,j\in\I_B}\!
  [\,X \,{:}\, S_i^A \otik S_j^B\,] \,\, \chii_i^A \otik \chii_j^B \,.
  \label{chiABX}
  \ee
For the case of our interest, i.e.\ $B \,{=}\, A\op$ and $X \,{=}\, A$ ,
this decomposition reads
  \be 
  \chii_A^{A\otimes A\op_{}} = \sum_{i,j\in\I}\, [\,A \,{:}\,S_i\otik T_j\,]
  \, \chii_{S_i\otimes_\ko^{}T_j}^{A\otimes A\op_{}}
  \labl{A-bimod-char}
with $T_k \eq Q_k/J(A)\,Q_k$ the simple quotients of the projective right 
$A$-modules $Q_k$.

Next we use that (for details see \cite[App.\,A]{fuSs4})
  \be
  [\,A \,{:}\,S_i\otik T_j\,] = \dimk\big( \HomA(P_i,P_j) \big) 
  = c_{i,j} \,,
  \labl{ASiTj=cij}
where the non-negative integers $c_{i,j}$ are defined by
  \be
  c_{i,j} := [\,P_i \,{:}\, S_j \,] \,.
  \labl{cij}
The matrix $C \,{=}\, \big(\,c_{i,j}\,\big)$ is called the 
\emph{Cartan matrix} of the algebra $A$, or of the category $A\Mod$.
It obviously depends only on $A\Mod$ as an abelian category. 

Assume now that $A$ is \emph{self-injective}, i.e.\ injective as a left module 
over itself. Then $T_k^{} \,{\cong}\, S_k^*$ as right $A$-modules, so that in 
view of \erf{ASiTj=cij} we can rewrite the character \erf{A-bimod-char} as
  \be
  \chii_A^{A\otimes A\op_{}}
  = \sum_{i,j\in\I} c_{i,j} \, \chii_i^A \oti \chii_j^A \,.
  \labl{A-bimod-char-2}


\subsection{Factorizable Hopf algebras} \label{facHopf}

In this paper we deal with \findim\ Hopf algebras over a field \ko\ that is
algebraically closed and has characteristic zero. In the application to
logarithmic CFT, \ko\ is the field \complex\ of complex numbers. We denote by 
$m \iN \Homk(H\otik H,H)$ the product, by $\eta \iN \Homk(\ko,H)$ the unit, 
by $\Delta \iN \Homk(H,H\otik H)$ the coproduct, by $\eps \iN \Homk(H,\ko)$ 
the counit, and by $\apo \iN \Homk(H,H)$ the antipode of $H$. 

An \emph{R-matrix} for $H$ is an invertible element $R$ of $H\otik H$
which intertwines the coproduct and opposite coproduct in the sense that
  \be
  R\, \Delta\, R^{-1} = \tauHH \cir \Delta \equiv \Delta^{\!\rm op}_{} 
  \ee
and which satisfies the equalities
  \be
  (\Delta \oti \id_H) \circ R = R_{13}\cdot R_{23} \qquand
  (\id_H \oti \Delta) \circ R = R_{13}\cdot R_{12}
  \labl{deqf-qt}
in $H\otik H\otik H$. (The notation $R_{13}$ means that $R$ is to be
considered as an element in the tensor product of the first and third
factors of $H\otik H\otik H$, and similarly for $R_{23}$ etc.)
A Hopf algebra $(H,m,\eta,\Delta,\eps,\apo)$ together with an R-matrix
$R$ is called a \emph{quasitriangular Hopf algebra}.
For more information about quasitriangular Hopf algebras see e.g.\ Chapters 1 
and 2 of \cite{MAji}.

For a quasitriangular Hopf algebra, the invertible element
  \be
  Q := R_{21}\,{\cdot}\, R
  \ee
of $H\otik H$ is called the \emph{monodromy matrix}.
A quasitriangular Hopf algebra for which $Q$ is non-de\-generate,
meaning that it can be expressed as $\sum_\ell h_\ell \oti k_\ell$,
in terms of two vector space bases $\{h_\ell\}$ and $\{k_\ell\}$ of $H$,
is called \emph{factorizable}.
Equivalently, factorizability means that the \emph{Drinfeld map}
   \be
   f_Q := (d_H\oti \id_H) \circ (\idHs\oti Q) ~\in \Hom(\Hs,H)
   \labl{def-drin}
is invertible.
With a view towards the non-quasitriangular Hopf algebras considered for 
logarithmic conformal field theories in e.g.\ \cite{fgst}, one should
note that for the notion of factorizability we only need the existence of
a monodromy matrix $Q$, but not of an R-matrix; moreover, the properties
of $Q$ can be formulated without any reference to $R$ \cite[Sect.\,2]{brug6}.

A factorizable Hopf algebra is minimal in the sense that it does not 
contain a proper quasitriangular Hopf subalgebra \cite[Prop.\,3b]{radf13},
and is thus \cite{radf11}
a quotient of the Drinfeld double $D(B)$ of some Hopf algebra $B$.

A \emph{ribbon element} for a quasitriangular Hopf algebra $H$
is an invertible element $v$ of the center of $H$ that obeys
   \be
   \apo \circ v = v \,, \qquad \eps \circ v = 1 \qquand
   \Delta \circ v = (v\oti v) \cdot Q^{-1} .
   \labl{def-ribbon}
A quasitriangular Hopf algebra together with a ribbon element is called a 
\emph{ribbon Hopf algebra}.

\medskip

By a slight abuse of terminology, for brevity we refer in this paper to 
a \findim\ factorizable ribbon Hopf algebra over \ko\ just as a
\emph{factorizable Hopf algebra}.
There exist plenty of such algebras. For example, the Drinfeld double of a 
\findim\ Hopf algebra $B$ is factorizable provided that the square of the 
antipode of $B$ obeys a certain condition \cite[Thm.\,3]{kaRad}. This includes 
e.g.\ the Drinfeld doubles of finite groups (for which explicit formulas for 
the morphisms $\Corro gpq$ \erf{Sk_morph} can be found in \cite{ffss}).
Another large class $\{ U_{(N,\nu,\omega)} \}$ (with $N\,{>}\,1$ an odd 
integer, $\omega$ a primitive $N$th root of unity, and $\nu\,{<}\,N$ a 
positive integer such that $N$ does not divide $\nu^2$) of factorizable Hopf 
algebras is described in \cite[Sect.\,5.2]{radf13}. $U_{(N,\nu,\omega)}$ has
dimension $N^3/(N,\nu^2)^2$ \cite[Prop.\,10b)]{radf13}, and this family
comprises the small quantum group that is a \findim\ quotient of 
$U_q(\mathrm{sl}(2))$ as the special case $\nu \eq 2$ with $q \eq \omega^{-2}$,
compare \cite[p.\,260)]{radf13} and \cite[Prop.\,4.6]{lyma}.

\medskip

We also need the notions of integrals and cointegrals for Hopf algebras.
A \emph{left integral} of $H$ is an element $\Lambda \iN H$ obeying
  \be
  m \circ (\id_H \oti \Lambda) = \Lambda \circ \eps \,,
  \ee
or, in other words, a morphism of left $H$-modules from the trivial $H$-module 
$(\ko,\eps)$ to the regular $H$-module $(H,m)$. Dually,
a \emph{right cointegral} of $H$ is an element $\lambda \iN \Hs$ that satisfies
  \be
  (\lambda \oti \id_H) \circ \Delta = \eta \circ \lambda \,.
  \ee
Right integrals and left cointegrals are defined analogously.

For a \findim\ Hopf algebra there is, up to normalization, a unique non-zero 
left integral $\Lambda$ and a unique non-zero right cointegral $\lambda$, and
the number $\lambda\cir\Lambda \iN \ko$ is invertible. 
Also, the antipode $\apo$ of $H$ is invertible. If $H$ is quasitriangular,
then the square of the antipode is an inner automorphism, acting as
$h \,{\mapsto}\, u^{-1}\,h\,u$ with $u\iN H$ the \emph{Drinfeld element}
  \be
  u := m \circ (\apo\oti\id_H) \circ R_{21} \,.
  \labl{def:uDrinfeld}
And if $H$ is factorizable, then it is \emph{unimodular} \cite[Prop.\,3c)]{radf13},
i.e.\ the left integral $\Lambda$ is also a right integral, which implies that 
$\apo\cir\Lambda \eq \Lambda$. Moreover, $f_Q(\lambda)$ is an integral, too, 
and thus is a non-zero multiple of $\Lambda$.  One may then fix the 
normalizations of the integral and cointegral in such a way that
  \be
  \lambda \circ \Lambda = 1  \qquand  f_Q(\lambda) = \Lambda \,.
  \labl{Lambdalambda}
This convention is adopted throughout this paper; it determines $\Lambda$ and 
$\lambda$ uniquely except for a common sign factor.


\subsection{Representations of mapping class groups} \label{app:mpg}

As has been established in \cite{lyub11}, the spaces \erf{VE} of chiral blocks
come with natural representations of mapping class groups for surfaces with 
holes (that is, with open disks excised). To describe these or, rather, the
representations on spaces of blocks with outgoing instead of incoming
field insertions (see formula \erf{Vpqg} below), it is convenient to
present these groups through generators (and relations,
but these are irrelevant for us, as we are interested in invariants). A
suitable set of generators of \Mapgn, the mapping class group of genus-$g$ 
surfaces with $n$ holes, is obtained by noticing the exact sequence
  \be
  1 \,\longrightarrow \,\BG gn \,\longrightarrow \,\Mapgn \,\longrightarrow\,
  \Map g0 \,\longrightarrow \, 1
  \ee
of groups, where $\BG gn$ is a central extension of the surface braid group by 
$\zet^n$ (compare Theorem 9.1 of \cite{FAma}). As a set of generators one may 
thus take the union of those for some known presentations of \Map g0\ 
\cite{wajn} and of $\BG gn$ \cite{scotG3}. This amounts to the following 
(non-minimal) system of generators \cite{lyub6,lyub11}:

\def\leftmargini{1.49em}~\\[-2.49em]\begin{enumerate}\addtolength{\itemsep}{-4pt}
    \item
Braidings which interchange neighboring boundary circles.
    \item
Dehn twists about boundary circles.
    \item
Homeomorphisms $S_l$, for $l \eq 1,2,...\,,g$, which act as the identity outside
a certain region $\mathcal T_l$ and as a modular S-transformation in a slightly 
smaller region $\mathcal T_l' \,{\subset}\, \mathcal T_l^{}$ that has the 
topology of a one-holed torus. 
\\
(For the relevant regions $\mathcal T_l^{}$ and $\mathcal T_l'$, as well as the
cycles appearing in the subsequent entries of the list, see the picture below.)
    \item
Dehn twists in tubular neighborhoods of certain cycles $a_m$
and $e_m$, for $m \eq 2,3,$ $...\,,g$.
    \item
Dehn twists in tubular neighborhoods of certain cycles $b_m$ and $d_m$,
for $m \eq 1,2,$ $...\,,g$.
    \item
Dehn twists in tubular neighborhoods of certain cycles $t_{j,m}$, for
$j\eq 1,2,...\,,n{-}1$ and 
  \\
$m \eq 1,2,...\,,g$.
\end{enumerate}

\noindent
In particular, for the torus without holes ($g \eq 1$ and $n \eq 0$), the 
generators $S \eq S_1$ and $T \eq d_1$ furnish the familiar S- and 
T-transformations which generate the modular group \slz.
The regions $\mathcal T_l$ and $\mathcal T_l'$ and cycles $a_m$, $b_m$, $e_m$, 
$d_m$, and $t_{j,m}$ are exhibited in the following picture:
  \Eqpic{surf_gn_PIC} {420} {45} { \put(35,0){ \setlength\unitlength{1.3pt}
    \put(-42,0)  {\INcludepichtft{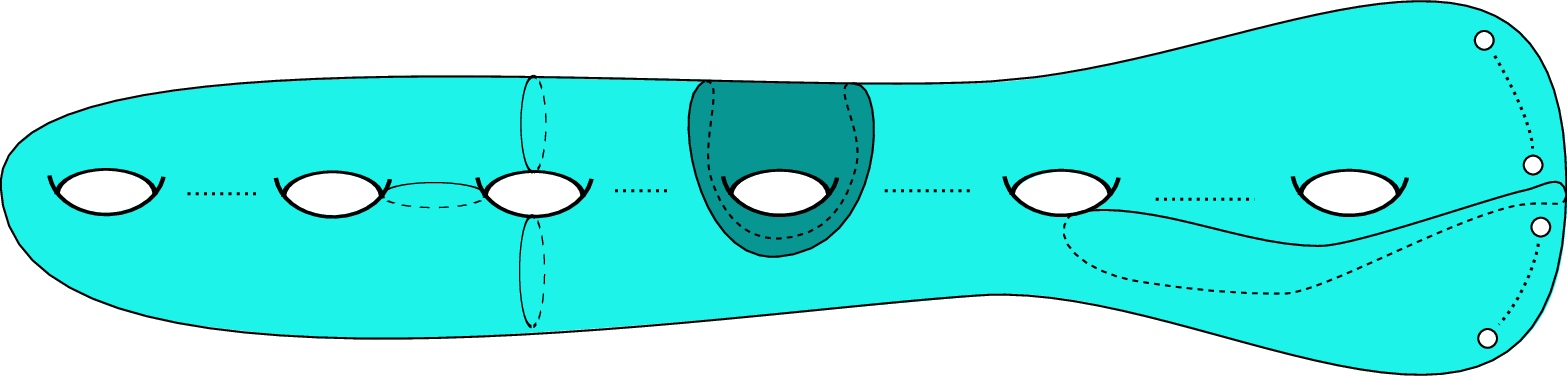}{65}}
  \put(-23,38)     {\sse$ b_1$}
  \put(23,38.5)    {\sse$ b_{m{-}1}$}
  \put(49,46)      {\sse$ a_m$}
  \put(71,38)      {\sse$ b_m$}
  \put(61,54)      {\sse$ d_m$}
  \put(61,22)      {\sse$ e_m$}
  \put(125,54)     {\sse$ S_l$}
  \put(127,38)     {\sse$ b_{l}$}
  \put(188,38)     {\sse$ b_{k}$}
  \put(253,39)     {\sse$ b_{g}$}
  \put(223,26)     {\sse$ t_{j,k}$}
  \put(291,80)     {\sse$ U_1^{} $}
  \put(306,44)     {\sse$ U_j$}
  \put(306,31)     {\sse$ U_{j+1}$}
  \put(291,-1)     {\sse$ U_n^{} $}
  } }
$\mathcal T_l$ is the shaded region in \eqref{surf_gn_PIC}; it
is a one-holed torus forming a neighborhood of the $l$th handle, while 
$\mathcal T_l'$ is the smaller region indicated by the dotted line inside 
$\mathcal T_l$.

Also shown in \eqref{surf_gn_PIC} are decorations of the boundary circles
by objects $U_1, U_2, ...\,, U_n$ of a factorizable finite tensor category \C. 
The representation of \Mapgn\ constructed in \cite{lyub11} acts on the space
  \be
  V^{U}_{g:n} := \HomC(L^{\otimes g},U)
  \labl{Vpqg}
of morphisms of \C, where
  \be
  U := \bigoplus_{\sigma\in \mathfrak S_n}\,
  U_{\sigma(1)} \oti U_{\sigma(2)} \oti \cdots \oti U_{\sigma(n)} \,,
  \labl{defX}
with $L$ the handle Hopf algebra \erf{defLK};
the summation in \erf{defX} can be restricted to the subgroup
$\mathfrak N \eq \mathfrak N(U_1{,}...{,} U_n)$ of the symmetric group
$\mathfrak S_n$ that is
generated by those permutations $\sigma$ for which for at
least one value of $i$ the objects $U_i$ and $U_{\sigma(i)}$ are non-isomorphic.
In this representation $\pi^U_{g,n}$ of \Mapgn\ the different types of generators 
described in the list above act on the space \erf{Vpqg} as follows (see
\cite[Prop.\,2.4]{fuSs5}):

\def\leftmargini{1.49em}~\\[-2.49em]\begin{enumerate}\addtolength{\itemsep}{-4pt}
    \item
Post-composition with a braiding morphism which interchanges the objects that
label neighboring field insertions.
    \item
Post-composition with a twist isomorphism of the object labeling a field 
insertion.
    \item
Pre-composition with an isomorphism $\id_{L^{\otimes g-l}_{}}
\oti S^L \oti\id_{L^{\otimes l-1}} \iN \EndC(L^{\otimes g})$,
for $l \eq 1,2,...\,,g$.
    \item
Pre-composition with an isomorphism $\id_{L^{\otimes g-m}_{}}
\oti [ \OL\cir (T^L \oti T^L)] \oti \id_{L^{\otimes m-2}_{}}$, respectively 
$\id_{L^{\otimes g-m}_{}} \oti \big[ (T^L \oti \theta_{L^{\otimes m-1}})
\cir \QL_{L^{\otimes m-1}_{}} \big]$, for $m \eq 2,3,$ $...\,,g$.
    \item
Pre-composition with an isomorphism $\id_{L^{\otimes g-m}_{}}
\oti (S^{L^{\scriptstyle -1}}\cir T^L\cir S^L) \oti \id_{L^{\otimes m-1}_{}}$, 
respectively $\id_{L^{\otimes g-m}_{}} \oti T^L \oti \id_{L^{\otimes m-1}}$,
for $m \eq 1,2,$ $...\,,g$.
    \item
The map that sends $f\iN \HomC(L^{\otimes g},U_1\oti\cdots\oti U_n)
\,{\subseteq}\,V^{U}_{g:n}$ to
  \be
  \bearl
  \big(\, \big[\, ( \id_{U_1\otimes\cdots\otimes U_j}
  \oti \tilde d_{U_{j+1}\otimes\cdots\otimes U_n} ) \,\circ\, 
 ( f \oti \id_{\Vee U_n\otimes\cdots\otimes \Vee U_{j+1}} )
  \\{}\\[-.4em] \hspace*{2.5em}
  \circ\, \{ \id_{L^{\otimes g-m}_{}} \oti [
  \QL_{L^{\otimes m-1}\otimes \Vee U_n\otimes\cdots\otimes \Vee U_{j+1}} \cir
  (T^L \oti \theta_{L^{\otimes m-1}\otimes \Vee U_n\otimes\cdots\otimes
  \Vee U_{j+1}}) ] \} \,\big]
  \\{}\\[-.6em] \hspace*{15.5em}
  \otimes\, \id_{U_{j+1}\otimes\cdots\otimes U_n} \,\big) \,\circ\, \big(
  \id_{L^{\otimes g}_{}} \oti \tilde b_{U_{j+1}\otimes\cdots\otimes U_n} \big) 
  \eear
  \labl{LyubactC3}
and acts analogously on the other direct summands $\HomC(L^{\otimes g},U_{\sigma(1)} 
\oti U_{\sigma(2)} {\otimes}\, \cdots \oti U_{\sigma(n)})$ of $V^{U}_{g:n}$,
with $\sigma\iN \mathfrak N$, for $j\eq 1,2,...\,,n{-}1$ and $m \eq 1,2,...\,,g$.
\end{enumerate}

\noindent
Here we have introduced the abbreviations $S^L$, $T^L$, $\OL$, and $\QL_W$ for 
$W\iN\C$, for specific morphisms of \C\ involving tensor powers of $L$.
These morphisms are defined, with the help of dinatural families, by
   \Eqpic{p9-and-more} {420}{91} {
        \put(0,135) {
   \put(0,28)      {$ T^L \,\circ\, \iL_U ~:= $}
   \put(90,0)  {\Includepichtft{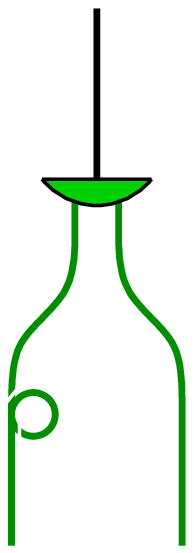}
   \put(-11.9,14)  {\sse$ \theta_{\!U^{\!\vee}_{}}^{}$}
   \put(-3,-9.2)   {\sse$ U^{\!\vee} $}
   \put(8.4,63.3)  {\sse$ L $}
   \put(18.9,37.7) {\sse$ \iL_U $}
   \put(17.4,-9.2) {\sse$ U $}
   } } 
   \put(0,50)     {$ \OL \,\circ\, (\iL_U \oti \iL_V) ~:= $}
   \put(114,0) { \Includepichtft{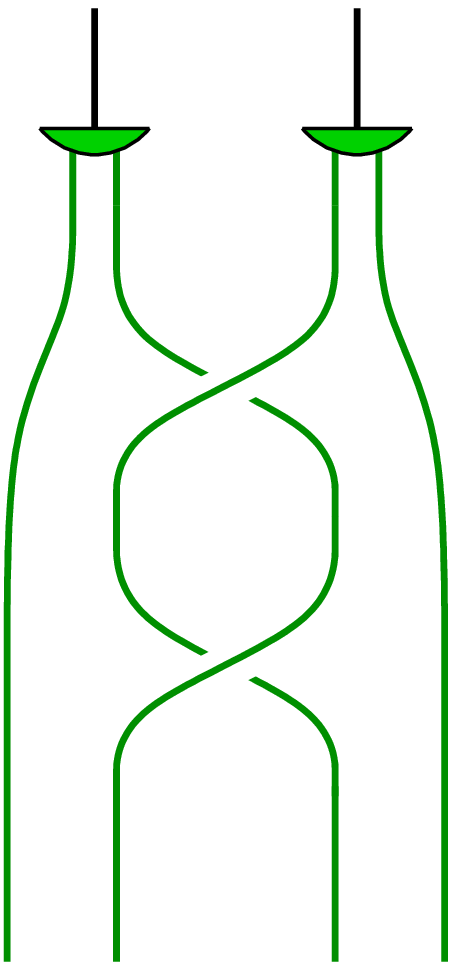}
   \put(-6,-9.2)  {\sse$ U^{\!\vee} $}
   \put(-5.8,89)  {\sse$ \iL_U $}
   \put(8,-9.2)   {\sse$ U $}
   \put(6.1,109)  {\sse$ L $}
   \put(22.8,27.3){\sse$ c $}
   \put(22.8,58.4){\sse$ c $}
   \put(32,-9.2)  {\sse$ V^{\!\vee} $}
   \put(36,109)   {\sse$ L $}
   \put(46.2,89)  {\sse$ \iL_V $}
   \put(46,-9.2)  {\sse$ V $}
   } 
   \put(225,50)   {$ \QL_W \,\circ\, (\iL_U \oti \id_W) ~:= $}
   \put(350,0) { \begin{picture}(0,0)(0,0)
         \scalebox{.38}{\includegraphics{imgs/pic_htft_103b.eps}}\end{picture}
   \put(-5.6,-9.2){\sse$ U^{\!\vee} $}
   \put(7.4,109)  {\sse$ L $}
   \put(9.4,-9.2)   {\sse$ U $}
   \put(22.2,27.3){\sse$ c $}
   \put(22.2,58.4){\sse$ c $}
   \put(34,-9.2)  {\sse$ W $}
   \put(34.3,109) {\sse$ W $}
   \put(17.8,90)  {\sse$ \iL_U $}
   } }
while
   \be
   S^L := (\eps_L \oti \id_L) \circ \OL \circ (\id_L \oti \Lambda_L) \,.
   \labl{S-HK}

\smallskip

For applications in CFT,
we also need to generalize the prescriptions above to the situation that
there are both outgoing and incoming field insertions. To treat this case,
we must partition the set of boundary circles into two subsets having, say, 
$p$ and $q$ elements. Denoting the objects labeling the corresponding 
insertions by $U_1, U_2, ...\,, U_p$ and by $W_1, W_2, ...\,, W_q$,
respectively, we can define objects $U$ and $W$ analogously as in \erf{defX} 
and consider the linear isomorphism
  \be
  \varphi:\quad \HomC(L^{\otimes g}{\otimes}\,W,U)
  \,\stackrel\cong\longrightarrow\, \HomC(L^{\otimes g},U\oti W^\vee) 
  \ee
that is supplied by the right duality of \C. Then by setting
  \be
  \pi^{W,U}_{g,p,q}(\gamma)
   := \varphi^{-1} \circ \pi^{U\otimes W^\vee}_{g,p+q}(\gamma) \circ \varphi
  \labl{piXYgpq}
for $\gamma \iN \Mapgppq$ we obtain a \rep\ of the subgroup \Mapgpq\ of the 
mapping class group \Mapgppq\ that leaves each subset of circles separately 
invariant, on the space $\HomC(L^{\otimes g}{\otimes}\,W,U)$.

Also, in the application to correlation functions of bulk fields in full CFT, 
we deal with the category
\CbC\ instead of \C, and accordingly with the \emph{bulk} handle Hopf algebra
$K$ instead of $L$. Then in particular for $\C \,{\simeq}\, \HMod$ the 
S- and T-transformations result in the pictures \erf{S_KH-TKH} for $S^K$ and
$T^K$ presented in the main text. For further details we refer to
\cite{fuSs5,stig7a}, e.g.\ $\mathcal O^K$ for $\C \,{\simeq}\, \HMod$ is
given by formula (4.2) of \cite{fuSs5}.


   \vfill

\noindent{\sc Acknowledgments:}
This work has been supported in part by the ESF Research Networking Programme
``Interactions of Low-Dimensional Topology and Geometry with Mathematical
Physics''.
CSc is partially supported by the Collaborative Research Centre 676 ``Particles,
Strings and the Early Universe - the Structure of Matter and Space-Time'' and
by the DFG Priority Programme 1388 ``Representation Theory''.

\newpage


  \newcommand\wb{\,\linebreak[0]} \def\wB {$\,$\wb}
  \newcommand\Bi[2]    {\bibitem[#2]{#1}}
  \newcommand\inBo[8]  {{\em #8}, in:\ {\em #1}, {#2}\ ({#3}, {#4} {#5}),
                         p.\ {#6--#7}}
  \newcommand\inBO[9]  {{\em #9}, in:\ {\em #1}, {#2}\ ({#3}, {#4} {#5}),
                         p.\ {#6--#7} {{\tt [#8]}}}
  \renewcommand\J[7]   {{\em #7}, {#1} {#2} ({#3}) {#4--#5} {{\tt [#6]}}}
  \newcommand\JO[6]    {{\em #6}, {#1} {#2} ({#3}) {#4--#5} }
  \newcommand\BOOK[4]  {{\em #1\/} ({#2}, {#3} {#4})}
  \newcommand\prep[2]  {{\em #2}, preprint {\tt #1}}
  \newcommand\phd[3]   {{\em #3}, {PhD thesis, #1} {{\tt [#2]}}}
  \def\adma  {Adv.\wb Math.}
  \def\aspm  {Adv.\wb Stu\-dies\wB in\wB Pure\wB Math.}
  \def\atmp  {Adv.\wb Theor.\wb Math.\wb Phys.}   
  \def\bacp  {Ba\-nach\wB Cen\-ter\wB Publ.}
  \def\coma  {Con\-temp.\wb Math.}
  \def\comp  {Com\-mun.\wb Math.\wb Phys.}
  \def\cpma  {Com\-pos.\wb Math.}
  \def\duke  {Duke\wB Math.\wb J.}
  \def\imrn  {Int.\wb Math.\wb Res.\wb Notices}
  \def\isjm  {Israel\wB J.\wb Math.}
  \def\jhep  {J.\wb High\wB Energy\wB Phys.}
  \def\joal  {J.\wB Al\-ge\-bra}
  \def\jktr  {J.\wB Knot\wB Theory\wB and\wB its\wB Ramif.}
  \def\jopa  {J.\wb Phys.\ A}
  \def\jomp  {J.\wb Math.\wb Phys.}
  \def\jpaa  {J.\wB Pure\wB Appl.\wb Alg.}
  \def\jram  {J.\wB rei\-ne\wB an\-gew.\wb Math.}
  \def\momj  {Mos\-cow\wB Math.\wb J.}
  \def\npbp  {Nucl.\wb Phys.\ B (Proc.\wb Suppl.)}
  \def\nupb  {Nucl.\wb Phys.\ B}
  \def\pams  {Proc.\wb Amer.\wb Math.\wb Soc.}
  \def\phlb  {Phys.\wb Lett.\ B}
  \def\phrl  {Phys.\wb Rev.\wb Lett.}
  \def\plms  {Proc.\wB Lon\-don\wB Math.\wb Soc.}
  \def\prja  {Proc.\wB Japan\wB Acad.}
  \def\pcps  {Proc.\wB Cam\-bridge\wB Philos.\wb Soc.}
  \def\ruma  {Revista de la Uni\'on Matem\'atica Argentina}
  \def\slnm  {Sprin\-ger\wB Lecture\wB Notes\wB in\wB Mathematics}
  \def\taac  {Theo\-ry\wB and\wB Appl.\wb Cat.}
  \def\taia  {Top\-o\-lo\-gy\wB and\wB its\wB Appl.}
  \def\tams  {Trans.\wb Amer.\wb Math.\wb Soc.}
  \def\thmp  {Theor.\wb Math.\wb Phys.}
  \def\trgr  {Trans\-form.\wB Groups}

  \end{document}